\documentclass[11pt]{amsart}

\usepackage{template}

\begin{document}
\title{Near-Optimal Generalized Private Testing}

\author{Anamay Chaturvedi\textsuperscript{1}}
\author{Monika Henzinger\textsuperscript{1}}
\author{Jalaj Upadhyay\textsuperscript{2}}

\thanks{\textsuperscript{1} Institute of Science and Technology Austria (ISTA), Klosterneuburg, Austria.  \\\phantom{\textsuperscript{1} }
\hspace{1em}\texttt{\{\href{mailto:anamay.chaturvedi@ist.ac.at}{anamay.chaturvedi}, \href{mailto:monika.henzinger@ist.ac.at}{monika.henzinger}\}@ist.ac.at.}}
\thanks{\textsuperscript{2} Rutgers University, USA. \href{mailto:jalaj.upadhyay@rutgers.edu}{\texttt{jalaj.upadhyay@rutgers.edu}}.}

\begin{abstract}
One of the most well-studied problems in differential privacy (DP) is private threshold testing, in which there is a sequence of real values $(v_t)_{t\geq1}$, and the goal is to identify the first value which exceeds a given threshold $\tau$. The \emph{sparse vector technique} is a privacy accounting technique that addresses this problem under the promise that the values have bounded Lipschitz sensitivity, and it is ubiquitous in the DP literature. In this paper, we introduce new algorithms, lower bounds and applications for the generalized version of this problem.

The generalized private testing problem was introduced by Liu and Talwar (STOC 2019). There is a private dataset $X$, and a sequence of black-box $\varepsilon_t$-DP mechanisms $\mathcal{M}_t$ taking values in $\{+1,-1\}$. The goal is to identify the first mechanism whose \emph{success probability} $p_t = \Pr[\mathcal{M}_t(X) = +1]$ exceeds a given \emph{acceptance threshold} $p^*\in(0,1)$, while achieving an $\varepsilon$-DP guarantee. In other words, the Lipschitz requirement on real-valued inputs in the standard setting is replaced by a stability guarantee in terms of DP itself. The accuracy of a solution for this problem can be characterized by the gap between the acceptance threshold $p^*$, and a corresponding \emph{rejection} threshold $\bar{p}$, such that with \allowbreak probability $1-\beta$ for all $t\geq1$, if $p_t \leq \bar{p}$, then the mechanism $\mathcal{M}_t$ is rejected, and if $p_t \ge p^*$, it is accepted. 

In this paper, we introduce a new DP mechanism called the Generalized Thresholding Mechanism (GTM). For any $\varepsilon>0$ and any sequence of \emph{approximate} $(\varepsilon_t,\delta_t)$-DP input mechanisms $\mathcal{M}_t$, the GTM achieves a pure $\varepsilon$-DP guarantee. For any $\theta>0$, $\gamma\in(1,2]$, and $\beta\in(0,1)$, it achieves a rejection threshold of $\max(p^*/\gamma\Lambda_t, 1 - \gamma\Lambda_t(1-p^*)) - \delta_t/\varepsilon_t$ for $\Lambda_t = (5t\ln^3(t+2))^{(2+\theta) \varepsilon_t/\varepsilon} (4/\beta)^{(3 + \theta + 2/\theta )\varepsilon_t /\varepsilon}$. With probability $1-\beta$, the number of evaluations of $\mathcal{M}_t$ is at most $O\!\left((\ln(t/\beta)/(\gamma-1)^2) \cdot \max\!\left(\Lambda_t/p^*,\; (1-p^*)^{-1}\right)\right)$ for all $t\geq 1$. We also prove lower bounds which show near-optimality of our accuracy and sample complexity guarantees.

We illustrate the power of the GTM by giving a \emph{black-box} reduction for private optimization from the continual observation (CO) setting to the batch setting. This gives us the first algorithms for a large class of maximization problems in the CO setting. Our mechanism permits an adaptive choice of acceptance thresholds $(p^*_t)_{t\geq1}$  (replacing a uniform choice of $p^*$ with an adaptive choice of $p^*_t$). This functionality addresses in part a challenge mentioned in prior work on using generalized private testing for hyperparameter optimization (Papernot and Steinke (ICLR 2022)).  
\end{abstract}
\pagenumbering{roman}
\maketitle
\clearpage
\pagenumbering{arabic}

\section{Introduction}

Two fundamental problems in the study of differential privacy (DP) are {\em private selection} and {\em private testing}. The private selection problem asks the analyst to pick the best item from a given ground set based on a {\em score function}. The private testing problem asks the analyst to pick the first item from a given stream whose score exceeds a given target score. These problems are well understood when the function is known to have bounded Lipschitz sensitivity, and their algorithmic solutions, the \textit{exponential mechanism} \cite{mcsherry2007mechanism} and the \textit{above threshold mechanism} \allowbreak\cite{DBLP:conf/stoc/DworkNRRV09}, are some of the most widely used mechanisms in the design of DP mechanisms.

A fundamental limitation to what can be achieved for private selection and testing for worst-case inputs is the requirement that the score functions have known bounded global Lipschitz sensitivity over the dataset space. In seminal work, Liu and Talwar \cite{DBLP:conf/stoc/0001T19} consider a generalized formulation of these problems: a sequence of $\eps_1$-DP mechanisms, $(\calM_t)_{1\leq t\leq T}$, is given as input (in black-box manner), such that each mechanism's outputs are items in the ground set along with associated scores (in selection), or a score for the mechanism itself (in testing). The analyst generates an output which is an item in the selection problem and accept/reject in the testing problem. This generalization relaxes the Lipschitz requirement on the score function to distributional stability for a randomized mechanism in the sense of DP itself. Liu and Talwar proposed mechanisms that evaluate every mechanism $\calM_t$ multiple times and output a value computed from the aggregate of evaluations, but incur privacy loss that is \emph{invariant} in the number of evaluations. This approach offers significantly better privacy-utility trade-offs than using composition theorems \cite{dwork2014algorithmic}, especially for tasks like private hyperparameter optimization. 

Mechanisms for generalized private selection have been studied and applied extensively in the literature. Following the work of Liu and Talwar \cite{DBLP:conf/stoc/0001T19}, Papernot and Steinke \cite{papernot2021hyperparameter}, and Cohen, Lyu, Nelson, Sarl{\'o}s, and Stemmer \cite{cohen2023generalized} gave improved mechanisms that for $\eps_1$-DP inputs, $(\calM_t)_{1\leq t \leq T}$, achieved privacy loss of $(2+\theta)\eps_1$ for any fixed $\theta>0$. A lower bound argument in \cite[Appendix D]{DBLP:conf/stoc/0001T19} indicates that under a reasonable notion of accuracy, this is the best privacy guarantee achievable.

In this paper, we study generalized private testing. Formally, we are given a private dataset $X$ drawn from a dataset universe $\calX$ equipped with an arbitrary symmetric notion of adjacency, a target success probability $\pstar \in (0,1)$, and a public sequence of black-box $\varepsilon_t$-DP mechanisms $(\calM_t)_{t\geq1}$ with range $\{ -1,+1\}$ such that $p_t=\Pr[\calM_t(X)=+1]$.  
The goal of generalized private testing is that the mechanism should correctly halt and output the first mechanism $\calM_t$ for which $p_t \geq \pstar$, while achieving an $\eps$-DP guarantee. The accuracy of a mechanism for generalized private testing is described in terms of a \textit{rejection threshold} $\bar{p}< \pstar$, such that, with high probability for all $t\geq 1$, if $p_t \leq \bar{p}$ then the mechanism does not halt, and if for some $t\geq 1$, $p_t \geq \pstar$, it halts. For $p_t \in (\bar{p},\pstar)$, the mechanism may halt or continue arbitrarily. Ideally, we would like $\bar p$ to be as close to $p^*$ as possible.  \cite{DBLP:conf/stoc/0001T19} gave a $(\eps,\delta)$-DP mechanism that achieves $\bar{p} \approx \pstar (\beta^2/T)^{12\eps_1/\eps}$, where $\eps$ and $\delta$ are the privacy loss parameters of the testing mechanism, $\beta$ is a given failure probability, and $T$ is the number of input mechanisms.

\cite{cohen2023generalized} showed how to unify the generalized private selection and testing problems into one framework by introducing \emph{correlated random dropping}.
Ghazi, Kamath, Knop, Kumar, Manurangsi and Zhang \cite{ghaziprivate}
gave similar algorithms for generalized private selection and testing but with new accuracy guarantees for selection in the {\em ex-post setting}~\cite{DBLP:journals/jpc/WuRLWN19}. In contrast with \cite{DBLP:conf/stoc/0001T19}, both of these works provide a pure DP mechanism when the input mechanisms are also pure DP, and an approximate DP mechanism for approximate DP input mechanisms. However, while the testing mechanism in \cite{DBLP:conf/stoc/0001T19} can achieve any output privacy loss parameter $\varepsilon$ regardless of the input privacy loss parameter $\varepsilon_1$, \cite{cohen2023generalized, ghaziprivate} require that $\varepsilon > 2\varepsilon_1$ and do not state an accuracy guarantee along the lines of \cite{DBLP:conf/stoc/0001T19} for generalized testing. 

We show in \Cref{sec:cohen2023generalized} that approaches based on correlated random dropping \emph{cannot} achieve a rejection threshold better than $\bar{p} \approx \pstar (\beta/T)$; this is asymptotically \emph{worse} than \cite{DBLP:conf/stoc/0001T19}. To be more precise, regardless of the choice of privacy parameters, $p^*/\bar p$ increases with  $T/\beta$. In contrast, for the mechanism of~\cite{DBLP:conf/stoc/0001T19}, $p^*/\bar p$ can be reduced to $(T/\beta)^c$ for any arbitrarily small constant $c$.

To summarize, existing mechanisms for  generalized private testing have 
\begin{enumerate}
    \item either a large constant factor coefficient in the exponent of $\beta/T$ in $\bar p$ (i.e, $\bar{p} \approx \pstar (\beta^2/T)^{12\eps_1/\eps}$), restrictive privacy analyses (i.e., achieving only approximate privacy and only for pure DP inputs), and a large number of evaluations of each mechanism $\calM_t$
    ~\cite{DBLP:conf/stoc/0001T19}; or 
    \item tight privacy guarantees, but with an accuracy that scales poorly with $T/\beta$, regardless of the choice of $\eps$~\citep{cohen2023generalized}.
\end{enumerate}

Further, while mechanisms for generalized private selection are used extensively, mechanisms for generalized private testing have not seen many applications. In this paper, we address these gaps by introducing new algorithms, lower bounds, and applications for generalized private testing.

The rest of this section is organized as follows. We give a formal problem statement in \Cref{sec:problem} followed by our contributions in \Cref{sec:contribution}, applications in \Cref{sec:applications}, a high-level technical overview of our proof in \Cref{sec:technical}, and practical implications of our mechanism in \Cref{sec:discussion}.

\subsection{Problem description}
\label{sec:problem}
Let $\calX$ denote the dataset universe, equipped with a fixed symmetric adjacency relation; we call $X, X' \in \calX$ \emph{neighboring} if they are adjacent.
All results hold for any choice of adjacency relation. We consider a sequence of binary mechanisms $\calM_t : \calX \to \{-1,+1\}$ for $t = 1, 2, \ldots$, where each $\calM_t$ satisfies $(\eps_t,\delta_t)$-differential privacy, i.e. for all neighboring $X, X'$ and all $v \in \{-1,+1\}$,
\[ \Pr[\calM_t(X) = v] \leq e^{\eps_t} \Pr[\calM_t(X') = v] + \delta_t.\]

We also consider a sequence of datasets $(X_t)_{t \geq 1}$ and privacy parameters $(\eps_t,\delta_t)_{t \geq 1}$, which, along with the mechanisms $(\calM_t)_{t\geq1}$,  may be chosen adaptively based on prior public outputs, subject only to the constraint that $\calM_t$ is $(\eps_t,\delta_t)$-DP.
We write $p_t := \Pr[\calM_t(X_t) = +1]$ for the success probability at step $t$. 

A \emph{direction} $s_t \in \{-1,+1\}$ and a \emph{target threshold} $\pstar_t \in (0,1)$
are associated with each step.
The \emph{target regime} at step $t$ is
$\{p_t : s_t\, p_t \geq s_t\, \pstar_t\}$. In other words, when $s_t = +1$, then we want a tester to halt when $p_t \geq \pstar_t$, and when $s_t = -1$, we want the tester to halt when $p_t \leq \pstar_t$. 

\begin{defn}[Generalized private tester]\label{def:tester}
Fix a global privacy budget $\eps > 0$.
A \emph{generalized private tester} is a randomized procedure that, at each step $t$, interacts with the dataset $X_t$ only through i.i.d.\ samples from $\calM_t(X_t)$.
It draws a (potentially random) number of independent evaluations, observes their outcomes, and outputs $a_t \in \{-s_t, s_t\}$, where $a_t = s_t$ denotes \emph{halting} and $a_t = -s_t$ denotes \emph{continuation}.
The \emph{transcript} is $(a_1, \ldots, a_T)$ where $T$ is the halting step ($a_T = s_T$, $a_t = -s_t$ for $t < T$), or the infinite sequence $(-s_1, -s_2, \ldots)$ if the tester never halts.
The tester must be $\eps$-DP, i.e. for every $t \geq 1$, every pair of dataset sequences $(X_j)_{j \geq 1}$, $(X'_j)_{j \geq 1}$ with $X_j, X'_j$ neighboring for all $j$,
and every measurable $\calE \subseteq \{-1,+1\}^t$,
\[
  \Pr\!\big[(a_1, \cdots, a_t) \in \calE\big]
  \;\leq\; e^{\eps}\,\Pr\!\big[(a_1', \cdots, a_t') \in \calE\big].
\]
\end{defn}

\begin{defn}[Accuracy and sample complexity]\label{def:rejection}
    We say that a generalized private tester when run with a failure probability of $\beta$ and target thresholds $(\pstar_t)_{t\geq1}$ achieves \emph{rejection thresholds} $(\pbar_t)_{t\geq1}$ if with probability $1-\beta$ for all $t\geq1$, $s_t p_t\leq s_t \pbar_t \Rightarrow a_t = -s_t$ (no false positives), and $s_t p_t\geq s_t \pstar_t \Rightarrow a_t = s_t$ (no false negatives). The \emph{sample complexity} of a generalized private tester is the number of times it evaluates $\calM_t$ on the dataset $X_t$.
\end{defn}

This problem setting is a mild reformulation of the version introduced by \cite{DBLP:conf/stoc/0001T19}. In prior work, the notion of $\beta$-global accuracy was considered only by \cite{DBLP:conf/stoc/0001T19}, and all work implicitly set $s_t = +1$ for all $t\geq1$, i.e. the \emph{Above Threshold} formulation. 

\subsection{Our contributions.} \label{sec:contribution}

Our main contribution is a new mechanism for generalized private testing.

\begin{thm}[Upper Bound, Informal version of \Cref{thm:GlobalUtility} and \Cref{cor:approxDP}]
\label{thm:intro_upper_bound}
    Given a sequence of target thresholds $(\pstar_t)_{t\geq1}$, a privacy parameter $\eps>0$, a tuning parameter $\theta>0$, a failure probability $\beta$, a ratio parameter $\gamma\in (1,2]$, and a stream of mechanisms and dataset pairs $(\calM_t,X_t)$ where $\calM_t$ is $(\eps_t,\delta_t)$-DP, there is an $\eps$-DP mechanism called the Generalized Thresholding Mechanism (\GTM, \Cref{alg:GTM}).
    For all $t\geq1$, define $\Lambda_t = (5t \ln^3 (t+2))^{(2+\theta)\eps_t/\eps}\,(4/\beta)^{(3 + \theta + 2/\theta)\eps_t/\eps}$ and the {\em rejection thresholds}
        \[ \bar{p}_t \;:=\; \max\!\left(\frac{\pstar_t}{\gamma\,\Lambda_t},\;\; 1 - \gamma\,\Lambda_t\,(1-\pstar_t)\right). \]
    Then, with probability $1-\beta$ over the randomness of the \GTM, for all $t\geq1$, the following statements hold:
    \begin{enumerate}
        \item If $p_t \geq \pstar_t + \delta_t/\eps_t$, then the mechanism  halts.
        \item If $p_t \leq \bar{p}_t - \delta_t/\eps_t$, then the mechanism  continues.
        \item The number of evaluations of the mechanism $\calM_t$ at step $t$ is
          $N_t$ drawn from a Poisson distribution $\Po(\lambda_t)$ for
          $\lambda_t \leq \tilde{A}_t \cdot w_t$, where
          \[
            \tilde{A}_t \;:=\; O\!\left(\frac{\ln(t/\beta)}{(\gamma-1)^2}
            \cdot \max\!\left(\frac{(4/\beta)^{(1+2/\theta)\eps_t/\eps}}{\pstar_t},\;
            \frac{1}{1-\pstar_t}\right)\right)
          \]
          and $w_t := e^{\eps_t\eta_{2,t}}$. $w_t$ is distributed as a Pareto distribution $\Par(\sigma_t)$ with
          $\sigma_t := \eps/(\eps_t(\theta+2))$, independently across steps.
            \end{enumerate}
\end{thm}

We see that for a generic value of $\pstar_t$, \Cref{thm:intro_upper_bound} achieves the rejection threshold 
    \[ \bar{p}_t \geq \pstar_t\, \cdot\, \tilde{\Theta} (t^{-(2 + \theta)\eps_t/\eps} \beta^{(3 + \theta + 2/\theta)\eps_t/\eps}).\]  
Our sample complexity is randomized, but in \Cref{lem:sample_complexity} we show that with probability $1-\beta$, for all $t\geq1$, $N_t \leq \frac{\ln(t/\beta)}{(\gamma-1)^2}\max\left(\frac{\Lambda_t}{\pstar_t},\frac{1}{1-\pstar_t}\right)$. Further, in the setting where $\eps_t = \eps_1 < \eps/(2+\theta)$ for all $t\geq1$, the \emph{amortized} sample complexity obeys the bound: 
\[ \Pr\left[ \frac{1}{T}\sum_{t=1}^T N_t \geq \tilde{O}\left( \frac{\eps}{\eps-(2+\theta)\eps_1}\cdot \max\left(\frac{(4/\beta)^{(1 + 2/\theta)\eps_1/\eps}}{\pstar},\frac{1}{1-\pstar}\right)\right)\right] < \beta. \] 

\noindent \textbf{Comparison with \cite{DBLP:conf/stoc/0001T19}.} The exponent of $t$ can be driven down to $(2+\theta)\eps_t/\eps$ for any $\theta>0$ at the cost of a reciprocally larger exponent of $\beta$, but the coefficients of both $t$ and $\beta$ are simultaneously smaller than those achieved in \cite{DBLP:conf/stoc/0001T19}, for which $\Lambda_t = (T/\beta^2)^{12\eps_t/\eps}$. The ratio between the thresholds is large but finite for any choice of $\eps$, and on studying the $\max$ expression for $\bar{p}_t$ in \Cref{thm:intro_upper_bound} we see that the gap between the thresholds $|\pstar - \bar{p}_t|$ reduces to $0$ as the success threshold $\pstar$ approaches $0$ or $1$; this is the case for both our work and \cite{DBLP:conf/stoc/0001T19}, but more rapidly in our case since our value of $\Lambda_t$ is smaller. Apart from the improvement in the exponent, the more significant points of comparison between our works are the privacy guarantees and sample complexities. Our mechanism guarantees pure DP regardless of whether $(\calM_t)_{t\geq 1}$ are pure or approx-DP, while \cite{DBLP:conf/stoc/0001T19} requires $(\calM_t)_{t\geq 1}$ to be pure DP and their mechanism achieves only approximate DP guarantees. 
\cite{DBLP:conf/stoc/0001T19} have a  \emph{per-round} sample complexity of $\approx \tfrac{(T/\beta)^{12\eps_1/\eps} \ln (T/\delta)}{\min\{\pstar,1-\pstar\}} $. 
As described above, our sample complexity obeys a high probability bound with a smaller exponent, and when $\eps>(2+\theta)\eps_t$, the amortized sample complexity scales only \emph{logarithmically} with the time-step index, and incurs a $(4/\beta)^{(1 + 2/\theta)\eps_1/\eps}$ dependence on the failure probability $\beta$.

\noindent \textbf{Comparison with \cite{cohen2023generalized}.} \cite{cohen2023generalized} do not state their accuracy bound along the lines of \cite{DBLP:conf/stoc/0001T19} (and the one used in this work). Therefore, to make a fair comparison, in \Cref{prop:CohenAccuracy}, we derive an accuracy guarantee for their mechanism. 
It has a similar multiplicative relationship between $\pstar$ and $\bar{p}$ but with two caveats: (1) the rejection threshold is at best $\bar{p}_t = \beta \pstar/t$, regardless of the choice of $\eps_t$ and $\eps$; and (2) the gap between the thresholds vanishes when $\pstar$ approaches $0$, but not when it approaches $1$. 
Further, they require that $\eps > 2\eps_1$. 
Although our analysis does not preclude a more refined analysis, we prove a lower bound (\Cref{cor:CohenAccuracyLimit}) showing that the factor of $\beta/t$ is unavoidable for their algorithm. 
The work of \cite{ghaziprivate} uses similar algorithmic techniques for generalized private testing and selection, and when applied to this problem, suffers the same $\beta/t$ slack (see \Cref{sec:ghaziprivate}).

For ease of comparison, we tabulate accuracy and sample complexity bounds in \Cref{tab:comparison}. We compare the rejection thresholds described in terms of the \emph{noise penalty} (which captures the impact of privatizing perturbations), the rejection threshold achieved, and the sample complexity. For a fair comparison we mention only our high-probability sample complexity bound above, but as mentioned before, our amortized scaling is significantly better.

\begin{table}[t]
\centering
\renewcommand{\arraystretch}{1.8}
\caption{Comparison of accuracy and sample complexity for generalized private testing. Expressions suppress lower-order terms. For our results, $\theta > 0$ and $\gamma \in (1,2]$ are free parameters. For \cite{DBLP:conf/stoc/0001T19}, $\eps_0 \in (0,1)$ is an auxiliary parameter.  ${}^\dagger$\cite{DBLP:conf/stoc/0001T19}'s guarantee is conditional on halting at step $t$ and does not guarantee halting at $p_t = \pstar$, but at a value somewhat larger than $\pstar$.}
\resizebox{\textwidth}{!}{%
\begin{tabular}{@{}lccc@{}}
\toprule
& \textbf{Our results} & \textbf{Cohen et al.~\cite{cohen2023generalized}} & \textbf{Liu--Talwar~\cite{DBLP:conf/stoc/0001T19}} \\
& (\Cref{thm:GlobalUtility}) & (\Cref{prop:CohenAccuracy}) & (\Cref{thm:LiuTalwar}) \\
\midrule
\textbf{Privacy Guarantee} 
  & pure $\eps$-DP 
  & pure $\eps$-DP 
  & $(\eps,\delta)$-DP \\
\textbf{Requirement}
  & $\eps > 0$
  & $\eps > 2\eps_1$
  & $\eps > 0$, finite $T$ \\
\midrule
Noise penalty, $\Lambda_t$
  & $\displaystyle \left(\frac{1}{\beta}\right)^{\!(1+2/\theta)\eps_1/\eps}\!\!\left(\frac{t}{\beta}\right)^{\!(2+\theta)\eps_1/\eps}$
  & $\displaystyle \frac{t}{\beta^{1+\eps_1/(\eps-2\eps_1)}}$
  & $\displaystyle \left(\frac{t}{\beta^2}\right)^{\!12(\eps_1+\eps_0)/\eps}{}^\dagger$ \\[10pt]
Rejection threshold, $\bar{p}_t$
  & $\displaystyle \max\!\left(\frac{\pstar_t}{\gamma\Lambda_t},\; 1 - \gamma\Lambda_t(1\!-\!\pstar_t)\right)$
  & $\displaystyle \frac{\pstar}{\Lambda_t}$
  & $\displaystyle \frac{\pstar}{\Lambda_t(1\!-\!\pstar) + \pstar}{}^\dagger$ \\[10pt]
\midrule
Sample complexity, $N_t$
  & $\displaystyle \frac{\ln(t/\beta)}{(\gamma-1)^2}\max\left(\frac{\Lambda_t}{\pstar_t},\frac{1}{1-\pstar_t}\right)$
  & $\displaystyle \frac{\ln(1/\beta)}{\pstar\,\beta^{\eps_1/(\eps-2\eps_1)}}$
  & $\displaystyle \frac{(T/\beta)^{12(\eps_1+\eps_0)/\eps}\,\ln(T/\delta)}{\eps_0^2\,\min\{\pstar,1-\pstar\}}$ \\[10pt]
\bottomrule
\end{tabular}
}
\label{tab:comparison}
\end{table}

We complement our upper bound on the generalized private testing with the following lower bound on the accuracy of any generalized private tester.
\begin{thm}
[Lower Bound, Informal version of \Cref{cor:GTM_LB} and \Cref{thm:sample_complexity_LB}]
\label{thm:intro_lower_bound}
    Given a target threshold $\pstar$, a stream length $T$, an input privacy parameter $\varepsilon_1>0$, output privacy parameters $\varepsilon>0$ and $\delta\in [0,1)$ and a failure probability $\beta$, define $\delta' := \delta/(e^{\eps} - 1)$, and suppose $\beta+\delta'<1/2$. Define $\overline{\Lambda}_T = e^{-2\eps_1} \cdot \left(T/(4(\beta + T\delta')(\beta + \delta'))\right)^{\!\eps_1/\eps}$. Suppose $\calA$ is any generalized private tester that achieves the rejection threshold $\bar{p}$. Then:
    \begin{enumerate}
        \item If $\pstar, \bar{p} \leq 1/2$, then $\bar{p}\leq \pstar/\overline{\Lambda}_T $.
        \item If $\pstar, \bar{p} \geq 1/2$, then $\bar{p} \leq 1-  \overline{\Lambda}_T (1-\pstar)$.
        \item If $\bar{p} < 1/2 < \pstar$, then $\bar{p} \leq 1/(4(1-\pstar)\overline{\Lambda}_T)$.
    \end{enumerate}
    Further, we show that any generalized private tester for a given target threshold $\pstar<1/2$ with failure probability at most $\beta$ must satisfy:
        \[ \Ex[N_t] = \Omega(1/(\pstar \cdot \beta^{\eps_1/\eps})). \]
\end{thm}
When $\pstar$ approaches $0$ or $1$, the functional form of our upper and lower bounds is identical: (i) $\bar{p}_t = \pstar_t/(\gamma\Lambda_t)$ when $\pstar$ is close to $0$ compared with $\bar{p} \leq \pstar_t/\overline{\Lambda}_T$ in the lower bound and (ii) $\bar{p}_t = 1 - \gamma\Lambda_t(1-\pstar_t)$ compared with $1 - \overline{\Lambda}_T(1-\pstar_t)$ in the lower bound. In these domains, the comparison reduces to the gap between $\gamma \Lambda_t$ and $\overline{\Lambda}_T$; we note that $\gamma\leq 2$ so this factor has limited impact. The lower bound for pure DP inputs and approximate DP outputs still requires that $\overline{\Lambda}_t \approx (T/((\beta+T\delta)(\beta+\delta)))^{\eps_1/\eps} \approx (T/\beta^2)^{\eps_1/\eps}$ as $\delta \ll \beta$. If we have $\eps_t = \eps_1$, our upper bound achieves a noisy penalty term $\Lambda_t = t^{(2 + \theta)\eps_t/\eps} \beta^{-(3 + \theta + 2/\theta)\eps_t/\eps}$. The gap between the upper and lower bounds is thus reduced to small constants in the exponents of $t^{\eps_t/\eps}$ and $\beta^{\eps_t/\eps}$ in all privacy settings. We can reduce the gap in the exponent of $t$ to a near-$2$ factor at the cost of a larger exponent on $\beta^{\eps_t/\eps}$; conversely, when we minimize the exponent of $\beta^{\eps_t/\eps}$, then we have that $\Lambda_t \approx \beta^{-(3+2\sqrt{2})\eps_t/\eps} t^{(2+\sqrt{2})\eps_t/\eps}$.

Our upper and lower bounds show that solving generalized private testing necessitates a loss in accuracy captured by a $\Lambda_T$ factor, that scales essentially as $(T/\beta)^{c\eps_t/\eps}$, for some constant $c$, and an average sample complexity that scales as $(1/\pstar)\beta^{c'\eps_1/\eps}$ for some constant $c'$. The lower bounds show that it is impossible to get rid of these polynomial in $t$ and $\beta$ factors, but that the impact of this factor in terms of the absolute gap between the acceptance and rejection thresholds $\pstar$ and $\bar{p}$ vanishes as $\pstar$ approaches $0$ or $1$.

Our upper bound suggests another way of mitigating the $\Lambda_t$-factor in the gap between $\pstar$ and $\bar{p}$, which is to let the ratio between the privacy parameters $\eps/\eps_t$ scale as $\ln(t/\beta)$. We formalize this as an \emph{ex-post} guarantee.

\begin{defn}[Ex-post privacy \cite{DBLP:journals/jpc/WuRLWN19}]
    Given a randomized algorithm $\calA : \calX \to \calY$, define the ex-post privacy loss of $\calA$ on $o\in\calY$ to be a function $\tilde{\eps}:\calY\to\R_{\geq0}$ if for all neighboring $X,X' \in \calX$, $\Pr[\calA(X) = o] \leq e^{\tilde{\eps}(o)} \Pr[\calA(X') = o]$.
\end{defn}

\begin{thm}
[Ex-post setting, Informal statement of \Cref{thm:ExPost}]
\label{thm:intro_ex_post}
    In the setting of \Cref{thm:intro_upper_bound}, define $\Lambda_0 = e^2 \left(4/\beta\right)^{\eps_1/\eps}$. For all $t\geq1$, \Cref{alg:GTM} suffers ex-post privacy loss at most 
    \[ \tilde{\eps}((a_1,\dots,a_t)) = \begin{cases} \eps + \eps_t \ln (t/\beta)&\mbox{ if } a_t = s_t \\ \eps &\mbox{ if }a_t = -s_t. \end{cases}\]
    The accuracy and high-probability sample complexity guarantees are identical to that of \Cref{thm:intro_upper_bound}, with $\Lambda_t$ replaced by $\Lambda_0$.
\end{thm}

\subsection{Applications} 
\label{sec:applications}

We have seen that in \Cref{alg:GTM}, one can reduce this loss in accuracy $\Lambda_t$ to any arbitrarily small constant $C_{\Lambda}>1$ in two ways: \textbf{(i)} by allowing the output privacy loss to scale logarithmically with $t/\beta$, i.e. setting $\eps = \eps_t \ln (t/\beta)$, or \textbf{(ii)} by forcing the input privacy guarantees to become more restrictive, setting $\eps_t = \eps \cdot 1/\ln (t/\beta)$. 

The former guarantee serves as a way to solve the testing problem with a constant gap between the rejection threshold and incurring privacy loss scaling logarithmically with the length of the stream. In comparison, a fixed choice of $\eps$ and desired constant rejection threshold of $\bar{p}$ would eventually cause false positives to occur at a constant rate; the privacy loss incurred by each false positive must be accounted for via basic or advanced composition, leading to a net privacy loss of $\approx \sqrt{t}\eps_t$ for long streams with an approximate DP guarantee. In contrast, by committing to a logarithmically scaling privacy loss, we are guaranteed to not incur any false positives with high probability and solve the same testing problem with a pure DP guarantee with privacy loss $= \eps_1 \ln t$. 
The benefits of the latter guarantee, i.e. when $\eps_t = \Theta( \eps/\ln (t/\beta))$, are best illustrated by our next contribution that we discuss now.

\medskip \noindent \textbf{Reduction from Optimization in the Continual Observation setting to the Batch setting.} Over the past few years, the \emph{continual observation} (CO) model of differential privacy~\cite{Dwork-continual} has received a lot of interest. Unlike the standard {\em batch setting}, where the entire dataset $X$ is known {\em a priori}, in the CO setting, the data is constantly modified by a stream of updates, with the dataset after the $t$-th update denoted by $X_t$. In the CO setting a solution $Y_t$ must be output after every update, and the privacy requirement is that the entire output stream $(Y_t)_{t\geq1}$ is differentially private with respect to sequence of datasets seen in the input $(X_t)_{t\geq1}$.

Dwork, Naor, Pitassi, and Rothblum in their seminal work~\cite{Dwork-continual} posed a question whether the existence of a DP algorithm for a problem in the batch setting implies the existence of a DP algorithm in the continual observation setting. 
They~\cite[Section 6]{Dwork-continual} proposed an approach in which a mechanism in the CO setting maintains and repeatedly outputs a fixed solution. It  modifies the output only when it no longer serves as a good solution for the updated dataset. Detecting when the prior solution is no longer of adequate quality is a \textit{threshold monitoring problem}: at each time step,  recomputation is triggered when the test passes and indicates that the function has changed significantly. By trading off some accuracy on average, over the length of the entire stream, one can achieve superior trade-offs between accuracy and privacy.

In the standard private testing setting, the {\em sparse vector technique} (SVT) is used to solve threshold monitoring efficiently when monitored function has bounded Lipschitz sensitivity. Unfortunately, the quality of the solutions tested by the CO scheme described above admits no tractable Lipschitz bound in general, and none of these privacy accounting techniques apply. This limits the scope of the general technique introduced by \cite{Dwork-continual}. We give a brief overview of how our mechanism for generalized testing allows us to extend this reduction for a broad class of problems, and we refer the readers to \Cref{sec:batchtocontinual} for more detail.

When a mechanism for an optimization problem exists in the batch setting, we can circumvent the bounded Lipschitz requirement of \cite{Dwork-continual} as follows. Rather than attempting to directly monitor the optimal utility $\OPT_t$ at time $t$, we use a black-box batch DP algorithm $A(\cdot)$. At each time step, we give the \GTM query access to $A(X_t)$ to test whether the quality of solutions generated by $A(X_t)$ is significantly better compared to the quality of the previous most recently generated output.

More precisely, let $f:\calX\times\calY\to \R$ be a maximization problem, which for a given dataset $X$ asks for $\argmax_{Y\in\calY} f(X,Y)$; with this notation $\OPT_t = \max_{Y\in \calY} f(X_t,Y)$. Let $A$ be an algorithm that is $\eps_1$-DP for any $\eps_1>0$, and such that for $Y\gets A(X)$, with probability $1-\beta_A$, $f(X,Y) \geq \alpha \OPT_t - E_A (\eps_1)$ for some multiplicative approximation $\alpha \in (0,1)$ and additive error $E_A(\eps_1) \geq 0$. Our reduction gives an algorithm in the continual observation setting that for any $\eps<1$, with probability $1-O(\beta)$, for all $t\geq 1$, generates a sequence of outputs $(Y_t)_{t\geq1}$ such that 
    \[ f(X_t,Y_t) \geq \frac{\alpha \OPT_t}{1+\kappa} - E_A\paren{\frac{\kappa \eps}{\ln^2(t/\beta)}} - O\left(\frac{\ln^3 (t/\beta)}{\kappa\eps} \left(1 + \ln\frac{1}{1-\beta_A} \right)\right).\]
Further, the total number of evaluations of $A$ at time-step $t$ is only $O(\ln (t/\beta))$.

The fact that we only need the batch algorithms to be differentially private, and not fulfill some more restrictive Lipschitz bound constraint, gives us significantly more flexibility in using algorithms that apply in the batch setting for constructing new ones in the CO setting.

Another key feature of our generalized threshold mechanism that makes this reduction effective is that, when we use the batch DP algorithm with $\eps_t < \eps / \ln(t/\beta)$ (we set $\eps_t < \eps / \ln^2(t/\beta)$ for technical reasons), the ex-post accuracy guarantee ensures that the noise penalty $\Lambda_0$ is bounded by a universal constant independent of $t$, so the per-step accuracy of each threshold test does not degrade over time. This implies that we can invoke the batch algorithm at each step with a constant failure probability $\beta_A$, rather than one that decays with the stream length. 
This time-independent failure probability requirement even allows us to adapt batch algorithms that provide only in-expectation accuracy guarantees, rather than high-probability bounds on accuracy. As an in-expectation accuracy bound of the form $\mathbb{E}[f(X, A(X))] \geq \alpha \cdot \text{OPT}(X) - E_A$ implies a bounded constant failure probability (see \Cref{lem:exp-to-prob}), such algorithms can be plugged directly into the reduction. In other words, there is no accumulation of the failure probabilities of the invocations of the batch algorithm.

Our reduction applies to any bounded-sensitivity, data-monotone maximization problem, and yields the first DP algorithms under continual observation for several fundamental problems, including submodular maximization under cardinality and matroid constraints, weighted densest subgraph, and Max-Cut. Crucially, the invocations of the batch algorithm in the $t$-th time-step occur with multiplicative privacy parameter set to equal $\eps_t = \eps/\ln^2 (t/\beta)$; this is in stark contrast with the $\sqrt{t}$ scaling needed via a direct application of advanced composition. In Table~\ref{tab:comparison_2_intro}, we cover some examples to illustrate the benefits of these flexibilities, and direct the reader to \Cref{sec:co-applications} for a more complete treatment.

\renewcommand{\arraystretch}{1.4}
\begin{table}[t]
\footnotesize
\caption{Application of \Cref{thm:batch-to-co} to batch DP submodular maximization algorithms. $\OPT$: batch optimum; $\OPT_t$: CO optimum at time $t$; $k$: cardinality or matroid rank constraint; $\kappa,\eta > 0$: user-defined parameters; $\beta$: failure probability. Monotone/non-monotone refers to argument-monotonicity; our framework additionally requires data-monotonicity. For in-expectation batch guarantees, \Cref{cor:batch-to-co-exp} is applied with small constant $c > 0$.}
\label{tab:comparison_2_intro}
\centering
    \begin{tabular}{@{}L{3.2cm} L{4.8cm} L{5.2cm}@{}}
    \toprule
    \textbf{Work, Setting} & \textbf{Batch Guarantee} & \textbf{CO Guarantee (Ours)} \\
    \midrule
    \cite{mitrovic2017differentially}\newline
    Monotone,
        $(\eps,\delta)$-DP, $k$-Card. 
        & $(1 {-} \tfrac{1}{e})\OPT - \tilde{O}\bigl(\tfrac{k^{3/2}}{\eps}\ln m \ln \tfrac{1}{\delta} \bigr)$ 
        & $\tfrac{1 - 1/e}{1+\kappa}\,\OPT_t - \tilde{O}\bigl(\tfrac{k^{3/2}}{\kappa \eps} \ln^3 t \ln m \ln \tfrac{1}{\delta} \bigr)$ \\ 
    \addlinespace
    \cite{DBLP:conf/icml/RafieyY20}\newline
    Monotone,
        $(\eps,\delta)$-DP, $k$-Card.
        & $(1 {-} \tfrac{1}{e})\OPT - \tilde{O}\bigl(\tfrac{k^{7}}{\eps^3} m \ln m\bigr)$ 
        & $\tfrac{1 - 1/e}{1+\kappa}\,\OPT_t - \tilde{O}\bigl(\tfrac{k^{7}}{\kappa \eps^3} m \ln m \ln^6 t\bigr)$ \\
    \addlinespace
    \cite{DBLP:conf/aaai/ChaturvediNZ21}\newline
    Decomp.\ non-mon.\ $(\eps,\delta)$-DP, $k$-Card.
        & $(\tfrac{1}{e} {-} \eta) \OPT - \tilde{O}\bigl(\tfrac{k}{\eta\eps}\ln m \ln\tfrac{1}{\delta}\bigr)$
        & $\tfrac{1/e - \eta}{1+\kappa}\,\OPT_t - \tilde{O}\bigl(\tfrac{k}{\eta\kappa \eps}\ln^3 \tfrac{t}{\beta}\ln m \ln\tfrac{1}{\delta}\bigr)$
        \\
    Decomp.\ non-mon.\ $(\eps,\delta)$-DP, $k$-Matroid
        & $(\tfrac{1}{e} {-} \eta) \OPT - \tilde{O}\bigl(\tfrac{k}{\eta\eps}\ln m \ln\tfrac{1}{\delta}\bigr)$
        & $\tfrac{1/e - \eta}{1+\kappa}\,\OPT_t - \tilde{O}\bigl(\tfrac{k}{\eta\kappa \eps}\ln^3 \tfrac{t}{\beta}\ln m \ln\tfrac{1}{\delta}\bigr)$
        \\
    \addlinespace
    \cite{ghazi2024individualized}\newline
    Decomp.\ monotone $\eps$-DP, $k$-Card., $\eta {\in} (0,1)$  
        &  $(1{-}\tfrac{1}{e}{-}\eta)\OPT - O\bigl(\tfrac{k\ln m/\beta}{\eps}\bigr)$ 
        & $\tfrac{1-1/e-\eta}{1+\kappa}\,\OPT_t - \tilde{O}\bigl(\tfrac{k}{\kappa \eps} \ln^3 \tfrac{t}{\beta} \ln \tfrac{m}{\beta}\bigr)$\\
    Decomp.\ monotone $\eps$-DP, $k$-Matroid, $\eta {\in} (0,1)$  
        &  $(1{-}\tfrac{1}{e}{-}\eta)\OPT - O\bigl(\tfrac{k\ln m/\eta \beta}{\eps}\bigr)$ 
        & $\tfrac{1-1/e-\eta}{1+\kappa}\,\OPT_t - \tilde{O}\bigl(\tfrac{k}{\kappa \eps} \ln^3 \tfrac{t}{\beta} \ln \tfrac{m}{\eta \beta}\bigr)$ \\
    \bottomrule
    \end{tabular}
\end{table}

\subsection{Technical Overview of the Generalized Thresholding Mechanism}
\label{sec:technical}
The starting point of our approach is that any mechanism for generalized private testing needs to evaluate $\calM_t(X_t)$ a (potentially randomized) $N_t$ number of times to gain information about $p_t$, where $p_t = \Pr[\calM_t(X_t)=+1]$.
The number $K_t$ of $+1$ outputs of these evaluations should be a sufficient statistic to reason about whether $\calM_t(X_t)$ should be accepted or rejected. The problem then reduces to defining (i) how to pick $N_t$; and (ii) for what values of $K_t$ should one accept $\calM_t(X_t)$, while ensuring a tight privacy-accuracy trade-off. 

\medskip \noindent \textbf{Achieving Privacy, Step 1: Poissonized sampling.} 
Our first key technical insight is that if we pick $N_t \sim \Po(\lambda_t)$ for mean parameter $\lambda_t$, then $K_t \sim \Po(\lambda_t p_t )$ as a  consequence of \textit{Poisson thinning}.   

\medskip \noindent \textbf{Achieving Privacy, Step 2: Randomizing $\lambda_t$.} 
Next, we perturb $\lambda_t$ just enough to obfuscate $p_t$ and $e^{-\eps_t}p_t$, which in turn allows us to use a coupling argument for the privacy proof similar to the proof of the SVT. 
To perturb $\lambda_t$, let it be a function of a random variable $Z$. Then $K_t \sim \Po(\lambda_t(z) p_t )$, where $z$ is the value assumed by $Z$, and the algorithm decides to stop if $K_t > c_t$ for a suitably chosen threshold $c_t$.
For all $t\geq1$, $p_t'=\Pr[\calM_t(X_t')]$ where $X_t'$ is some dataset neighboring $X_t$.
As we are guaranteed that $p_t \ge e^{-\eps_t}p'_t$, we pick $\lambda_t(z)$ such that $\lambda_t(z) p_t  \ge \lambda_t(z-1) p'_t$.  
Specifically we choose $\lambda_t(z) = \rho_t e^{\eps_tz}$, where $\rho_t$ is a parameter that is needed for the accuracy bound and discussed further below.
This is useful as the CDF of the Poisson random variable $K_t$ is non-increasing in $\lambda_t(z) p_t $.
Thus, for any $z$, $F(c_t, \lambda_t(z) p_t) \le F(c_t, \lambda_t(z-1)p'_t)$, where $F(c, \ell):=\Pr_{A \sim \Po(\ell)}[A \le c]$ is the probability that a Poisson random variable with parameter $\ell$ returns a value less than $c.$

Let $a_t$ ($a'_t$, resp.) denote the output of the mechanism at time $t$ with input $X_t$ ($X'_t$, resp.). We show here the privacy argument for only the first time step, the detailed proof is given in Section~\ref{subsec:privacy}.
The probability that the mechanism when started on $X_t$ does not stop in the first time step is  $$\Pr(a_1=-1) = \int_{-\infty}^\infty F(c,\lambda(z)p_t) f_Z(z) dz \le  \int_{-\infty}^\infty F(c,\lambda(z-1)p'_t) f_Z(z) dz.$$
Shifting the variable to $u=z-1$ and assuming that $f_Z(u+1) \le e^{\eps}f_Z(u)$ upper bounds the previous expression by $$\int_{-\infty}^\infty F(c,\lambda(u)p'_t) f_Z(u+1) du \le
e^{\eps} \int_{-\infty}^\infty F(c,\lambda(u)p'_t) f_Z(u) du = e^{\eps} \Pr(a'_1 = -1)$$

To guarantee $f_Z(u+1) \le e^{\eps}f_Z(u)$, we sample a random variable $\eta_1$ using an exponential distribution with scale $1/\eps$ and could set $Z = - \eta_1$.
However, for the privacy proof of the halting case and for the accuracy analysis, we need to choose $Z$ in a more refined way, without destroying the crucial property of $f_Z$:
Setting $\lambda_t$ in this manner allows us to bound the privacy loss by $\eps$ for an arbitrarily long sequence of failures, i.e., $a_1 = \dots = a_t = -1$; however, we need to introduce an appropriately scaled noise at every input, $\eta_{2,t}$, to bound the privacy loss to account for $a_t = +1$ at some indeterminate time-step $t$ (similar to the coupling argument in the standard private testing setting). We sample $\eta_{2,t}$  using an exponential distribution with scale $2/\eps$, i.e.,  $\eta_{2,t} \sim \Exp(2/\eps)$ and for the coupling argument to work we could set $Z = - \eta_1 + \eta_{2,t}.$
However, for the purposes of the accuracy analysis, it will be convenient to introduce a constant offset $\mu_1$ to have $Z$ be positive with high probability. Formally, let $\mu_1 = (1/\eps) \ln (3/\beta)$ be a high-probability bound on the magnitude of $\eta_1$. Then we define 
$ Z_t := \mu_1 - \eta_1 + \eta_{2,t}.$  
Using the tail bound on exponential distribution, we can show that $Z_t\geq 0$ for all $t\geq 1$ with probability $1-\beta/3$ {over $\eta_1 \sim \Exp(1/\eps)$}. 

For point (ii) mentioned above, we test $K_t \geq c_t$ for an appropriately chosen $c_t$ (see \cref{eq:choosing_c_t}). Since threshold tests $K_t \geq c_t$ are monotonic in the value of $K_t$, by post-processing, the outcomes of the tests will also be $\eps$-DP.

\medskip \noindent \textbf{Achieving Tight Accuracy, Step 1: Choosing Base Rate, $\rho_t$.}
Recall that we define $\lambda_t$ to equal $\rho_t \exp(\eps_t Z_t)$ with $\rho_t>0$ and $Z_t := \mu_1 - \eta_1 + \eta_{2,t}$ and that there are so far no further restrictions on $\rho_t$.
By the choice of $\mu_1$ and since $\eta_{2,t}\geq0$ unconditionally, we have that with probability $1-\beta/3$ over the draw of $\eta_1$, we have $\lambda_t \geq \rho_t$. 
For the accuracy analysis we also need an upper bound on $\eta_{2,t}$ for all $t$. 
We define a sequence of failure probabilities $\beta_t = \Theta(\beta/(t \ln^3 t))$ such that $\sum_{t\geq 1}\beta_t = \beta/6$. 
Thus, with probability at least $1-\beta/2$, it holds that for all $t$, $\eta_{2,t} \leq (2/\eps) \ln (1/\beta_t)$ and  $\rho_t \le \lambda_t \leq \Lambda_t \rho_t$ 
with $\Lambda_t = \exp((\eps_t/\eps) (\ln(3/\beta) + \ln (1/\beta_t)))$.   

We now reason about how to conduct threshold tests where $N_t \sim \Po(\lambda_t)$ is the number of such tests and  $K_t \sim \Po(\lambda_t p_t )$ is the number of $+1$ outcomes.
At a high level, 
the distribution of $K_t$ has left and right tails roughly $\lambda_t p_t \pm \sqrt{6\lambda_t p_t  \ln (1/\beta_t)}$ with probability at least $1-2\beta_t$; both grow monotonically with $\lambda_t p_t $. 
We set $\rho_t := \rho_t^\dagger/\pstar_t$ for appropriately chosen $\rho_t^\dagger$.
Thus with probability at least $1-\beta/2 - \beta/3 = 1-5\beta/6$ for all $t$ it holds that
\begin{align*}
    K_t &\geq \rho_t^\dagger - \sqrt{6 \rho_t^\dagger \ln (1/\beta_t)}~~\textrm{ for }p_t \geq \pstar_t \quad \textrm{ and} \quad %\\
    K_t &\leq \frac{\rho_t^\dagger \Lambda_t \bar{p_t}}{\pstar_t} + \sqrt{\frac{6 \rho_t^\dagger \Lambda_t \bar{p_t}}{\pstar_t}  \ln (1/\beta_t)} ~~\textrm{ for }p_t \le \bar{p_t}.
\end{align*}

At a high level, we need to set $\overline{p_t}$ so as to compensate for the magnitude $\Lambda_t$ of the privatizing noise $\exp(\eps_t Z_t)$.
Setting $\overline{p_t} = \pstar_t/(\gamma\Lambda_t)$, for some fixed constant $\gamma>1$,  allows us to offset the $\Lambda_t$ factor and gives the following bound when $p_t \le \bar{p_t}$:
    \[ K_t \leq \frac{\rho_t^\dagger}{\gamma} + \sqrt{\frac{6 \rho_t^\dagger }{\gamma}} \ln (1/\beta_t). \]
To summarize, we require as threshold for our test a value of $c_t$ such that
\begin{align}
\label{eq:choosing_c_t}
    \rho_t^\dagger/{\gamma} + \sqrt{6 (\rho_t^\dagger/{\gamma})  \ln \paren{1/ \beta_t}} \, \leq\quad  c_t \quad \leq \, \rho_t^\dagger - \sqrt{6\rho_t^\dagger  \ln \paren{1/ \beta_t}}.
\end{align}

Note that, as $\rho_t^\dagger$ increases, the leading term of both bounds dominates, and for any $\gamma>1$, these bounds become simultaneously satisfiable. Working out the details we get that there is a value $\rho_t^\dagger = \Theta\big(\frac{\ln (1/\beta_t)}{(\gamma-1)^2}\big)$ that suffices. Combining  everything together, we arrive at the expression
\begin{align}
    \label{eq:choosing_rho_t}
    \rho_t = \Theta\left(\frac{\Lambda_t \ln (1/\beta_t)}{\pstar (\gamma-1)^2}\right) = \Theta\left(\frac{\Lambda_t \ln (t/\beta)}{\pstar (\gamma-1)^2}\right).
\end{align}

\medskip \noindent \textbf{The Flipped variant.} In the above discussion, we have generalized private testing with the rejection threshold $\overline{p_t} = \pstar/\gamma\Lambda_t$. When $\pstar$ approaches $0$, the gap $\pstar - \overline{p_t}$ in absolute terms vanishes, even though $\pstar/\overline{p_t}$ remains $\gamma \Lambda_t$. This is desirable because it means there are few values for $p_t$ where we do not have a guarantee on the performance of the mechanism. 
However, when $\pstar$ is close to $1$, then $\pstar - \overline{p_t}$ is large, which is undesirable for the same reason. This suggests the following modification when $\pstar$ is close to $1$: instead of running the mechanism as is, we \emph{flip} each test by running a threshold test on the evaluations of $-\calM_t(X_t)$ with the target threshold $1-\pstar$. Since $\Pr[\calM(X_t) = -1] = 1-p_t$, this is the same test as $1 - p_t < 1 - \pstar$ iff $ p_t > \pstar$. Now, this is the \emph{Below Threshold} variant of the Thresholding Mechanism, i.e., the mechanism halts when $K_t \le c_t.$

Recall that the noise added to ensure privacy is one-sided with different scale (i.e., $\eta_1\sim \Exp(1/\eps)$ and $\eta_{2,t}\sim \Exp(2/\eps)$). Therefore, to account for the difference in the problem (i.e., Below threshold vs Above threshold), we also need to flip the noise variable: let $s_t$ denote whether we want the mechanism to halt on acceptance or on rejection, then we define $Z_t = s_t (\mu_1 - \eta_1 + \eta_{2,t}).$

Clearly, this is equivalent to the description above when $s_t = +1$. Subtle adjustments  are required, but otherwise the mechanism is identical, with the caveat that it must halt when it rejects some stream element. On implementing this modification, the analyst is free to pick whichever variant offers us less error, and, thus, we are able to show that the rejection threshold equals
\[ \bar{p_t} \;:=\; \max\!\left\{\frac{\pstar_t}{\gamma\,\Lambda_t},\;\; 1 - \gamma\,\Lambda_t\,(1-\pstar_t)\right\}. \]

Furthermore, since the variable $Z_t$ is negative with high probability for $s_t = -1$, this implies that the Poisson mean parameter is deflated by up to a factor of $\Lambda_t$, instead of being inflated. In other words, $\rho_t = \Theta \left(\tfrac{\ln(t/\beta)}{\pstar (\gamma-1)^2}\right)$ leading (with a failure probability of at most $\beta/6$) to a high probability sample complexity bound of
\[ N_t =  \frac{\ln(t/\beta)}{(\gamma-1)^2}\max\left(\frac{\Lambda_t}{\pstar_t},\frac{1}{1-\pstar_t}\right).\]

The first term of the $\max$ expression occurs on the direct variant, and the latter occurs on the flipped variant. There is \emph{no inflation} in the sample complexity via the noise penalty term $\Lambda_t$ in the flipped variant.

\medskip \noindent \textbf{Purification via randomized response.} Thus far we have only reasoned about pure DP inputs. To deal with mechanisms $\calM_t(X_t)$ that only fulfill a weaker $(\eps_t,\delta_t)$-DP guarantee, we appeal to a black box modification based on randomized response. Each evaluation of $\calM_t(X_t)$ is flipped with probability $\phi = \delta_t/(e^{\eps_t} - 1 + 2\delta_t) \leq \delta_t/\eps_t$. We denote this ``$\phi-$smoothed'' mechanism by $\widetilde{\calM}_t$, and its success probability by $\widetilde{p}_t$. It is not hard to see that $|p_t - \widetilde{p}_t| \leq \delta_t/\eps_t$; in other words, the perturbation to the success probabilities is vanishingly small.  

The advantage of this $\phi$-smoothed mechanism is that it is \emph{pure} $\eps_t$-DP. The key idea here is that if a mechanism with co-domain $\{+1,-1\}$ is $(\eps_t,\delta_t)$-DP, then in some sense the worst privacy loss occurs when for  $v\in \{+1,-1\}$, $\Pr[\calM_t(X_t) = v] = \delta_t$, but $\Pr[\calM_t(X'_t) = v] = 0$. Clearly such a mechanism cannot be $\eps$-DP for any value of $\eps$. However, if we increase the probability of both outputs by $\delta_t$, then the ratio of these two outputs changes from being undefined to $2$, and the effective privacy loss as measured by the multiplicative privacy parameter changes from $\infty$ to $\ln 2$. Refining this outline, we get the stated claim.

We can now run \Cref{alg:GTM} on $\widetilde{\calM}_t(X_t)$, and the bound on $|p_t - \widetilde{p}_t|$ implies that essentially the same acceptance and rejection thresholds hold, perturbed by at most a value of $\delta_t/\eps_t$. This technique of using randomized response as post-processing to purify approximate-DP mechanisms with finite co-domains is attributed to folklore by \cite{DBLP:journals/corr/abs-2211-11189}. The term purification was coined by \cite{DBLP:journals/corr/abs-2503-21071} to describe the generation of pure DP outputs via randomized post-processing, and they give further generalizations and applications of this approach.

\subsection{Discussion on Private Testing vs Selection} \label{sec:discussion}
A major motivation behind the introduction of generalized private selection and testing in \cite{DBLP:conf/stoc/0001T19} was private training of machine learning models. In this section, we discuss potential benefits of mechanisms like \GTM for generalized private testing in this context.

\medskip \noindent \textbf{Implication I: Hyperparameter Optimization.}
Hyperparameter optimization (HPO) in private learning is the process of selecting parameters (like noise level, clipping norm, batch size, and learning rate) that balance model accuracy with privacy guarantees. Since finding these parameters itself can leak information, it must be done via privacy-preserving methods. Solving HPO was one of the motivations for generalized private selection and testing in \cite{DBLP:conf/stoc/0001T19}. 

While subsequent works~\cite{papernot2021hyperparameter,ghaziprivate} used generalized private selection instead of testing for this problem, they mention that algorithms for generalized private testing have potential advantages over selection, but point out that if the acceptance threshold were set too high, one could evaluate large sets of hyperparameter choices without identifying a useful candidate, and conversely if we set it too low, then we would get poor quality outputs (see \cite[Appendix C]{papernot2021hyperparameter}).

\Cref{alg:GTM} addresses this challenge in part by allowing the analyst to set the acceptance threshold $\pstar_t$ adaptively in every step. This suggests a way to deal with the challenge mentioned by \cite{papernot2021hyperparameter}: initialize $\pstar_t$ with a restrictive value at the beginning of the stream and make it gradually more permissive as $t$ increases. Further, in the ex-post setting, if we commit to an ex-post privacy loss guarantee that scales logarithmically in the index of hyperparameter value output by the \GTM, then we are guaranteed a rejection threshold $\bar{p}_t = \pstar_t/\gamma\Lambda_0$ in the notation of \cref{thm:intro_ex_post} and \cref{thm:intro_upper_bound}, where $1/\gamma\Lambda_0$ does not depend on the number of hyperparameter values tested. Further, the privacy loss incurred is pure and constant if no hyperparameter value is accepted even when testing on approximate DP mechanisms $\calM_t$.

\medskip \noindent
\textbf{Implication II: Reproducibility in Private Machine Learning.} 
\cite{ghaziprivate} suggested the following approach to use generalized private selection in practical private training: instead of trying to find the best hyperparameter directly, one privately trains models using all possible choices of hyperparameters and \emph{selects} the output model that has the best generalization error (on a given set of test data). They prove that with sufficiently many repetitions per hyperparameter, one is guaranteed to output a model that is competitive with the best model generated (see \Cref{sec:ghaziprivate} for a complete discussion). However, this does not imply that the selection identified the best set of hyperparameters. Typically, there are a large number of hyperparameters with a small probability of generating a good trained model, and a few set of hyperparameters with a high probability of generating a good model. For {\em reproducibility} in learning (and considering intrinsic stochasticity of all learning algorithms), one would ideally like to use the hyperparameters of the latter type.

The guarantee of \cite{ghaziprivate} (\cref{thm:GhaziAccuracy}) cannot circumvent the hardness of generalized private testing and identifying a good model; we formalize this via an idealized toy example in \Cref{cor:GhaziAccuracyLimit}. We find that regardless of the choice of privacy parameters, when we measure the quality of their private selection algorithm for \emph{testing}, their approach incurs a $\approx \beta/T$ loss in accuracy over a set of $T$ hyperparameter values. At a high level, although correlated random dropping works well for generalized private selection, it seems to always incur a $T/\beta$ factor in the ratio between acceptance and rejection thresholds when used for generalized private testing.
\section{Preliminaries}

We collect some common notational choices here for reference throughout this paper.

\begin{defn}[Notation]\label{def:Notation}
    We make the following definitions of convenience.
    \begin{enumerate}
        \item We define \[
            \omega_t \;=\; \frac{1}{\zeta \cdot t \cdot \ln^3(t+2)}, \qquad \zeta \;=\; \sum_{j=1}^{\infty} \frac{1}{j\,\ln^3(j+2)}.
            \]
        With this definition, $\sum_{t=1}^\infty \omega_t = 1$. One can check that $\zeta\in(1,5/4)$, and that therefore $1/\omega_t \in (t\ln^3 (t+2),2t\ln^3 (t+2))$. \item We define $ \beta_t \;:=\; \omega_t \beta$,
        whence we have that $\sum_{t=1}^\infty \beta_t = \beta$, and $\beta_t \geq \frac{\beta}{\zeta t \ln^3 (t+2)}$.
        These expressions will be used heavily in all failure probability accounting.
    \end{enumerate}
\end{defn}

We recall the definitions and some properties of the distributions used in the sequel.

\begin{defn}[Distributions]\label{def:Distributions}
We use the following standard distributions.
\begin{enumerate}
\item \textbf{Bernoulli.} For $p \in [0,1]$, a random variable $B \sim \Ber(p)$ takes values in $\{0,1\}$ with $\Pr[B = 1] = p$ and $\Pr[B = 0] = 1 - p$.

\item \textbf{Exponential.} For $\tau > 0$, a random variable $\eta \sim \Exp(\tau)$ has support $[0,\infty)$ and density $f(z) = \frac{1}{\tau}\,e^{-z/\tau}$ for $z \ge 0$.

\item \textbf{Laplace.} For $\tau > 0$, a random variable $\nu \sim \Lap(\tau)$ has support $\R$ and density $f(z) = \frac{1}{2\tau}e^{-|z|/(\tau)}$ for $z\in\R$.

\item \textbf{Poisson.} For $\lambda \ge 0$, a random variable $M \sim \Po(\lambda)$ takes values in $\mathbb{Z}_{\ge 0}$ with mass function $\Pr[M = k] = e^{-\lambda}\,\frac{\lambda^k}{k!}$ for $k \in \mathbb{Z}_{\ge 0}$.

\item \textbf{Pareto.} For $\sigma > 0$, a random variable $w \sim \Par(\sigma)$
has support $[1,\infty)$ with $\Pr[w > u] = u^{-\sigma}$ for $u \geq 1$
and density $f(u) = \sigma\, u^{-\sigma-1}$ for $u \geq 1$.
\end{enumerate}
\end{defn}

\begin{defn}[Differential Privacy fundamentals]
    \begin{enumerate}
        \item We define the \emph{sensitivity} of a function $f:\calX\to\R$ by $\Delta f := \max_{d(X,X')=1} |f(X) - f(X')|$, where $d(X,X')=1$ is shorthand for $X$ and $X'$ are neighboring.
        \item Given a sensitivity-$1$ query $f:\calX\to\R$, the Laplace mechanism (\cite{DMNS}) releases $f(X) + \nu$ for $\nu\sim\Lap(1/\eps)$. The Laplace mechanism is $\eps$-DP.
    \end{enumerate}
\end{defn}

\paragraph{Poisson mean thresholds.}
For $x \ge 0$ and $\alpha \in (0,1)$, define the \emph{left} and \emph{right mean thresholds}
\begin{align}
\lamL{x}{\alpha} &:= x + \ln(1/\alpha) + \sqrt{2x\ln(1/\alpha) + \ln^2(1/\alpha)}, \label{eq:lamL_def} \\
\lamR{x}{\alpha} &:= x + 1 - \sqrt{2(x+1)\ln(1/\alpha)}. \label{eq:lamR_def}
\end{align}
These thresholds are derived from Poisson Chernoff bounds (\Cref{lem:PoissonChernoffBounds}) and provide sufficient conditions for tail control: if $M \sim \Po(\lambda)$, then $\lambda \ge \lamL{x}{\alpha}$ guarantees $\Pr[M \le x] \le \alpha$, and $\lambda \le \lamR{x}{\alpha}$ guarantees $\Pr[M \ge x+1] \le \alpha$ (under a mild validity condition that $x$ not be too small). The thresholds may not be tight with respect to the exact Poisson tail probabilities (see \Cref{lem:PoissonFacts}, item~3). 

We define the mean thresholds for real valued $x$ as opposed to just integer valued so as to be able to exploit the algebraic inversion property. See \Cref{lem:PoissonFacts}, item~4.

\begin{lem}[Poisson and exponential facts]\label{lem:PoissonFacts}
The following properties are used throughout.
\begin{enumerate}
\item \textbf{Additivity.}
If $M_1 \sim \Po(\lambda_1)$ and $M_2 \sim \Po(\lambda_2)$ are independent, then $M_1 + M_2 \sim \Po(\lambda_1 + \lambda_2)$.

\item \textbf{Thinning.}
Let $M \sim \Po(\lambda)$, and conditional on $M$, let $B_1, \dots, B_M$ be i.i.d.\ $\Ber(p)$ random variables, independent of $M$. Then $K := \sum_{i=1}^{M} B_i \sim \Po(\lambda p)$.

\item \textbf{Tail bounds.}
\label{item:tail_bound_Poisson}
Let $M \sim \Po(\lambda)$ and fix $\alpha \in (0,1)$ and $c \in \mathbb{Z}_{\ge 0}$.
\begin{enumerate}
\item \emph{(Left tail.)} If $\lambda \ge \lamL{c}{\alpha}$, then $\Pr[M \le c] \le \alpha$.
\label{item:tail_bound_Poisson_left}

\item \emph{(Right tail.)} If $c + 1 \ge 2\ln(1/\alpha)$ and $\lambda \le \lamR{c}{\alpha}$, then $\Pr[M \ge c + 1] \le \alpha$.
\label{item:tail_bound_Poisson_right}

\end{enumerate}

\item \textbf{Algebraic inversion.}
Fix $\alpha \in (0,1)$ and let $L = \ln(1/\alpha)$. For $\lambda \ge 0$ and $x \ge 0$ with $x + 1 > 2L$,
\[
\lambda \le \lamR{x}{\alpha} \quad\iff\quad x + 1 \ge \lamL{\lambda}{\alpha}.
\]

\item \textbf{Exponential density shift.}
Let $\eta \sim \Exp(\tau)$ with density $f(z) = \frac{1}{\tau}\,e^{-z/\tau}$ for $z \ge 0$. For any $z \ge 0$ and $\Delta \in \mathbb{R}$ with $z + \Delta \ge 0$,
\[
f(z + \Delta) = e^{-\Delta/\tau}\,f(z).
\]
\end{enumerate}
\end{lem}

\begin{proof}
\textbf{Items 1--2.} Standard; see, e.g., \cite{DBLP:books/daglib/0012859}.

\medskip
\noindent\textbf{Item 3.}
Write $L = \ln(1/\alpha)$.

\emph{Left tail.}
We apply \Cref{lem:PoissonChernoffBounds} (lower tail) with shift $u = \lambda - c$. 
The assumption $\lambda \ge \lamL{c}{\alpha} \ge c$ ensures $u \ge 0$, so the bound gives
\[
\Pr[M \le c] = \Pr[M \le \lambda - u] \le \exp\!\left(-\frac{(\lambda - c)^2}{2\lambda}\right).
\]
We require this to be at most $\alpha$, i.e., $(\lambda - c)^2 \ge 2\lambda L$.
Expanding the left side and rearranging by powers of $\lambda$ yields the quadratic inequality
\[
\lambda^2 - 2(c + L)\lambda + c^2 \ge 0.
\]
The roots of the corresponding equation are
\[
\lambda = c + L \pm \sqrt{(c + L)^2 - c^2} = c + L \pm \sqrt{2cL + L^2}.
\]
Since $u \ge 0$ requires $\lambda \ge c$, we need $\lambda$ at or above the larger root.
Recognizing the larger root as $\lamL{c}{\alpha}$, the bound $\Pr[M \le c] \le \alpha$ holds for all $\lambda \ge \lamL{c}{\alpha}$.

\emph{Right tail.}
We apply \Cref{lem:PoissonChernoffBounds} (upper tail) with shift $u = c + 1 - \lambda$ for $u \ge 0$. 
Note that the denominator in the Chernoff exponent becomes $2(\lambda + u) = 2(c + 1)$, giving
\[
\Pr[M \ge c + 1] = \Pr[M \ge \lambda + u] \le \exp\!\left(-\frac{(c + 1 - \lambda)^2}{2(c + 1)}\right).
\]
We require this to be at most $\alpha$, i.e., $(c + 1 - \lambda)^2 \ge 2(c + 1)L$ as $L = \ln(1/\alpha)$.
We first show that this implies $c + 1 \ge 2L$.
The Chernoff upper tail bound requires the shift $u = c + 1 - \lambda$ to be nonneg\-ative, so we must have $\lambda \le c + 1$.
If $c + 1 < 2L$, then for every $\lambda \in [0,\, c + 1]$ we have $(c + 1 - \lambda)^2 \le (c+1)^2 < 2(c+1)L$, so the required inequality $(c + 1 - \lambda)^2 \ge 2(c+1)L$ cannot hold.
The condition $c + 1 \ge 2L$ is therefore necessary.

Under this condition, we take the square root of both sides.
Since the shift $u = c + 1 - \lambda$ must be nonnegative,
the left side $c + 1 - \lambda$ is nonnegative, and we obtain
\[
c + 1 - \lambda \ge \sqrt{2(c+1)L}.
\]
Isolating $\lambda$ yields $\lambda \le c + 1 - \sqrt{2(c+1)L} = \lamR{c}{\alpha}$.
Note that this automatically enforces $\lambda \le c + 1$, so the nonnegative shift requirement $u \ge 0$ is satisfied.

\medskip
\noindent\textbf{Item 4.}
Let $L = \ln(1/\alpha)$ and set $y = x + 1$, so the validity condition is $y > 2L$. Starting from $\lambda \le \lamR{x}{\alpha}$:
\[
\lambda \le y - \sqrt{2yL}
\quad\iff\quad
y - \lambda \ge \sqrt{2yL}.
\]
Since $\sqrt{2yL} \ge 0$, this implies $y \ge \lambda$. So squaring $y - \lambda \ge \sqrt{2yL}$ gives $(y - \lambda)^2 \ge 2yL$, i.e., $y^2 - 2(\lambda + L)y + \lambda^2 \ge 0$. Because $y \ge \lambda$, the variable $y$ lies at or above the larger root:
\[
y \ge \lambda + L + \sqrt{2\lambda L + L^2} = \lamL{\lambda}{\alpha}.
\]
Substituting $y = x + 1$ yields $x + 1 \ge \lamL{\lambda}{\alpha}$. All steps are reversible (the implication $y \ge \lambda$ is enforced by $y - \lambda \ge \sqrt{2yL} \ge 0$), establishing the equivalence.

\medskip
\noindent\textbf{Item 5.}
Direct computation: $f(z + \Delta) = \frac{1}{\tau}\,e^{-(z+\Delta)/\tau} = e^{-\Delta/\tau} \cdot \frac{1}{\tau}\,e^{-z/\tau} = e^{-\Delta/\tau}\,f(z)$.
\end{proof}

\begin{lem}[Poisson overhead control]\label{lem:PoissonOverhead}
Fix error parameters $\alpha, \beta \in (0,1)$ and a target ratio $\gamma \in (1,2]$.
Let $A = \ln(1/\alpha)$, $B=\ln(1/\beta)$ and $c = \max(\lfloor 2\max (A,B)\rfloor,\lceil r^2 \rceil)$, where
\[ r = \frac{\sqrt{2A} + \gamma\sqrt{2B} + \sqrt{(\sqrt{2A}+\gamma\sqrt{2B})^2 + 8(\gamma-1)A}}{2(\gamma-1)}.\]
Then, the following statements hold:
\begin{enumerate}
\item $\lamL{c}{\alpha} \leq \gamma\,\lamR{c}{\beta}$,
\item $c + 1 > 2\max\!\big(\ln(1/\alpha),\;\ln(1/\beta)\big)$,
\item $c = O\!\big((\ln(1/\alpha) + \ln(1/\beta))/(\gamma-1)^2\big)$.
\end{enumerate}
\end{lem}

\begin{proof}
Setting $m = \sqrt{c+1}$, the right threshold expands as
\[ \lamR{c}{\beta} = c + 1 - \sqrt{2(c+1)B} = m^2 - m\sqrt{2B}.\]
For the left threshold, bounding $c < m^2$ and applying
$\sqrt{2cA + A^2} < \sqrt{2m^2 A + A^2} \leq m\sqrt{2A} + A$ gives
\[\lamL{c}{\alpha} = c + A + \sqrt{2cA + A^2} < m^2 + 2A + m\sqrt{2A}.\]
The condition $\lamL{c}{\alpha} \leq \gamma\,\lamR{c}{\beta}$ is therefore implied by
\[
  m^2 + 2A + m\sqrt{2A} \;\leq\; \gamma\!\left(m^2 - m\sqrt{2B}\right),
\]
which rearranges to
\begin{equation}\label{eq:overhead_quadratic}
  (\gamma - 1)\,m^2
  - \bigl(\sqrt{2A} + \gamma\sqrt{2B}\bigr)\,m
  - 2A \;\geq\; 0.
\end{equation}
Since $\gamma > 1$, this quadratic inequality holds for all $m$ at or above its larger positive root
\[
  r \;:=\;
  \frac{\sqrt{2A} + \gamma\sqrt{2B}
    + \sqrt{\bigl(\sqrt{2A} + \gamma\sqrt{2B}\bigr)^2 + 8(\gamma-1)A}}
  {2(\gamma-1)}.
\]
Define
\[
  c \;:=\; \max\!\Big(\big\lfloor 2\max(A,B) \big\rfloor,\;\big\lceil r^2 \big\rceil\Big).
\]

\medskip
\noindent\emph{Item~1.}
Since $c \geq \lceil r^2 \rceil \geq r^2$, we have $m = \sqrt{c+1} > r$,
so \eqref{eq:overhead_quadratic} holds and the gap condition follows.

\medskip
\noindent\emph{Item~2.}
Since $\lfloor y \rfloor + 1 > y$ for all $y$, we have
$c + 1 \geq \lfloor 2\max(A,B) \rfloor + 1 > 2\max(A,B)$.

\medskip
\noindent\emph{Item~3.}
Write $a = \sqrt{2A} + \gamma\sqrt{2B}$ for the linear coefficient in \eqref{eq:overhead_quadratic}.
Using $\sqrt{a^2 + b} \leq a + \sqrt{b}$ for $a \geq 0$, the numerator of $r$ satisfies
\[
  a + \sqrt{a^2 + 8(\gamma-1)A}
  \;\leq\; 2a + \sqrt{8(\gamma-1)A}.
\]
Since $\gamma \leq 2$, we have $a \leq \sqrt{2A} + 2\sqrt{2B}$ and
$\sqrt{8(\gamma-1)A} \leq \sqrt{8A}$, so the numerator is
$O(\sqrt{A} + \sqrt{B})$. Therefore
$r = O((\sqrt{A} + \sqrt{B})/(\gamma-1))$ and
$r^2 = O((A + B)/(\gamma-1)^2)$.
Since $\lfloor 2\max(A,B) \rfloor \leq 2(A + B)$,
the conclusion $c = O((A + B)/(\gamma-1)^2)$ follows.
\end{proof}

\begin{lem}[Pareto facts]\label{lem:ParetoFacts}
The following properties are used throughout.
\begin{enumerate}
\item \textbf{Exponential to Pareto transformation.}
If\, $\eta \sim \Exp(\tau)$, then $e^{\eta/(\sigma\tau)} \sim \Par(\sigma)$.
 
\item \textbf{Moments.}
Let $w \sim \Par(\sigma)$. For any $k > 0$,
$\Ex[w^{-k}] = \sigma/(\sigma + k)$.
For any $0 < k < \sigma$,
$\Ex[w^k] = \sigma/(\sigma - k)$.
If $k \geq \sigma$, then $\Ex[w^k] = \infty$.
 
\item \textbf{Truncated moments.}
Let $w \sim \Par(\sigma)$ with $\sigma > 1$, and define $\tilde{w} := \min(w, W)$ for $W \geq 1$. Then:
\begin{enumerate}
\item $\displaystyle\Ex[\tilde{w}] = \frac{\sigma}{\sigma - 1} - \frac{W^{1-\sigma}}{\sigma - 1}
\leq \frac{\sigma}{\sigma-1}$.
 
\item $\displaystyle\Ex[\tilde{w}^2] =
  \begin{cases}
    \displaystyle\frac{2\,W^{2-\sigma} - \sigma}{2-\sigma}
    & \text{if } \sigma \neq 2, \\[8pt]
    2\ln W + 1 & \text{if } \sigma = 2.
  \end{cases}$
 
In particular: when $1 < \sigma < 2$, $\Ex[\tilde{w}^2] \leq \frac{2}{2-\sigma}\,W^{2-\sigma}$;
when $\sigma = 2$, $\Ex[\tilde{w}^2] = 2\ln W + 1$;
when $\sigma > 2$, $\Ex[\tilde{w}^2] \leq \frac{\sigma}{\sigma - 2}$.
\end{enumerate}
\end{enumerate}
\end{lem}
 
\begin{proof}
\textbf{Item 1.}
For $u \geq 1$, $\Pr[e^{\eta/(\sigma\tau)} > u] = \Pr[\eta > \sigma\tau\ln u] = e^{-\sigma\ln u} = u^{-\sigma}$.
 
\medskip
\noindent\textbf{Item 2.}
By direct integration,
$\Ex[w^k] = \sigma \int_1^{\infty} u^{k-\sigma-1}\,\mathsf{d}u = \sigma/(\sigma-k)$
when $k < \sigma$, and the integral diverges when $k \geq \sigma$.
Similarly,
$\Ex[w^{-k}] = \sigma \int_1^{\infty} u^{-k-\sigma-1}\,\mathsf{d}u = \sigma/(\sigma+k)$
for all $k > 0$.
 
\medskip
\noindent\textbf{Item 3(a).}
$\Ex[\tilde{w}] = \int_1^{W} u \cdot \sigma u^{-\sigma-1}\,\mathsf{d}u + W \cdot W^{-\sigma}$.
The integral equals $\frac{\sigma}{1-\sigma}(W^{1-\sigma} - 1)$ for $\sigma > 1$.
Adding $W^{1-\sigma}$ and simplifying gives
$\Ex[\tilde{w}] = \frac{\sigma}{\sigma-1} - \frac{W^{1-\sigma}}{\sigma-1}$.
Since $W \geq 1$ and $\sigma > 1$, the subtracted term is nonnegative.
 
\medskip
\noindent\textbf{Item 3(b).}
$\Ex[\tilde{w}^2] = \sigma \int_1^{W} u^{1-\sigma}\,\mathsf{d}u + W^{2-\sigma}$.
For $\sigma \neq 2$, the integral equals $\frac{\sigma}{2-\sigma}(W^{2-\sigma} - 1)$,
and adding $W^{2-\sigma}$ gives $\frac{2\,W^{2-\sigma} - \sigma}{2-\sigma}$.
For $\sigma = 2$, the integral equals $2\ln W$, and adding $1$ gives $2\ln W + 1$.
 
For the bounds: when $1 < \sigma < 2$, $2 - \sigma > 0$,
so $\frac{2W^{2-\sigma} - \sigma}{2-\sigma} \leq \frac{2W^{2-\sigma}}{2-\sigma}$.
When $\sigma > 2$, rewrite $\frac{2W^{2-\sigma} - \sigma}{2-\sigma}
= \frac{\sigma - 2W^{2-\sigma}}{\sigma - 2} \leq \frac{\sigma}{\sigma-2}$
since $W^{2-\sigma} \geq 0$.
\end{proof}

We will find it convenient to reason about the accuracy of generalized private testers in terms of their \emph{per-step type I and type II} errors.

\begin{defn}[Accuracy]\label{def:accuracy}
    Consider a generalized private tester that for some given failure probability $\beta$ and target thresholds $(\pstar_t)_{t\geq1}$ achieves rejection thresholds $(\pbar_t)_{t\geq1}$.
    A \emph{Type~I error} (false halt) at step $t$ is the event that the tester halts and $s_t\, p_t \leq s_t\, \bar{p}_t$. 
    A \emph{Type~II error} (missed halt) at step $t$ is the event that the tester continues and $s_t\, p_t \geq s_t\, \pstar_t$. 
    Then we define the following notions of accuracy.
    \begin{enumerate}
        \item \emph{Global accuracy.}
          The tester is $\beta$-globally accurate if,
          with probability at least $1-\beta$,
          no Type~I or Type~II error occurs at any step.
        \item \emph{Per-step accuracy.}
          The tester has $(\beta_t^{\tI}, \beta_t^{\tII})$ per-step accuracy
          at step $t$ if the Type~I and Type~II error probabilities
          at step $t$ are at most $\beta_t^{\tI}$ and $\beta_t^{\tII}$, respectively.
    \end{enumerate}
    For the given choice of thresholds, the generalized private tester considered in this definition is $\beta$-globally accurate.
\end{defn}
\section{The Generalized Thresholding Mechanism}\label{sec:algorithm}

In this section, we introduce the Generalized Thresholding Mechanism (\Cref{alg:GTM}). The algorithm monitors a sequence of mechanisms $\calM_t : \calX \to \{-1,+1\}$, each evaluated on a dataset $X_t$ with an input privacy promise $\eps_t \geq 0$, and halts when it believes that the success probability $p_t := \Pr[\calM_t(X_t) = +1]$ lies in the target regime.

The sequence is provided by an \emph{adaptive analyst}. 
At the beginning of the stream, the analyst fixes (a) a value $\tau_1$ that determines the scale parameter of a privatizing noise value $\eta_1 \sim \Exp(\tau_1)$ that is shared across the stream, and (b) a \emph{calibration parameter} $\mu_1$ which will be a high probability upper bound on the magnitude of $\eta_1$. \emph{At each step $t$}, after observing the prior outputs $a_1, \dots, a_{t-1}$, the analyst selects (c) the mechanism $\calM_t$, (d) dataset $X_t$, (e) privacy parameter $\eps_t$, subject only to the constraint that $\calM_t$ is $\eps_t$-differentially private,
and (f) a noise scale $\tau_{2,t}$. 
The mechanism will choose
various per-step tuning parameters, namely (1) a  threshold $c_t$, (2) a base rate $\rho_t$, and (3) a direction $s_t$. The privacy guarantee (\Cref{lem:GeneralizedPrivacy}) holds against all such adaptive choices; it is parameterized by $\tau_1$  and by the per-step noise scales $\tau_{2,t}$ (which directly determines the privacy cost at halting time, regardless of the values of $\eps_1,\dots,\eps_t$). 
However, different values of $s_t$, $\rho_t$, and $c_t$ yield qualitatively different accuracy guarantees in terms of the rejection threshold $\bar{p}$ and we give below values that minimize accuracy loss.

\begin{algorithm}[t]
\caption{Generalized Thresholding Mechanism}
\label{alg:GTM}
\SetAlgoLined
\SetKwProg{Proc}{procedure}{:}{}

\Proc{\textnormal{GTM.Init}$(\tau_1,\, \mu_1)$}{
  \KwIn{Shared noise scale $\tau_1 > 0$; calibration offset $\mu_1 > 0$.}
  Draw $\eta_1 \sim \Exp(\tau_1)$\;
  Store $\eta_1$ and $\mu_1$\;
}

\BlankLine

\Proc{\textnormal{GTM.Step}$(\calM_t,\, X_t,\, \eps_t,\, \tau_{2,t},\, c_t,\, \rho_t,\, s_t)$}{
  \KwIn{privacy parameter $\eps_t \geq 0$;
    Mechanism $\calM_t : \calX \to \{-1,+1\}$ ($\eps_t$-DP);
    dataset $X_t \in \calX$;
    per-step noise scale $\tau_{2,t} > 0$;
    threshold $c_t \in \mathbb{Z}_{\geq 0}$;
    base rate $\rho_t > 0$;
    direction $s_t \in \{-1,+1\}$.}
  \KwOut{$a_t \in \{-s_t,\, s_t\}$, where $a_t = s_t$ denotes halting
    and $a_t = -s_t$ denotes continuation.}
  \BlankLine
  Draw $\eta_{2,t} \sim \Exp(\tau_{2,t})$\;
  $Z_t \gets s_t\bigl(\mu_1 - \eta_1 + \eta_{2,t}\bigr)$\;
  $\lambda_t \gets \rho_t \cdot \exp\!\left(\eps_t\, Z_t\right)$\;
  Draw $N_t \sim \Po(\lambda_t)$\;
  $K_t \gets 0$\;
  \For{$i = 1, \ldots, N_t$}{
    Draw $y_{t,i} \sim \calM_t(X_t)$\;
    \lIf{$y_{t,i} = +1$}{$K_t \gets K_t + 1$}
  }
  \lIf{$s_t \cdot (K_t - c_t - \tfrac{1}{2}) > 0$}{\textbf{Output} $a_t \gets s_t$  and \textbf{halt}}
  \lElse{\textbf{Output} $a_t \gets -s_t$}
}
\end{algorithm}

The algorithm uses two sources of randomness for privacy. A shared noise value $\eta_1 \sim \Exp(\tau_1)$ is drawn once at initialization and reused at every step. At each step $t$, a fresh noise value $\eta_{2,t} \sim \Exp(\tau_{2,t})$ is drawn independently. They are combined with the calibration offset $\mu_1 > 0$  and the direction $s_t \in \{-1,+1\}$ as follows to form the noise term
    \[ Z_t = s_t \cdot (\mu_1 - \eta_1 + \eta_{2,t}). \]
\sloppy The noise $Z_t$ controls the number of mechanism evaluations at step $t$. The algorithm makes $N_t$ i.i.d. evaluations of  $\calM_t(X_t)$, where
\[ N_t \sim \Po(\rho_t \exp(\eps_t Z_t)).\]  
The exponential scaling $\exp(\eps_t Z_t)$ scales the perturbation with the input privacy parameter $\eps_t$; when $\eps_t = 0$, the mechanism already satisfies $0$-DP and no perturbation is needed, so  $N_t$ reduces to the base rate $\rho_t$. 

The algorithm counts the number $K_t$ of $+1$ outcomes among the $N_t$-many evaluations of $\calM_t(X_t)$. The direction parameter $s_t$ determines two modes of operation. 
\begin{enumerate}
    \item When $s_t = +1$ (\emph{above threshold}), the algorithm halts if $K_t \geq c_t + 1$, indicating that $p_t$ is large. 
    \item When $s_t = -1$ (\emph{below threshold}), the algorithm halts if $K_t \leq c_t$, indicating that $p_t$ is small.
\end{enumerate} 
Conditioned on $\eta_1 \leq \mu_1$, the signed quantity $s_t Z_t = \mu_1 - \eta_1 + \eta_{2,t}$ is nonnegative. Therefore, $Z_t$ is positive when $s_t = +1$ and negative when $s_t = -1$. This means that above threshold ($s_t = +1$) never deflates $N_t$ below $\rho_t$, while the below threshold ($s_t = -1$) never inflates it above $\rho_t$.

\begin{defn}[Execution variables and continuation probability]\label{def:ExecutionVariables}
Fix a step $t \geq 1$. Let the parameters up to step $t$ be as specified in \Cref{alg:GTM}. Write $p_t := \Pr[\calM_t(X_t) = +1]$.

\Cref{alg:GTM} induces the following random variables at step $t$:
\begin{itemize}
\item \emph{Aggregated noise and conditional Poisson mean.} Define
\[
Z_t := s_t(\mu_1 - \eta_1 + \eta_{2,t}), \qquad \lambda_t(z) := \rho_t \exp(\eps_t\, z) \quad \text{for } z \in \mathbb{R},
\]
and write $\lambda_t := \lambda_t(Z_t) = \rho_t \exp(\eps_t Z_t)$ for the \emph{realized Poisson mean}. The algorithm draws $N_t \sim \Po(\lambda_t)$ evaluations of $\calM_t(X_t)$.

\item \emph{Count and thinning.} Let $K_t := \sum_{i=1}^{N_t} \mathbf{1}\{y_{t,i} = +1\}$ denote the number of $+1$ outcomes. By Poisson thinning (\Cref{lem:PoissonFacts}), conditional on $\lambda_t$,
\[
K_t \;\sim\; \Po(\lambda_t\, p_t).
\]

\item \emph{Continuation probability.} Define $g_t(z)$ to be the probability that step $t$ produces the non-halting output $a_t = -s_t$, given $Z_t = z$:
\[
g_t(z) := \Pr\!\big[a_t = -s_t \;\big|\; Z_t = z\big].
\]
For a neighboring dataset $X'_t$, write $p'_t := \Pr[\calM_t(X'_t) = +1]$, and let $g'_t(z)$ denote the corresponding continuation probability when $\calM_t$ is evaluated on $X'_t$.
\end{itemize}
\end{defn}

\subsection{Privacy analysis}\label{subsec:privacy}

\begin{lem}[Poisson representation and pointwise shifts]\label{lem:PointwiseShifts}
Fix a step $t \geq 1$ and neighboring datasets $X_t, X'_t$.
Let $g_t(\cdot)$ and $g'_t(\cdot)$ be the continuation probabilities from \Cref{def:ExecutionVariables}, and define the Poisson cumulative distribution function $F(c;\lambda) := \Pr[M \leq c]$ where $M \sim \Po(\lambda)$.
Then for every $z \in \mathbb{R}$, the following hold:

\begin{enumerate}
\item \emph{(Poisson representation.)}
  \[
  g_t(z) =
  \begin{cases}
    F\!\big(c_t;\,\lambda_t(z)\,p_t\big), & \text{if } s_t = +1,\\[1ex]
    1 - F\!\big(c_t;\,\lambda_t(z)\,p_t\big), & \text{if } s_t = -1,
  \end{cases}
  \]
  and an identical representation holds for $g'_t(z)$ with $p'_t$ in place of $p_t$.

\item \emph{(Shift identities for the conditional mean.)}
  \[
  \lambda_t(z - s_t) = e^{-s_t \eps_t}\,\lambda_t(z)
  \qquad\text{and}\qquad
  \lambda_t(z + s_t) = e^{s_t \eps_t}\,\lambda_t(z).
  \]

\item  \label{item:shift_inequalities}\emph{(Shift inequalities for continuation probability.)}
  \[
  g_t(z) \leq g'_t(z - s_t)
  \qquad\text{and}\qquad
  1 - g_t(z) \leq 1 - g'_t(z + s_t).
  \]
\end{enumerate}
\end{lem}

\begin{proof}
Fix $z \in \mathbb{R}$.

\medskip
\noindent\textbf{Part (1).}
Condition on the event $\{Z_t = z\}$. Under this conditioning, the algorithm draws $N_t \sim \Po(\lambda_t(z))$ evaluations of $\calM_t(X_t)$, and counts $K_t = \sum_{i=1}^{N_t} \mathbf{1}\{y_{t,i} = +1\}$. By Poisson thinning (\Cref{lem:PoissonFacts}),
\[
K_t \sim \Po\!\big(\lambda_t(z)\,p_t\big).
\]
The step outputs $-s_t$ exactly when $s_t \cdot (K_t - c_t - \tfrac{1}{2}) \leq 0$.
For $s_t = +1$, this condition is $K_t \leq c_t$ as $c_t$ is an integer, yielding $g_t(z) = F(c_t;\,\lambda_t(z)\,p_t)$.
For $s_t = -1$, this condition is $K_t \geq c_t + 1$, yielding $g_t(z) = 1 - F(c_t;\,\lambda_t(z)\,p_t)$.
The derivation for $g'_t(z)$ is identical with $p'_t$ replacing $p_t$.

\medskip
\noindent\textbf{Part (2).}
This follows directly from the definition of $\lambda_t(\cdot)$:
\[
\lambda_t(z \pm s_t)
= \rho_t \exp\!\big(\eps_t(z \pm s_t)\big)
= e^{\pm s_t \eps_t}\,\lambda_t(z).
\]

\medskip
\noindent\textbf{Part (3).}
We establish the shift inequalities by considering the two cases of $s_t$.

\smallskip
\noindent\emph{Case $s_t = +1$:}
By part~(1), $g_t(z) = F\!\big(c_t;\,\lambda_t(z)\,p_t\big)$ and $g'_t(z-1) = F\!\big(c_t;\,\lambda_t(z-1)\,p'_t\big)$.
By $\eps_t$-differential privacy, $p_t \geq e^{-\eps_t} p'_t$. By part~(2), $\lambda_t(z) = e^{\eps_t}\lambda_t(z-1)$. Multiplying these yields $\lambda_t(z)\,p_t \geq \lambda_t(z-1)\,p'_t$.
Because the Poisson CDF $\lambda \mapsto F(c;\lambda)$ is non-increasing in $\lambda$, it follows that
\[
g_t(z) = F\!\big(c_t;\,\lambda_t(z)\,p_t\big) \leq F\!\big(c_t;\,\lambda_t(z-1)\,p'_t\big) = g'_t(z-1).
\]

For the second inequality, we evaluate the halting probability $1 - g_t(z) = \Pr[\Po(\lambda_t(z)\,p_t) \geq c_t + 1]$. By differential privacy, $p_t \leq e^{\eps_t} p'_t$, and by part~(2), $\lambda_t(z+1) = e^{\eps_t}\lambda_t(z)$, yielding $\lambda_t(z)\,p_t \leq \lambda_t(z+1)\,p'_t$. Since $1 - F(c;\lambda)$ is non-decreasing in $\lambda$, we obtain
\[
1 - g_t(z) \leq 1 - g'_t(z+1).
\]

\smallskip
\noindent\emph{Case $s_t = -1$:}
By part~(1), $g_t(z) = 1 - F\!\big(c_t;\,\lambda_t(z)\,p_t\big)$ and $g'_t(z+1) = 1 - F\!\big(c_t;\,\lambda_t(z+1)\,p'_t\big)$.
By differential privacy, $p_t \leq e^{\eps_t} p'_t$, and by part~(2), $\lambda_t(z+1) = e^{\eps_t}\lambda_t(z)$, yielding $\lambda_t(z)\,p_t \leq \lambda_t(z+1)\,p'_t$.
Because $F(c_t;\lambda)$ is non-increasing in $\lambda$, the complement $1 - F(c_t;\lambda)$ is non-decreasing. Thus
\[
g_t(z) \leq g'_t(z+1) = g'_t(z - s_t).
\]

For the second inequality, note that $1 - g_t(z) = F(c_t;\,\lambda_t(z)\,p_t)$ and $1 - g'_t(z-1) = F(c_t;\,\lambda_t(z-1)\,p'_t)$.
Using $p_t \geq e^{-\eps_t} p'_t$ and $\lambda_t(z) = e^{\eps_t}\lambda_t(z-1)$, we find $\lambda_t(z)\,p_t \geq \lambda_t(z-1)\,p'_t$.
Applying the non-increasing property of $F(c_t;\lambda)$ gives
\[
1 - g_t(z) \leq 1 - g'_t(z-1) = 1 - g'_t(z + s_t).
\]
\end{proof}

\begin{remark}[Asymmetric use of input privacy]\label{rem:asymmetric-privacy}
The shift inequalities in Part~(3) depend on the input privacy guarantee only through bounds on the ratio $p_t/p'_t$ of the $+1$ output probabilities. In all four sub-cases, the Poisson thinning mean is $\lambda_t(z)\,p_t$, and the comparison is with $\lambda_t(z \pm 1)\,p'_t$; the ratio $(1-p_t)/(1-p'_t)$ of the $-1$ output probabilities never appears. 
When the algorithm operates on the negated mechanism (flipped variant, $s_t = -1$), the role of $p_t$ and $1 - p_t$ is exchanged, and the relevant ratio becomes $(1-p_t)/(1-p'_t)$ in terms of the original mechanism.
\end{remark}

\begin{lem}[Privacy of the generalized thresholding mechanism]\label{lem:GeneralizedPrivacy}
Let $(X_t)_{t\ge 1}$ and $(X'_t)_{t\ge 1}$ be dataset sequences with $X_t$ neighboring $X'_t$ for every $t$.
Run \Cref{alg:GTM} on $(X_t)_{t\ge 1}$ and on $(X'_t)_{t\ge 1}$ with the same sequences  $(\eps_t)$, $(c_t)$, $(\rho_t)$, $(\tau_{2,t})$, and $(s_t)$. Let $(a_t)_{t\ge 1}$ and $(a'_t)_{t\ge 1}$ denote the resulting output sequences.

Then for every $t\ge 1$,
\begin{align}
\Pr\!\big[a_{[t]}=(-s_1,\dots,-s_t)\big]
&\le
e^{1/\tau_1}\,\Pr\!\big[a'_{[t]}=(-s_1,\dots,-s_t)\big], \label{eq:privacy_continue}\\[6pt]
\Pr\!\big[a_{[t]}=(-s_1,\dots,-s_{t-1},s_t)\big]
&\le
e^{1/\tau_1+2/\tau_{2,t}}\,\Pr\!\big[a'_{[t]}=(-s_1,\dots,-s_{t-1},s_t)\big]. \label{eq:privacy_halt}
\end{align}
In particular, if $\tau_{2,t} = \tau_2$ for all $t$, the algorithm satisfies pure $\eps$-differential privacy with $\eps = 1/\tau_1 + 2/\tau_2$.
\end{lem}

\begin{proof}
Write $f_1$ for the pdf of $\eta_1 \sim \Exp(\tau_1)$ and $f_{2,t}$ for the density of $\eta_{2,t} \sim \Exp(\tau_{2,t})$. By the exponential density shift (\Cref{lem:PoissonFacts}, item~5),
\begin{equation}\label{eq:exp_density_shift}
f_1(x+1) = e^{-1/\tau_1}\,f_1(x) \quad\text{for } x \geq 0,
\qquad
f_{2,t}(y+2) = e^{-2/\tau_{2,t}}\,f_{2,t}(y) \quad\text{for } y \geq 0,
\end{equation}
or equivalently, $f_1(u-1) = e^{1/\tau_1}\,f_1(u)$ for $u \geq 1$ and $f_{2,t}(v-2) = e^{2/\tau_{2,t}}\,f_{2,t}(v)$ for $v \geq 2$.

\smallskip
\noindent\textbf{Conditional factorization.}
Conditional on $\eta_1$ and $\eta_{2,1}, \ldots, \eta_{2,t}$, each $Z_i$ is deterministic, and the Poisson draw and mechanism evaluations at distinct time steps are mutually independent. By the chain rule,
\begin{align*}
\Pr\!\big[a_{[t]} = (-s_1, \dots, -s_t) \;\big|\; \eta_1, \eta_{2,[t]}\big]
&= \prod_{i=1}^{t} \Pr\!\big[a_i = -s_i \;\big|\; a_{[i-1]} = (-s_1, \ldots, -s_{i-1}),\, \eta_1, \eta_{2,[t]}\big].
\end{align*}
The prefix $(-s_1, \ldots, -s_{i-1})$ determines the mechanism, dataset, and parameters that the analyst selects at step $i$. Given these, the conditional probability $\Pr[a_i = -s_i \mid a_{[i-1]} = (-s_1, \ldots, -s_{i-1}),\, \eta_1, \eta_{2,[t]}]$ depends only on $Z_i$, where the probability is taken over the random variables  $N_i$ and $y_{i,j}$, and equals $g_i(Z_i)$ by \Cref{def:ExecutionVariables}. Therefore,
\begin{equation}\label{eq:transcript_factorization}
\Pr\!\big[a_{[t]} = (-s_1, \dots, -s_t) \;\big|\; \eta_1, \eta_{2,[t]}\big]
= \prod_{i=1}^{t} g_i(Z_i).
\end{equation}
Similarly, the halt-at-$t$ probability factors as
\begin{equation}\label{eq:halt_factorization}
\Pr\!\big[a_{[t]} = (-s_1, \dots, -s_{t-1}, s_t) \;\big|\; \eta_1, \eta_{2,[t]}\big]
= \left(\prod_{i=1}^{t-1} g_i(Z_i)\right)(1 - g_t(Z_t)).
\end{equation}
The shift inequality $g_i(z) \leq g'_i(z - s_i)$ from \Cref{lem:PointwiseShifts} holds for every $\eps_i$-DP mechanism, regardless of how the analyst selected it.

\smallskip
\noindent\textbf{All-continuation transcript.}
We bound the probability that the algorithm continues through all $t$ steps. Integrating over $\eta_1$ and taking the expectation over the independent per-step noises $\eta_{2,[t]} := (\eta_{2,1}, \dots, \eta_{2,t})$:
\begin{equation}\label{eq:all_cont}
\Pr\!\big[a_{[t]}=(-s_1,\dots,-s_t)\big]
=
\int_0^\infty
\Ex_{\eta_{2,[t]}}\!\left[ \underbrace{
  \prod_{i=1}^t g_i\!\big(s_i(\mu_1 - x + \eta_{2,i})\big)}_{S}
\right]
f_1(x)\,\mathsf{d}x.
\end{equation}

By \Cref{lem:PointwiseShifts}~(item \ref{item:shift_inequalities}), $g_i(z) \leq g'_i(z - s_i)$ for all $z \in \mathbb{R}$. Evaluating at $z = s_i(\mu_1 - x + \eta_{2,i})$ gives $z - s_i = s_i(\mu_1 - (x+1) + \eta_{2,i})$, so the term in the integrand in \eqref{eq:all_cont} is bounded above by
\[
S \leq \prod_{i=1}^t g'_i\!\big(s_i(\mu_1 - (x+1) + \eta_{2,i})\big) .
\]
That is, 
\[
\Pr\!\big[a_{[t]}=(-s_1,\dots,-s_t)\big] \leq \int_0^\infty
\Ex_{\eta_{2,[t]}}\!\left[
  \prod_{i=1}^t g'_i\!\big(s_i(\mu_1 - (x+1) + \eta_{2,i})\big)
\right]
f_1(x)\,\mathsf{d}x
\]

Let $u = x + 1$. Then the integration domain $\{x \geq 0\}$ maps to $\{u \geq 1\}$, and we obtain
\[
\Pr\!\big[a_{[t]}=(-s_1,\dots,-s_t)\big]
\leq
\int_1^\infty
\Ex_{\eta_{2,[t]}}\!\left[
  \prod_{i=1}^t g'_i\!\big(s_i(\mu_1 - u + \eta_{2,i})\big)
\right]
f_1(u-1)\,\mathsf{d}u.
\]

In the region $\{u \geq 1\}$, the density shift \eqref{eq:exp_density_shift} gives $f_1(u-1) = e^{1/\tau_1}\,f_1(u)$. Since the integrand is non-negative, enlarging the domain from $\{u \geq 1\}$ back to $\{u \geq 0\}$ only increases the value of integration, i.e., 
\begin{align*}
    \Pr\!\big[a_{[t]}=(-s_1,\dots,-s_t)\big]
    &\leq e^{1/\tau_1}
    \int_0^\infty
    \Ex_{\eta_{2,[t]}}\!\left[
  \prod_{i=1}^t g'_i\!\big(s_i(\mu_1 - u + \eta_{2,i})\big)
\right]
f_1(u)\,\mathsf{d}u \\
    &=
e^{1/\tau_1}\,\Pr\!\big[a'_{[t]}=(-s_1,\dots,-s_t)\big].
\end{align*}

\smallskip
\noindent\textbf{Halt-at-$t$ transcript.}
We bound the probability that the algorithm continues through steps $1, \dots, t-1$ and halts at step $t$. Writing $x$ for the integration variable of $\eta_1$ and $y$ for that of $\eta_{2,t}$, and taking the expectation over the remaining per-step noises $\eta_{2,[t-1]} := (\eta_{2,1}, \dots, \eta_{2,t-1})$:
\begin{align}
&\Pr\!\big[a_{[t]}=(-s_1,\dots,-s_{t-1},s_t)\big] \notag\\
&\quad=
\Ex_{\eta_{2,[t-1]}}\!\left[
  \int_0^\infty\!\int_0^\infty
    \left(\prod_{i=1}^{t-1} g_i\!\big(s_i(\mu_1 - x + \eta_{2,i})\big)\right)
    \bigl(1-g_t\!\big(s_t(\mu_1 - x + y)\big)\bigr)\,
    f_1(x)\,f_{2,t}(y)\,
    \mathsf{d}y\,\mathsf{d}x
\right]. \label{eq:halt_transcript}
\end{align}

The proof idea remains the same as above except that we have to treat the halting step differently. In particular, as in the previous case, we first apply the shift inequalities from \Cref{lem:PointwiseShifts} (item~\ref{item:shift_inequalities}): 
 for each continuation step $i \leq t-1$, we use $g_i(z) \leq g'_i(z-s_i)$, and for the halting step, $1-g_t(z) \leq 1-g'_t(z+s_t)$. Evaluating the shifted arguments:
\begin{align*}
s_i(\mu_1 - x + \eta_{2,i}) - s_i &= s_i\!\big(\mu_1 - (x+1) + \eta_{2,i}\big), \\
s_t(\mu_1 - x + y) + s_t &= s_t(\mu_1 - x + y + 1).
\end{align*}
We perform a simultaneous change of variables: $u = x + 1$ and $v = y + 2$, where the variable $v$ deals with the halting step. Then $x = u - 1$ and $y = v - 2$, and the arguments become
\begin{align*}
s_i\!\big(\mu_1 - (x+1) + \eta_{2,i}\big) &= s_i(\mu_1 - u + \eta_{2,i}), \\
s_t(\mu_1 - x + y + 1) &= s_t\!\big(\mu_1 - (u-1) + (v-2) + 1\big) = s_t(\mu_1 - u + v).
\end{align*}

The integration domain $\{x \geq 0,\, y \geq 0\}$ maps to $\{u \geq 1,\, v \geq 2\}$. Substituting into \eqref{eq:halt_transcript} and applying the shift inequalities:
\begin{align*}
&\Pr\!\big[a_{[t]}=(-s_1,\dots,-s_{t-1},s_t)\big]\\
&\quad\leq
\Ex_{\eta_{2,[t-1]}}\!\left[
  \int_1^\infty\!\int_2^\infty
    \left(\prod_{i=1}^{t-1} g'_i\!\big(s_i(\mu_1 - u + \eta_{2,i})\big)\right)
    \bigl(1-g'_t\!\big(s_t(\mu_1 - u + v)\big)\bigr)\,
    f_1(u-1)\,f_{2,t}(v-2)\,
    \mathsf{d}v\,\mathsf{d}u
\right].
\end{align*}

In the restricted region $\{u \geq 1,\, v \geq 2\}$, we apply the density shifts \eqref{eq:exp_density_shift}: $f_1(u-1) = e^{1/\tau_1}\,f_1(u)$ and $f_{2,t}(v-2) = e^{2/\tau_{2,t}}\,f_{2,t}(v)$. Enlarging the integration region back to $\{u \geq 0,\, v \geq 0\}$ yields
\begin{align*}
&\Pr\!\big[a_{[t]}=(-s_1,\dots,-s_{t-1},s_t)\big]\\
&\quad\leq
e^{1/\tau_1+2/\tau_{2,t}}
\Ex_{\eta_{2,[t-1]}}\!\left[
  \int_0^\infty\!\int_0^\infty
    \left(\prod_{i=1}^{t-1} g'_i\!\big(s_i(\mu_1 - u + \eta_{2,i})\big)\right)
    \bigl(1-g'_t\!\big(s_t(\mu_1 - u + v)\big)\bigr)\,
    f_1(u)\,f_{2,t}(v)\,
    \mathsf{d}v\,\mathsf{d}u
\right]\\
&\quad=
e^{1/\tau_1+2/\tau_{2,t}}\,\Pr\!\big[a'_{[t]}=(-s_1,\dots,-s_{t-1},s_t)\big].
\end{align*}

\smallskip
\noindent\textbf{Overall privacy.}
When $\tau_{2,t} = \tau_2$ for all $t$, the multiplicative factor in both equations \eqref{eq:privacy_continue} and \eqref{eq:privacy_halt} is at most $e^{1/\tau_1 + 2/\tau_2}$. Every reachable transcript of the algorithm is either an all-continuation prefix or a halt-at-$t$ transcript for some $t \geq 1$. Summing the per-transcript bounds over all $t$ yields $\Pr[\text{output} \in F] \leq e^{1/\tau_1 + 2/\tau_2}\,\Pr[\text{output}' \in F]$ for every measurable event $F$, establishing pure $\eps$-differential privacy with $\eps = 1/\tau_1 + 2/\tau_2$.
\end{proof}

\begin{remark}[Sample complexity]\label{rem:SampleComplexity}
Run \Cref{alg:GTM} for $t$ steps. At each step $i \leq t$, the algorithm draws $N_i \sim \Po(\lambda_i)$ evaluations of $\calM_i(X_i)$, where $\lambda_i = \rho_i \exp(\eps_i Z_i)$. The total number of evaluations through step $t$ is $N_{\leq t} := \sum_{i=1}^t N_i$.

Conditional on the shared noise $\eta_1$ and the per-step noises $\eta_{2,1}, \dots, \eta_{2,t}$, each $Z_i = s_i(\mu_1 - \eta_1 + \eta_{2,i})$ is deterministic, so the counts $N_1, \dots, N_t$ are conditionally independent with $N_i \mid \eta_1, \eta_{2,i} \sim \Po(\lambda_i)$. By Poisson additivity (\Cref{lem:PoissonFacts}),
\[
N_{\leq t} \mid \eta_1, \eta_{2,[t]} \;\sim\; \Po\!\left(\sum_{i=1}^t \lambda_i\right).
\]
Since $Z_i$ depends on $\eta_1$ and $\eta_{2,i}$ but not on $\eta_{2,j}$ for $j \neq i$, the conditional total mean $\sum_{i=1}^t \lambda_i$ depends on $\eta_1$ through every term, and on each $\eta_{2,i}$ only through the $i$-th term.
\end{remark}
\section{Accuracy Guarantees}\label{sec:accuracy}

The privacy guarantee of the \GTM holds for adaptive choices of threshold parameter $c_t$ and base rate $\rho_t$. We now show how to set these parameters to achieve concrete accuracy guarantees. In \Cref{thm:GlobalUtility}, we derive a uniform high probability accuracy guarantee, and in \Cref{thm:ExPost} we derive a similar accuracy guarantee, but under an ex-post privacy guarantee. Both of these results follow from a technical lemma, \Cref{lem:ParameterSelection}.

To motivate \Cref{lem:ParameterSelection}, consider a step with $s_t = +1$ (above threshold). If $p_t \geq \pstar_t$, we want the algorithm to halt. 
The algorithm generates a count $K_t \sim \Po(\lambda_t\, p_t)$ and halts if $K_t > c_t$. Since $\lambda_t = \rho_t \exp(\eps_t Z_t)$, there is inherent uncertainty in the value of $\lambda_t$ due to the privatizing noise $Z_t = s_t(\mu_1 - \eta_1 + \eta_{2,t})$. 
We call the event $\calG$ the event that certain high-probability bounds hold on the value of $\eta_1$ and $\eta_{2,t}$, see \Cref{lem:NoiseControl}.

Now, conditioned on $\calG$, there is a value $\Lambda_t>1$ such that $\lambda_t \in [\rho_t,\rho_t\Lambda_t]$.
To ensure that the probability of a Type~II error (missed halt) is at most $\beta^{\tII}$, we need that even under a worst-case draw of $\lambda_t$, the thinned Poisson mean $\lambda_t\, \pstar_t$ is large enough so that $\Pr[K_t \leq c_t]$ is small; by the left tail bound of Poisson distribution (i.e., item~\ref{item:tail_bound_Poisson_left} in \Cref{lem:PoissonFacts}), this requires $\lambda_t\, \pstar_t \geq \lamL{c_t}{\beta^{\tII}}$. If $p_t\geq\pstar_t$, then conditioned on $\calG$, since $\lambda_t \geq \rho_t$, it suffices that $\rho_t \geq \lamL{c_t}{\beta^{\tII}} / \pstar_t$ for $K_t >c_t$ with probability $1-\beta^{\tII}$. In particular, combining $p_t \geq \pstar_t$ and $\lambda_t \geq \rho_t$, if $\rho_t \geq \lamL{c_t}{\beta^{\tII}} / \pstar_t$, then we get the inequality:
\[
\lambda_t\,\geq\, \rho_t \,\geq\, \frac{\lamL{c_t}{\beta^{\tII}}}{ \pstar_t} \,\geq\, \frac{\lamL{c_t}{\beta^{\tII}}}{ p_t}.
\]

Conversely, to ensure that the probability of a Type~I error (i.e., false halt) is at most $\beta^{\tI}$ when $p_t \leq \bar{p}_t$ for some $\bar{p_t} < \pstar_t$, we need that for any value of $\lambda_t \in [\rho_t,\rho_t\Lambda_t]$, the mean $\lambda_t\, \bar{p}_t $ is small enough that $\Pr[K_t > c_t]$ is small. By item~\ref{item:tail_bound_Poisson_right} in  \Cref{lem:PoissonFacts}, this requires $\lambda_t\, \bar{p}_t \leq \lamR{c_t}{\beta^{\tI}}$. Again, conditioned on $\calG$, this implies that for continuation, it suffices that 
    \[ \bar{p}_t \leq \frac{\lamR{c_t}{\beta^{\tI}}}{ \rho_t\, \Lambda_t}. \]

We now have two requirements: the halting condition $\rho_t \geq \lamL{c_t}{\beta^{\tII}}/\pstar_t$, and the continuation condition $\bar{p}_t \leq \lamR{c_t}{\beta^{\tI}}/(\rho_t \Lambda_t)$. Since larger $\rho_t$ means higher sample complexity (the Poisson mean $\lambda_t = \rho_t \exp(\eps_t Z_t)$ scales linearly with $\rho_t$), we set $\rho_t$ as small as possible, namely $\rho_t = \lamL{c_t}{\beta^{\tII}}/\pstar_t$. Substituting into the continuation condition, the rejection threshold below which the test is guaranteed to continue becomes
\[
  \bar{p}_t = \frac{\lamR{c_t}{\beta^{\tI}}}{\rho_t \Lambda_t} = \frac{\lamR{c_t}{\beta^{\tI}}}{\lamL{c_t}{\beta^{\tII}}} \cdot \frac{\pstar_t}{\Lambda_t}.
\]
This factorization isolates the two sources of accuracy loss. The factor $\pstar_t/\Lambda_t$ is the cost of privacy: it reflects the worst-case multiplicative perturbation of $\lambda_t$ by the noise $\exp(\eps_t Z_t)$, and we show in \Cref{sec:lower_bounds} that it is unavoidable up to constant factors in the exponent. The ratio $\lamR{c_t}{\beta^{\tI}}/\lamL{c_t}{\beta^{\tII}} \leq 1$ is the cost of finite-sample Poisson testing: both thresholds grow linearly in $c_t$, but their gap grows sublinearly, so the ratio approaches $1$ as $c_t \to \infty$. \Cref{lem:ParameterSelection} makes this precise: for any target ratio $\gamma \in (1,2]$, it constructs a value of $c_t$ such that $\lamL{c_t}{\beta^{\tII}} \leq \gamma \cdot \lamR{c_t}{\beta^{\tI}}$, giving $\bar{p}_t \geq \pstar_t/(\gamma \Lambda_t)$. The required base rate scales as $\rho_t = O((\ln(1/\beta^{\tI}) + \ln(1/\beta^{\tII}))/((\gamma-1)^2 \pstar_t))$.

The case $s_t = -1$ (below threshold) is analogous, with $\lambda_t \in [\rho_t/\Lambda_t,\rho_t]$ (conditioned on high probability bounds on $\eta_1$ and $\eta_{2,t}$), and with the roles of the left and right tails exchanged: halting uses the right tail (with parameter $\beta^{\tII}$) and continuation uses the left tail (with parameter $\beta^{\tI}$).

In this description, we have described the above threshold variant of \Cref{alg:GTM}, where for every step $t$, $s_t = +1$; the below threshold analysis is entirely analogous. However, there is another consequence of this functionality.

The choice of direction parameter $s_t$ gives us two ways of running the algorithm to solve the standard above threshold formulation of the generalized private testing problem, where we are required to halt for $p_t \geq \pstar_t$.
\begin{enumerate}
    \item \emph{Direct variant:} We set $s_t = +1$ for all $t\geq1$, and run the mechanism as is
    \item \emph{Flipped variant:} We set $s_t = -1$ for all $t\geq1$, and feed the \GTM the `flipped' input mechanisms $-\calM_t$ with success probabilities $q_t = 1 - p_t$, a flipped target success threshold $q^*_t = 1-\pstar_t$.
\end{enumerate} 
For the second formulation, we see that the \GTM halts when $s_t = -1$, indicating that $q_t \leq q^*_t \Leftrightarrow p_t \geq \pstar_t$, so clearly both variants solve the same problem. 

We find that the flipped variant offers better accuracy when $\pstar_t$ is close to $1$. From the description above, we see that the direct variant with $s_t = +1$ for all $t\geq1$ achieves a rejection threshold of $\bar{p}_t = \pstar_t/(\gamma\,\Lambda_t)$, for some $\gamma,\Lambda_t>1$. It follows that $\pstar_t - \bar{p}_t = \pstar_t(1 - 1/\gamma\Lambda_t)$, i.e the gap between the thresholds is large when $\pstar_t$ approaches $1$, but small when $\pstar_t$ approaches $0$. Similarly, for the flipped variant, we get that $\bar{q}_t = q^*_t/\gamma\Lambda_t$, and the gap between the thresholds is large when $q^*_t$ approaches $1$ and small when $q^*_t$ approaches $0$. However, since $q^*_t = 1-\pstar_t$, this implies that the flipped variant has a smaller threshold gap when the given value of $\pstar_t$ approaches $1$. By choosing the variant adaptively that maximizes our accuracy, we can achieve a tight accuracy for both regimes, i.e. when $\pstar_t$ approaches $0$ \emph{and} when $\pstar_t$ approaches $1$.

\begin{lem}[Bounding $\lambda_t$]\label{lem:NoiseControl}
Run Algorithm~\ref{alg:GTM} with parameters $\tau_1 > 0$ and $\mu_1 > 0$, per-step noise scales $\tau_{2,t} > 0$, and per-step privacy parameters $\eps_t \geq 0$. Fix a constant $c_\calG \leq 1$, and define the \emph{good noise event}
\[
\calG \;:=\; \bigl\{\eta_1 \leq \mu_1\bigr\}
\;\cap\; \bigcap_{t \geq 1}
\bigl\{\eta_{2,t} \leq \tau_{2,t}\ln(1/c_\calG \beta_t)\bigr\}.
\]
Then the following hold.
\begin{enumerate}
\item \emph{Failure probability.}
$\Pr[\calG^c] \leq e^{-\mu_1/\tau_1} + c_\calG \sum_{t=1}^{\infty} \beta_t$.

\item \emph{Noise range conditioned on $\calG$.}
\label{item:noiseConditionedonG}
Define $\Lambda_t := \exp\!\bigl(\eps_t(\mu_1 + \tau_{2,t}\ln(1/(c_\calG \beta_t)))\bigr)$.
Conditioned on $\calG$, the signed noise satisfies $s_t Z_t \in [0,\;\mu_1 + \tau_{2,t}\ln(1/(c_\calG \beta_t))]$, and the Poisson mean satisfies
\begin{align*}
s_t = +1 &\implies \lambda_t \in [\rho_t,\;\rho_t\Lambda_t], \\
s_t = -1 &\implies \lambda_t \in [\rho_t/\Lambda_t,\;\rho_t].
\end{align*}

\item \emph{Halting range on $\{\eta_1 \leq \mu_1\}$.} \label{item:haltingRange}
Conditioned on the event $\{\eta_1 \leq \mu_1\}$ alone (and not on all of $\calG$), $s_t Z_t \geq 0$ for all $t$, and therefore $\lambda_t \geq \rho_t$ when $s_t = +1$ and $\lambda_t \leq \rho_t$ when $s_t = -1$.
\end{enumerate}
\end{lem}

\begin{proof}
By the tail bound of the exponential distribution, $\Pr[\eta_1 > \mu_1] = e^{-\mu_1/\tau_1}$ and $\Pr[\eta_{2,t} > \tau_{2,t}\ln(1/c_\calG \beta_t)] = c_\calG \beta_t$ for each $t$. A union bound gives item~1.

For Items~2 and~3, recall that $Z_t = s_t(\mu_1 - \eta_1 + \eta_{2,t})$ and $s_t \in \{-1,1\}$. Therefore, $s_t Z_t = \mu_1 - \eta_1 + \eta_{2,t}$. When $\eta_1 \leq \mu_1$, the lower bound $s_t Z_t \geq \eta_{2,t} \geq 0$ holds since $\eta_{2,t}$ has nonneg\-ative support and the rest of Item 3 holds by the definition of $\lambda_t$. On $\calG$, the upper bound $s_t Z_t \leq \mu_1 + \tau_{2,t}\ln(1/(c_\calG \beta_t))$ holds by the per-step noise constraint. Exponentiating and recalling $\lambda_t = \rho_t\exp(\eps_t Z_t)$ gives the stated ranges.
\end{proof}

\begin{lem}[Parameter selection]\label{lem:ParameterSelection}
Fix a target threshold $\pstar_t \in (0,1)$, per-step error parameters $\beta^{\tI}, \beta^{\tII} \in (0,1)$, a parameter $\Lambda \geq 1$, and a target ratio $\gamma \in (1, 2]$. Let
    \[p\in (0,1), \quad \lambda > 0, \quad\mbox{and}\quad K \sim \Po(\lambda\,p).\]
For $\alpha,\beta\in(0,1)$, let $c_{\alpha,\beta}$ be defined as in \cref{lem:PoissonOverhead}, i.e. for  $A = \ln(1/\alpha)$ and $B = \ln(1/\beta)$,  $c_{\alpha,\beta} = \max(\lfloor 2\max (A,B)\rfloor,\lceil r^2 \rceil)$, where
\[ r = \frac{\sqrt{2A} + \gamma\sqrt{2B} + \sqrt{(\sqrt{2A}+\gamma\sqrt{2B})^2 + 8(\gamma-1)A}}{2(\gamma-1)}.\]

\begin{enumerate}
\item \textbf{Above threshold} {\em ($s = +1$)}. 
\label{item:aboveThreshold}
Let $(\alpha,\beta) = (\beta^{\tII},\beta^{\tI})$, and let $c = c_{\alpha,\beta}$. Let $\rho = \lamL{c}{\beta^{\tII}}/\pstar_t$. Then the following statements hold:  
\begin{enumerate}
\item \emph{Halting:} If $\lambda \geq \rho$ and $p \geq \pstar_t$, then $\Pr[K \leq c] \leq \beta^{\tII}$.
\item \emph{Continuation:} If $\lambda \leq \rho\Lambda$ and $p \leq \pstar_t/(\gamma\Lambda)$, then $\Pr[K \geq c+1] \leq \beta^{\tI}$.
\end{enumerate}

\item \textbf{Below threshold} {\em ($s = -1$)}.
\label{item:belowThreshold}
Let $(\alpha,\beta) = (\beta^{\tI},\beta^{\tII})$ and let $c = c_{\alpha,\beta}$. Set $\rho = \lamR{c}{\beta^{\tII}}/\pstar_t$.
\begin{enumerate}
\item \emph{Halting:} If $\lambda \leq \rho$ and $p \leq \pstar_t$, then $\Pr[K \geq c+1] \leq \beta^{\tII}$.
\item \emph{Continuation:} If $\lambda \geq \rho/\Lambda$ and $p \geq \gamma\Lambda\pstar_t$, then $\Pr[K \leq c] \leq \beta^{\tI}$.
\end{enumerate}
\end{enumerate}
Further, in both cases (i.e., $s\in\{+1,-1\}$), we have the base rate upper bound $\rho = O\!\big(\tfrac{\ln(1/\beta^{\tI}) + \ln(1/\beta^{\tII})}{ (\gamma-1)^2\,\pstar_t}\big)$, and the base rate lower bound that if\, $\beta^{\tI}, \beta^{\tII} \leq 1/e$, then $\rho \geq \tfrac{1}{\pstar_t}\max\!\big\{\ln(1/\beta^{\tI}),\, \ln(1/\beta^{\tII})\big\}$.
\end{lem}

\begin{proof}
We prove each of the parts separately.

\medskip\noindent\textbf{Above threshold.}
Since $c$ is defined via \Cref{lem:PoissonOverhead} with $(\alpha,\beta) = (\beta^{\tII},\beta^{\tI})$, it follows from \Cref{lem:PoissonOverhead} that $\lamL{c}{\beta^{\tII}} \leq \gamma\,\lamR{c}{\beta^{\tI}}$ and $c + 1 > 2\max(L^{\tI}, L^{\tII})$, where $L^{\tI} = \ln(1/\beta^{\tI})$ and $L^{\tII} = \ln(1/\beta^{\tII})$. Since we set $\rho = \lamL{c}{\beta^{\tII}}/\pstar_t$, rearranging, 
\[\rho\,\pstar_t = \lamL{c}{\beta^{\tII}}.\] 

\emph{Halting.}
If $p \geq \pstar_t$ and $\lambda \geq \rho$, then $\lambda\,p \geq \rho\,\pstar_t = \lamL{c}{\beta^{\tII}}$. By \Cref{lem:PoissonFacts}(3a), $\Pr[K \leq c] \leq \beta^{\tII}$.

\emph{Continuation.}
If $p \leq \pstar_t / (\gamma \Lambda)$, and $\lambda \leq \rho\Lambda$, then 
\[\lambda\,p \leq \rho\,\pstar_t/\gamma = \lamL{c}{\beta^{\tII}}/\gamma \leq \lamR{c}{\beta^{\tI}}.\] 
The hypothesis $c + 1 > 2\ln(1/\beta^{\tI})$ ensures the domain condition of \Cref{lem:PoissonFacts}(3b), giving $\Pr[K \geq c+1] \leq \beta^{\tI}$.

\medskip
\noindent\textbf{Below threshold.}
Since $c$ is defined via \Cref{lem:PoissonOverhead} with $(\alpha,\beta) = (\beta^{\tI},\beta^{\tII})$, it follows from \cref{lem:PoissonOverhead} that $\lamL{c}{\beta^{\tI}} \leq \gamma\,\lamR{c}{\beta^{\tII}}$ and $c + 1 > 2\max(L^{\tI}, L^{\tII})$ where $L^{\tI} = \ln(1/\beta^{\tI})$ and $L^{\tII} = \ln(1/\beta^{\tII})$. Since we set $\rho = \lamR{c}{\beta^{\tII}}/\pstar_t$, rearranging, 
\[\rho\,\pstar_t = \lamR{c}{\beta^{\tII}}.\] 

\emph{Halting.}
If $p \leq \pstar_t$ and $\lambda \leq \rho$, then $\lambda\,p \leq \rho\,\pstar_t = \lamR{c}{\beta^{\tII}}$. Since $c + 1 > 2\ln(1/\beta^{\tII})$ the domain condition of \Cref{lem:PoissonFacts}(3b) is fulfilled, and so $\Pr[K \geq c+1] \leq \beta^{\tII}$.

\emph{Continuation.}
If $p \geq \gamma\Lambda \pstar_t$ 
and $\lambda \geq \rho/\Lambda$, then 
    \[\lambda\,p \geq \gamma\,\rho\,\pstar_t = \gamma \lamR{c}{\beta^{\tII}} \geq \lamL{c}{\beta^{\tI}}.\]
By item~\ref{item:tail_bound_Poisson_left} in \Cref{lem:PoissonFacts}, $\Pr[K \leq c] \leq \beta^{\tI}$.

\medskip
\noindent\textbf{Base rate bounds}

\emph{Upper bound.}
In the above-threshold setting, $\rho = \lamL{c}{\beta^{\tII}}/\pstar_t$ and $\lamL{c}{\beta^{\tII}} = O(c + \sqrt{c\,L^{\tII}} + L^{\tII})$. In the below-threshold setting, $\rho = \lamR{c}{\beta^{\tII}}/\pstar_t$ and $\lamR{c}{\beta^{\tII}} \leq c + 1$. In both cases, $c = O((L^{\tI} + L^{\tII})/(\gamma-1)^2)$ by \Cref{lem:PoissonOverhead}(3), and dividing by $\pstar_t$ gives the stated bounds.

\emph{Lower bound.}
Write $L^{\tI} = \ln(1/\beta^{\tI}) \geq 1$ and
$L^{\tII} = \ln(1/\beta^{\tII}) \geq 1$.
From~\eqref{eq:lamL_def},
$\lamL{c}{\alpha} \geq c + 2\ln(1/\alpha)$
for all $c \geq 0$ and $\alpha \in (0,1)$,
since $\sqrt{2c\ln(1/\alpha) + \ln^2(1/\alpha)} \geq \ln(1/\alpha)$.
From \Cref{lem:PoissonOverhead}(2) and integrality of $c$,
$c \geq 2\max(L^{\tI}, L^{\tII}) - 1$.

In the above threshold setting, $\rho\,\pstar_t = \lamL{c}{\beta^{\tII}} \geq c + 2L^{\tII}
\geq (2L^{\tI} - 1) + 2L^{\tII}$.
Since $L^{\tI} \geq 1$, this gives
$\rho\,\pstar_t \geq L^{\tI} + 2L^{\tII} \geq \max(L^{\tI}, L^{\tII})$.
 
In the below threshold setting, by \Cref{lem:PoissonOverhead}(1),
$\gamma\,\rho\,\pstar_t = \gamma\,\lamR{c}{\beta^{\tII}}
\geq \lamL{c}{\beta^{\tI}} \geq c + 2L^{\tI}
\geq (2L^{\tII} - 1) + 2L^{\tI}$.
Dividing by $\gamma \leq 2$ gives
$\rho\,\pstar_t \geq L^{\tI} + L^{\tII} - \tfrac{1}{2}
\geq \max(L^{\tI}, L^{\tII})$,
since $\min(L^{\tI}, L^{\tII}) \geq 1 > \tfrac{1}{2}$.
\end{proof}

\subsection{Global accuracy guarantee}

\Cref{thm:GlobalUtility} instantiates \Cref{lem:GeneralizedPrivacy}, \Cref{lem:NoiseControl}, and \Cref{lem:ParameterSelection} with specific choices of $\gamma$, $\beta^{\tI}$, $\beta^{\tII}$, and $\Lambda_t$ determined by the privacy budget and failure probability allocation, and establishes a worst-case accuracy bound for \Cref{alg:GTM}.

\begin{theorem}[Global utility]\label{thm:GlobalUtility}
Fix $\varepsilon >0$, $\theta >0$, a gap parameter $\gamma \in (1,2]$, and a failure probability $\beta \in (0,1)$. Define
\[
  \Lambda_t \;:=\; \left(\frac{4}{\beta}\right)^{(\eps_t/\eps)(1 + 2/\theta)} \left(\frac{4}{\beta_t}\right)^{(\eps_t/\eps)(2+\theta)}.
\]
For any choice of $\pstar_t$, $s_t$, $\rho_t$ and $c_t$, \Cref{alg:GTM} is $\varepsilon$-differentially private. Further, there exists a choice of parameter settings for $s_t$, $\rho_t$ and $c_t$ such that, with probability at least $1 - \beta$, the following hold simultaneously for all $t \geq 1$:

\begin{enumerate}
\item \textbf{Halting.}
  At every step $t$ where $p_t \geq \pstar_t$, the algorithm halts.

\item \textbf{Continuation.}
  The algorithm continues past every step $t$ satisfying $p_t \leq \bar{p}_t$, where the \emph{rejection threshold}
  \begin{align*}
    \bar{p}_t \;:=\; \max\!\left(\frac{\pstar_t}{\gamma\,\Lambda_t},\;\; 1 - \gamma\,\Lambda_t\,(1-\pstar_t)\right).
  \end{align*}

\item \textbf{Sample complexity.}
  The number of evaluations of the mechanism $\calM_t$ at step $t$ is
  $N_t \sim \Po(\lambda_t)$ for
  $\lambda_t \leq \tilde{A}_t \cdot w_t$, where
  \[
    \tilde{A}_t \;:=\; \frac{C_A \cdot \ln(t/\beta)}{(\gamma-1)^2}
    \cdot \max\!\left(\frac{(4/\beta)^{(1+2/\theta)\eps_t/\eps}}{\pstar_t},\;
    \frac{1}{1-\pstar_t}\right)
  \]
  for some universal constant $C_A$ and $w_t := e^{\eps_t\eta_{2,t}} \sim \Par(\sigma_t)$ with
  $\sigma_t := \eps/(\eps_t(\theta+2))$, independently across steps.
\end{enumerate}
\end{theorem}

\begin{proof}[Proof of \Cref{thm:GlobalUtility}]
We run Algorithm~\ref{alg:GTM} with the following parameter setting.
For $\theta>0$, define 
\[
  \frac{1}{\tau_1} = \frac{\theta\eps}{\theta+2}, \qquad
  \frac{2}{\tau_2} = \frac{2\eps}{\theta + 2}, \qquad \forall t \ \ \tau_{2,t} = \tau_2 \qquad
  \mu_1 = \tau_1 \ln\!\bigl(\tfrac{4}{\beta}\bigr).
\]

\medskip
\noindent \textbf{Privacy.}
Since $1/\tau_1 + 2/\tau_2 = \eps$ and $\tau_{2,t} = \tau_2$ for all $t$, \Cref{lem:GeneralizedPrivacy} gives pure $\eps$-differential privacy.

\medskip
\noindent \textbf{Bounding $\lambda_t$.}
Applying \Cref{lem:NoiseControl} with $c_\calG = 1/4$, since $e^{-\mu_1/\tau_1} = \beta/4$ and $c_\calG \sum_t \beta_t = \beta/4$, we can write
\[
  \Pr[\calG^c] \leq \frac{\beta}{4} + \frac{\beta}{4} = \frac{\beta}{2}.
\]
From the definition of $\Lambda_t$ in \cref{lem:NoiseControl} and using that $\mu_1 = \tau_1 \ln(4/\beta)$, we can write
    \[ \Lambda_t = \exp\!\bigl(\eps_t(\mu_1 + \tau_{2,t}\ln(4/\beta_t))\bigr) = \left(\frac{4}{\beta}\right)^{\tau_1 \eps_t} \left(\frac{4}{\beta_t}\right)^{\tau_{2,t} \eps_t} \]
From the definition of $\tau_1$ and $\tau_{2,t}$, we have that
    \[ \tau_1 \eps_t = \frac{\eps_t}{\eps} \left(\frac{2}{\theta} + 1\right) \quad \tau_{2,t} \eps_t =  \frac{\eps_t}{\eps} \left( 2 + \theta \right)\quad \Rightarrow \Lambda_t = \left(\frac{4}{\beta}\right)^{(\eps_t/\eps)(1 + 2/\theta)} \left(\frac{4}{\beta_t}\right)^{(\eps_t/\eps)(2+\theta)} .\]
Now, \Cref{lem:NoiseControl} shows that conditioning on $\calG$, $\lambda_t \in [\rho_t,\,\rho_t\Lambda_t]$ when $s_t = +1$ and $\lambda_t \in [\rho_t/\Lambda_t,\,\rho_t]$ when $s_t = -1$. When $\eta_1 \leq \mu_1$, then  $\lambda_t \geq \rho_t$ when $s_t = +1$ and $\lambda_t \leq \rho_t$ when $s_t = -1$.

\medskip \noindent \textbf{Accuracy.}
Conditioned on $\calG$, we know that for all time-steps $t$, for any choice of $\rho_t$, if $s_t = +1$ then $\lambda_t \in [\rho_t,\rho_t\,\Lambda_t]$, and if $s_t = -1$ then $\lambda_t \in [\rho_t/\Lambda_t, \rho_t]$.

If $s_t = +1$: applying item \ref{item:aboveThreshold} of \Cref{lem:ParameterSelection} with target $\pstar_t$, error parameters $(\beta_t/4, \beta_t/4)$, gap $\gamma$, and parameter $\Lambda = \Lambda_t$ gives us that there exists an integer $c_t$ such that for $\rho_t = \lamL{c_t}{\beta_t/4} /\pstar_t$, we have
\begin{enumerate}
    \item Halting: if $\lambda_t \geq \rho_t$ and $p_t \geq \pstar_t$, then $\Pr[K_t \leq c_t] \leq \beta_t/4$.
    \item Continuation: if $\lambda_t \leq \rho_t\Lambda_t$ and $p_t \leq \pstar_t/(\gamma\Lambda_t)$, then $\Pr[K_t \geq c_t + 1] \leq \beta_t/4$.
    \item Base rate upper bound $\rho_t = O(\ln (1/\beta_t)/((\gamma-1)^2\pstar_t))$, and lower bound $\rho_t \geq (1/\pstar_t) \ln(4/\beta_t)$ .
\end{enumerate}
  
Similarly, if $s_t = -1$: item~\ref{item:belowThreshold} of \Cref{lem:ParameterSelection} applied to the
negated mechanism with target $1-\pstar_t$, error parameters
$(\beta_t/4, \beta_t/4)$, gap $\gamma$, parameter $\Lambda = \Lambda_t$
gives us that there exists an integer $c_t$ such that for
$\rho_t = \lamR{c_t}{\beta_t/4}/(1-\pstar_t)$, we have
\begin{enumerate}
    \item Halting: if $\lambda_t \leq \rho_t$ and $1 - p_t \leq 1 - \pstar_t \Leftrightarrow p_t \geq \pstar_t$, then $\Pr[K_t \geq c_t+1] \leq \beta_t/4$.
    \item Continuation: if $\lambda_t \geq \rho_t/\Lambda_t$ and
      $1 - p_t \geq \gamma \Lambda_t (1 - \pstar_t) \Leftrightarrow p_t \leq 1 - \gamma\Lambda_t(1-\pstar_t)$, then
      $\Pr[K_t \leq c_t] \leq \beta_t/4$.
    \item Base rate:
      $\rho_t = O(\ln(1/\beta_t)/((\gamma-1)^2(1-\pstar_t)))$, and lower bound $\rho_t \geq (1/(1-\pstar_t)) \ln(4/\beta_t)$.
\end{enumerate}

\medskip
\noindent\textbf{Flipping rule.}
The direct execution guarantees continuation at $\calM_t(X_t)$ for $p_t \leq \pstar_t/(\gamma\Lambda_t)$
and the flipped execution for $p_t \leq 1 - \gamma\Lambda_t(1-\pstar_t)$.
The analyst may select whichever threshold is larger for a better test.
Writing $x = \gamma\Lambda_t \geq 1$,
\[
  \frac{\pstar_t}{x} - \bigl(1 - x(1-\pstar_t)\bigr)
  \;=\; \frac{(1-\pstar_t)(x-1)(x - \pstar_t/(1-\pstar_t))}{x}.
\]
Since $x \geq 1$, the sign is determined by
$x - \pstar_t/(1-\pstar_t)$. It follows that the direct path has a higher rejection threshold when
$\gamma\Lambda_t \geq \pstar_t/(1-\pstar_t)$, equivalently
$\gamma\Lambda_t(1-\pstar_t) \geq \pstar_t$, and the flipped path
is better otherwise.

Therefore, to achieve the stated accuracy bound, the analyst can therefore set 
\[
s_t = \begin{cases}
    +1 & \text{where}~\gamma\Lambda_t(1-\pstar_t) > \pstar_t \\
    -1 & \text{otherwise} ,
\end{cases}
\]
achieving continuation at
$p_t \leq \max(\pstar_t/(\gamma\Lambda_t),\;
1 - \gamma\Lambda_t(1-\pstar_t)) = \bar{p}_t$.

\medskip
\noindent \textbf{Sample complexity.}
Condition on $\calG_0 := \{\eta_1 \leq \mu_1\} \supset \calG$,
which holds with probability at least $1-\beta/4$.
On $\calG_0$, $\mu_1 - \eta_1 \geq 0$ and
$e^{\eps_t(\mu_1-\eta_1)} \leq e^{\eps_t\mu_1} = (4/\beta)^{(1+2/\theta)\eps_t/\eps}$.
Write $w_t := e^{\eps_t\eta_{2,t}} \sim \Par(\sigma_t)$ with
$\sigma_t = 1/(\eps_t\tau_2) = \eps/(\eps_t(\theta+2))$,
independently across steps by \Cref{lem:ParetoFacts}~(1).
 
We show that $\lambda_t \leq \tilde{A}_t \cdot w_t$ on $\calG_0$
for all $t$, where $\tilde{A}_t$ is defined below.
 
On the direct path ($s_t = +1$):
$\lambda_t = \rho_t\,e^{\eps_t(\mu_1-\eta_1)}\,w_t
\leq \rho_t\,e^{\eps_t\mu_1}\,w_t$.
By the base rate upper bound (\Cref{lem:ParameterSelection}),
$\rho_t\,e^{\eps_t\mu_1}
= O(e^{\eps_t\mu_1}\ln(t/\beta)/(\pstar_t(\gamma-1)^2))$.
 
On the flipped path ($s_t = -1$):
on $\calG_0$, $\lambda_t = \rho_t\,e^{-\eps_t(\mu_1-\eta_1+\eta_{2,t})} \leq \rho_t$
since $\mu_1 - \eta_1 \geq 0$ and $\eta_{2,t} \geq 0$.
By the base rate upper bound,
$\rho_t = O(\ln(t/\beta)/((1-\pstar_t)(\gamma-1)^2))$.
Since $w_t \geq 1$, we have $\lambda_t \leq \rho_t \leq \rho_t\,w_t$.
 
Define
\[
  \tilde{A}_t \;:=\; C_A \cdot \frac{\ln(t/\beta)}{(\gamma-1)^2}
  \cdot \max\!\left(\frac{e^{\eps_t\mu_1}}{\pstar_t},\;
  \frac{1}{1-\pstar_t}\right).
\]
On the direct path, $\rho_t\,e^{\eps_t\mu_1}
= O(\frac{e^{\eps_t\mu_1}\ln(t/\beta)}{\pstar_t(\gamma-1)^2})$, and on the flipped path, $\rho_t
= O(\frac{\ln(t/\beta)}{(1-\pstar_t)(\gamma-1)^2})$. Let $C'$ be the constant for which both of these big-Oh expressions are explicit upper bounds. It follows that for $C_A = \max(\gamma\, C',2)$, $\lambda_t \leq \tilde{A}_t\,w_t$.
 
Substituting $e^{\eps_t\mu_1} = (4/\beta)^{(1+2/\theta)\eps_t/\eps}$
gives the stated expression for $\tilde{A}_t$; since this bound holds conditioned on $\calG_0$, it holds on $\calG$.

\medskip
\noindent \textbf{Failure probability accounting.}
We apply a union bound across all points of failure - the probability of the event $\calG$ not holding, and the probability of a type I or type II error at any time-step.
\[
  \underbrace{\Pr[\calG^c]}_{\leq\,\beta/2}
  \;+\; \underbrace{\sum_{t=1}^{\infty} \Pr[\text{Type I Error at }t \mid \calG]}_{\leq\,\beta/4}
  \;+\; \underbrace{\sum_{t=1}^{\infty}
    \Pr[\text{Type II Error at } t \mid \calG]}_{\leq\,\beta/4}
  \;\leq\; \beta.
\]
The bound on the first term comes from the application of \cref{lem:NoiseControl} at the beginning of this proof, and the bounds on the second and third terms come from the application of \cref{lem:ParameterSelection} with type I and type II failure probabilities $\beta_t/4$.
\end{proof}

\begin{lem}[Sample complexity]\label{lem:sample_complexity}
Under the hypotheses and parameter settings of \Cref{thm:GlobalUtility},
assume $\beta \leq 1/e$, and
let $N_{\leq T} := \sum_{t=1}^T N_t$ denote the total number of evaluations through step $T$.
The following hold.

\begin{enumerate}
\item \textbf{Per-step bound.}
Conditioned on $\calG$, for all $t \geq 1$, the Poisson mean satisfies
$\lambda_t \leq \bar\lambda_t$, where
\[
  \bar\lambda_t
  \;:=\; \max\!\left(
  \frac{C_A\,\ln(t/\beta)}{(\gamma-1)^2}
  \cdot \max\!\left(\frac{\Lambda_t}{\pstar_t},\;
  \frac{1}{1-\pstar_t}\right),\;\;
  \ln(4/\beta_t)\right).
\]
In particular, $\Pr[N_t > e\,\bar\lambda_t \mid \calG] \leq \beta_t/4$.

\item \textbf{Amortized bound.}
Assume additionally that $\eps_t = \eps_1$ and $\pstar_t = \pstar$ are constant across steps,
that $\sigma := \eps/(\eps_1(\theta+2)) > 1$,
and that $T \geq 2\ln(16/\beta)$.
Define $\Lambda_0 := (4/\beta)^{(1+2/\theta)\eps_1/\eps}$ and
\[
  R \;:=\; \frac{C_A\,\ln(T/\beta)}{(\gamma-1)^2}
  \cdot \max\!\left(\frac{\Lambda_0}{\pstar},\;
  \frac{1}{1-\pstar}\right).
\]
Then with probability at least $1 - \beta/2$,
\[
  \frac{N_{\leq T}}{T}
  \;\leq\;
    \frac{2\sigma}{\sigma-1}\,R
    \;+\; O\!\left(R \cdot T^{1/\sigma - 1}\,\beta^{-1/\sigma}\right).
\]
\end{enumerate}
\end{lem}

\begin{proof}
Throughout, $C_A = \max(\gamma\,C', 2)$ is the constant defined in the proof of \Cref{thm:GlobalUtility}, where $C'$ is the base rate upper bound constant from \Cref{lem:ParameterSelection}. In particular, $C_A \geq \gamma\,C'$ and $C_A \geq 2$.

\medskip\noindent\textbf{Item 1.}
We show $\lambda_t \leq \frac{C_A\,\ln(t/\beta)}{(\gamma-1)^2}\max(\Lambda_t/\pstar_t,\; 1/(1-\pstar_t))$ on $\calG$, which implies $\lambda_t \leq \bar\lambda_t$ since $\bar\lambda_t$ is defined as the maximum of this expression and $\ln(4/\beta_t)$.

From the proof of \Cref{thm:GlobalUtility},
conditioned on $\calG$,
$\lambda_t \leq \rho_t\Lambda_t$ when $s_t = +1$
and $\lambda_t \leq \rho_t$ when $s_t = -1$.
By the base rate upper bound from \Cref{lem:ParameterSelection},
$\rho_t \leq C'\ln(t/\beta)/((\gamma-1)^2\pstar_t)$ when $s_t = +1$
and $\rho_t \leq C'\ln(t/\beta)/((\gamma-1)^2(1-\pstar_t))$ when $s_t = -1$.

\emph{Case $s_t = +1$.}
The bound on $\lambda_t$ is:
\[
  \lambda_t \;\leq\; \rho_t\Lambda_t
  \;\leq\; \frac{C'\,\Lambda_t\,\ln(t/\beta)}{\pstar_t\,(\gamma-1)^2}.
\]
The flipping rule selects $s_t = +1$ only when $\gamma\Lambda_t(1-\pstar_t) > \pstar_t$,
which rearranges to $1/(1-\pstar_t) < \gamma\Lambda_t/\pstar_t$.
Therefore the max is controlled by the first term:
\[
  \max\!\left(\frac{\Lambda_t}{\pstar_t},\; \frac{1}{1-\pstar_t}\right)
  \;\leq\; \gamma \cdot \frac{\Lambda_t}{\pstar_t}.
\]
Since $C_A \geq \gamma\,C'$, we conclude
$\lambda_t \leq \frac{C'\,\Lambda_t\,\ln(t/\beta)}{\pstar_t(\gamma-1)^2}
\leq \frac{C_A\,\ln(t/\beta)}{(\gamma-1)^2} \cdot \frac{\Lambda_t}{\pstar_t}
\leq \frac{C_A\,\ln(t/\beta)}{(\gamma-1)^2}\max(\Lambda_t/\pstar_t,\; 1/(1-\pstar_t))$.

\emph{Case $s_t = -1$.}
The bound on $\lambda_t$ is:
\[
  \lambda_t \;\leq\; \rho_t
  \;\leq\; \frac{C'\,\ln(t/\beta)}{(1-\pstar_t)\,(\gamma-1)^2}.
\]
The flipping rule selects $s_t = -1$ when $\gamma\Lambda_t(1-\pstar_t) \leq \pstar_t$,
which rearranges to $\Lambda_t/\pstar_t \leq 1/(\gamma(1-\pstar_t)) \leq 1/(1-\pstar_t)$.
Therefore the max equals the second term:
\[
  \max\!\left(\frac{\Lambda_t}{\pstar_t},\; \frac{1}{1-\pstar_t}\right)
  \;=\; \frac{1}{1-\pstar_t}.
\]
Since $C_A \geq C'$, we conclude
$\lambda_t \leq \frac{C'\,\ln(t/\beta)}{(1-\pstar_t)(\gamma-1)^2}
\leq \frac{C_A\,\ln(t/\beta)}{(\gamma-1)^2} \cdot \frac{1}{1-\pstar_t}
= \frac{C_A\,\ln(t/\beta)}{(\gamma-1)^2}\max(\Lambda_t/\pstar_t,\; 1/(1-\pstar_t))$.

\emph{Tail bound.}
Since $\bar\lambda_t \geq \ln(4/\beta_t)$ by definition,
and $N_t \sim \Po(\lambda_t)$ with $\lambda_t \leq \bar\lambda_t$ on $\calG$,
by stochastic dominance $N_t \preceq \Po(\bar\lambda_t)$ on $\calG$.
By \Cref{lem:PoissonChernoffMultiplicative} (upper tail),
$\Pr[\Po(\bar\lambda_t) > e\,\bar\lambda_t] \leq e^{-\bar\lambda_t}$.
Since $\bar\lambda_t \geq \ln(4/\beta_t)$,
$e^{-\bar\lambda_t} \leq \beta_t/4$.
Therefore $\Pr[N_t > e\,\bar\lambda_t \mid \calG] \leq \beta_t/4$.

\medskip\noindent\textbf{Item 2.}
The proof proceeds in four stages: we use the Pareto decomposition from \Cref{thm:GlobalUtility} to write $\sum_t \lambda_t$ as a weighted sum of i.i.d.\ Pareto random variables, then apply truncation, Chebyshev, and Poisson concentration.

\emph{Setup.}
By Item~3 of \Cref{thm:GlobalUtility},
on the event $\calG_0 = \{\eta_1 \leq \mu_1\}$
(which holds with probability at least $1-\beta/4$),
we have $\lambda_t \leq \tilde{A}_t\,w_t$ for all $t$,
where $w_t \sim \Par(\sigma_t)$ independently across steps.
Since $\eps_t = \eps_1$ is constant,
$\sigma_t = \eps/(\eps_1(\theta+2)) = \sigma$ for all $t$,
so the $w_t$ are i.i.d.\ $\Par(\sigma)$.
Since $\pstar_t = \pstar$ is constant and
$(4/\beta)^{(1+2/\theta)\eps_t/\eps} = (4/\beta)^{(1+2/\theta)\eps_1/\eps} = \Lambda_0$
for all $t$, the expression for $\tilde{A}_t$ from \Cref{thm:GlobalUtility} specializes to
\[
  \tilde{A}_t \;=\; \frac{C_A\,\ln(t/\beta)}{(\gamma-1)^2}
  \cdot \max\!\left(\frac{\Lambda_0}{\pstar},\;
  \frac{1}{1-\pstar}\right).
\]
We bound $\tilde{A}_t$ from above and below.
First, since $\ln(t/\beta) \leq \ln(T/\beta)$ for $t \leq T$,
\[
  \tilde{A}_t \;\leq\; \frac{C_A\,\ln(T/\beta)}{(\gamma-1)^2}
  \cdot \max\!\left(\frac{\Lambda_0}{\pstar},\;
  \frac{1}{1-\pstar}\right) \;=\; R.
\]
Second, since $C_A \geq 2$,
$(\gamma-1)^2 \leq 1$ (as $\gamma \in (1,2]$),
$\max(\Lambda_0/\pstar,\; 1/(1-\pstar)) \geq 1$
(as $\Lambda_0 \geq 1$ and $\pstar \leq 1$),
and $\ln(t/\beta) \geq \ln(1/\beta) \geq 1$ (as $\beta \leq 1/e$),
we have
\[
  \tilde{A}_t \;\geq\; C_A \cdot 1 \cdot 1 \cdot 1 \;\geq\; 2.
\]
Define $B := \sum_{t=1}^T \tilde{A}_t$ and $A_{\max} := \max_{t \leq T} \tilde{A}_t$.
By the two properties above, $A_{\max} \leq R$ and $B \leq T\,R$.

\emph{Truncation.}
Define $W := (16T/\beta)^{1/\sigma}$ and $\tilde{w}_t := \min(w_t, W)$.
Since $\Pr[w_t > W] = W^{-\sigma} = \beta/(16T)$ for each $t$,
a union bound gives
$\Pr[\exists\,t \leq T : w_t > W] \leq T \cdot \beta/(16T) = \beta/16$.
On the complementary event, $w_t = \tilde{w}_t$ for all $t \leq T$, so
$\sum_{t=1}^T \lambda_t \leq \sum_{t=1}^T \tilde{A}_t\,w_t = \sum_{t=1}^T \tilde{A}_t\,\tilde{w}_t$.

\emph{Chebyshev.}
Since the $\tilde{w}_t$ are i.i.d.\ and the $\tilde{A}_t$ are deterministic on $\calG_0$:
\begin{align*}
  \Ex\!\left[\sum_{t=1}^T \tilde{A}_t\,\tilde{w}_t\right]
  &\;=\; \sum_{t=1}^T \tilde{A}_t\,\Ex[\tilde{w}_t]
  \;\leq\; \frac{\sigma}{\sigma-1}\,B,
\end{align*}
where we used $\Ex[\tilde{w}_t] \leq \sigma/(\sigma-1)$ from \Cref{lem:ParetoFacts}~(3a).
For the variance, since the $\tilde{w}_t$ are independent:
\begin{align*}
  \Var\!\left(\sum_{t=1}^T \tilde{A}_t\,\tilde{w}_t\right)
  &\;=\; \sum_{t=1}^T \tilde{A}_t^2\,\Var(\tilde{w}_t)
  \;\leq\; \Var(\tilde{w}) \cdot A_{\max} \cdot \sum_{t=1}^T \tilde{A}_t
  \;=\; \Var(\tilde{w}) \cdot A_{\max} \cdot B,
\end{align*}
where $\Var(\tilde{w}) := \Var(\tilde{w}_1)$ is bounded by \Cref{lem:ParetoFacts}~(3b).
By Chebyshev's inequality,
\[
  \Pr\!\left[\sum_{t=1}^T \tilde{A}_t\,\tilde{w}_t
  > \frac{\sigma}{\sigma-1}\,B + u\right]
  \;\leq\; \frac{\Var(\tilde{w}) \cdot A_{\max} \cdot B}{u^2}.
\]
Setting the right-hand side equal to $\beta/16$ and solving:
$u := 4\sqrt{\Var(\tilde{w}) \cdot A_{\max} \cdot B / \beta}$.
Since $A_{\max} \leq R$ and $B \leq T\,R$,
$u^2 \leq 16\,\Var(\tilde{w})\,T\,R^2/\beta$.
By \Cref{lem:ParetoFacts}~(3b),
$\Var(\tilde{w}) \leq \frac{2}{2-\sigma}W^{2-\sigma}$ when $1 < \sigma < 2$,
$\Var(\tilde{w}) \leq 2\ln W + 1$ when $\sigma = 2$,
and $\Var(\tilde{w}) \leq \sigma/(\sigma-2)$ when $\sigma > 2$.
Substituting $W = (16T/\beta)^{1/\sigma}$ and dividing by $T$:
$u/T = O(R \cdot T^{1/\sigma - 1}\,\beta^{-1/\sigma})$
when $1 < \sigma < 2$ (and strictly smaller when $\sigma \geq 2$).

\emph{Poisson concentration.}
Recall $B := \sum_{t=1}^T \tilde{A}_t$ and define $M := \frac{\sigma}{\sigma-1}B + u$.
On the event that $\sum_{t=1}^T \lambda_t \leq M$
(which holds with probability at least $1 - \beta/16 - \beta/16$
by the truncation and Chebyshev bounds),
the total sample count
$N_{\leq T} = \sum_{t=1}^T N_t \sim \Po(\sum_{t=1}^T \lambda_t) \preceq \Po(M)$
by Poisson additivity and stochastic dominance.
Since $T \geq 2\ln(16/\beta)$,
we have $M/4 \geq T/2 \geq \ln(16/\beta)$, so
$e^{-M/4} \leq e^{-T/2} \leq \beta/16$.
Therefore $\Pr[\Po(M) > 2M] \leq e^{-M/4} \leq \beta/16$,
where the first inequality is by \Cref{lem:PoissonChernoffBounds} (upper tail) with $u = M$.
On the complementary event:
\[
  \frac{N_{\leq T}}{T} \;\leq\; \frac{2M}{T}
  \;=\; \frac{2\sigma}{\sigma-1} \cdot \frac{B}{T} + \frac{2u}{T}
  \;\leq\; \frac{2\sigma}{\sigma-1}\,R
  + O\!\left(R \cdot T^{1/\sigma-1}\,\beta^{-1/\sigma}\right).
\]

\emph{Failure probability.}
The total failure probability is the sum of four terms:
$\Pr[\calG_0^c] \leq \beta/4$,
$\Pr[\text{truncation fails}] \leq \beta/16$,
$\Pr[\text{Chebyshev fails}] \leq \beta/16$, and
$\Pr[\Po(M) > 2M] \leq \beta/16$.
The total is at most $\beta/4 + 3\beta/16 < \beta/2$.
\end{proof}

\subsection{\GTM with Ex-post privacy guarantee}

The drawback of the choice of parameters $\tau_1$, $\tau_{2,t}$, and $\mu_1$ in \Cref{thm:GlobalUtility} is the worst-case polynomial decay in accuracy forced by the $\Lambda_t$  factor, which equals essentially $(t/\beta)^{\Theta(\eps_t/\eps)}$ (see~\Cref{thm:GlobalUtility}). In this subsection, we present a different choice of parameters leading to different tradeoffs.
More specifically, \Cref{thm:ExPost} below circumvents this poor scaling with the length of the stream by allowing the ratio of the output privacy parameter $\eps$ to the input privacy parameter $\eps_t$ to scale as $\ln (t/\beta)$. With this scaling, the $\Lambda_t$ slack factor is bounded from above by a constant.
We describe \Cref{thm:ExPost} as an ex-post privacy loss guarantee, but it can be achieved in two ways: either (a) the input privacy parameter is forced to decay logarithmically or (b) the output privacy parameter is allowed to grow logarithmically. Both these settings have natural applications. 

The setting in (a) is most reasonable for our application in reducing optimization under continual observation to the batch setting, as described later in Section~\ref{sec:batchtocontinual}, and gives us a standard uniform privacy loss bound. 

The  setting in (b) is the standard setting for ex-post privacy loss guarantees, and is well suited for potential applications in hyperparameter optimization. It allows \Cref{alg:GTM} to achieve a privacy loss bound that scales only logarithmically with the length of the stream. Further, if the whole stream is processed without halting at any step, then the privacy loss is constant instead of growing with the input length.

\begin{theorem}[Ex-post privacy guarantee]\label{thm:ExPost}
Fix a gap parameter $\gamma \in (1,2]$, a failure probability $\beta \in (0,1)$, privacy parameters $\eps_C$ and a constant $C_H > 0$. Define 
\[
  \Lambda_t \;:=\; \left(\frac{4}{\beta}\right)^{\!\eps_t/\eps_C} \cdot e^{2/C_H}.
\]
Define the rejection threshold
\[
  \bar{p}_t \;:=\; \max\!\left(\frac{\pstar_t}{\gamma\Lambda_t},\;\; 1 - \gamma\Lambda_t(1-\pstar_t)\right).
\]

There exists a parameter setting for Algorithm~\ref{alg:GTM} for which we achieve an \emph{ex-post} privacy guarantee with privacy loss function \[\tilde{\eps}((a_1,\dots,a_t)) = \begin{cases} \eps_C + C_H \eps_t \ln (4/\beta_t) &\mbox{ if }a_t = s_t \\ \eps_C &\mbox{ if }a_t = -s_t \end{cases}.\] 
Further, with probability $1-\beta$, for all $t\geq1$,
\begin{enumerate}
\item \textbf{Halting.}
At every step $t$ where $p_t \geq \pstar_t$, the algorithm halts.

\item \textbf{Continuation.}
The algorithm continues past every step $t$ satisfying $p_t \leq \bar{p}_t$.

\item \textbf{Sample complexity.}
The number of evaluations of the mechanism $\calM_t$ at each step satisfies
\[
  N_t \;\leq\;
  \begin{cases}
    \displaystyle O\!\left(\frac{\Lambda_t\,\ln(t/\beta)}{ \pstar_t\,(\gamma-1)^2}\right)
    & \text{if } \gamma\Lambda_t(1-\pstar_t) > \pstar_t, \\[6pt]
    \displaystyle O\!\left(\frac{\ln(t/\beta)}{(1-\pstar_t)\,(\gamma-1)^2}\right)
    & \text{if } \gamma\Lambda_t(1-\pstar_t) \leq \pstar_t.
  \end{cases}
\]
\end{enumerate}
\end{theorem}

Note that, as in the case of \Cref{thm:GlobalUtility}, the gap parameter $\gamma$ controls the trade-off between the rejection threshold and the per-step sample complexity. The sample complexity grows as $\ln(t/\beta)$ per step, identical to \Cref{thm:GlobalUtility}; what improves is the gap between the thresholds $\pstar_t$ and $\bar{p}_t$.

\begin{proof}[Proof of \Cref{thm:ExPost}]
The proof follows the same structure as \Cref{thm:GlobalUtility}, with the following parameter changes. In the initialization of \Cref{alg:GTM}, set $1/\tau_1 = \eps_C$ and $\mu_1 = \ln(4/\beta)/\eps_C = \tau_1 \ln(4/\beta)$, so that $\Pr[\eta_1 > \mu_1] = \beta/4$. At each step $t$, set $\tau_{2,t} = 2/(C_H \eps_t\ln(4/\beta_t))$.

\medskip
\noindent \textbf{Privacy.}
The result follows using \Cref{lem:GeneralizedPrivacy} and observing that $1/\tau_1 = \eps_C$ and $1/\tau_1 + 2/\tau_{2,t} = \eps_C + C_H \eps_t\ln(4/\beta_t)$ for any $t \ge 1$.

\medskip
\noindent \textbf{Bounding $\lambda_t$.}
Apply \Cref{lem:NoiseControl} with $c_\calG = 1/4$. Since $e^{-\mu_1/\tau_1} = \beta/4$ and $c_\calG \sum_t \beta_t = \beta/4$, we have $\Pr[\calG^c] \leq \beta/2$. The choice of $\tau_{2,t}$ ensures exact cancellation:
\[
  \eps_t\,\tau_{2,t}\,\ln(4/\beta_t) \;=\; \frac{2}{C_H}, \qquad 
  \Lambda_t \;=\; \exp\!\bigl(\eps_t(\mu_1 + \tau_{2,t}\ln(1/(c_\calG \beta_t)))\bigr) \;=\; \left(\frac{4}{\beta}\right)^{\!\eps_t/\eps_C} \cdot e^{2/C_H}.
\]
\Cref{lem:NoiseControl} (part~(\ref{item:noiseConditionedonG})) shows that conditioned on $\calG$, for all time-steps $t$ and any choice of $\rho_t$, if $s_t = +1$ then $\lambda_t \in [\rho_t,\rho_t\,\Lambda_t]$, and if $s_t = -1$ then $\lambda_t \in [\rho_t/\Lambda_t, \rho_t]$.
Further, part~(\ref{item:haltingRange}) in \Cref{lem:NoiseControl} shows that, when $\eta_1 \leq \mu_1$, then $\lambda_t \geq \rho_t$ when $s_t = +1$ and $\lambda_t \leq \rho_t$ when $s_t = -1$.

\medskip
\noindent \textbf{Accuracy.} We analyze the case depending on $s_t=-1$ or $s_t=+1$ separately. 
When $s_t = +1$, then applying \Cref{lem:ParameterSelection} (case (\ref{item:aboveThreshold})) with target $\pstar_t$, error parameters $(\beta_t/6, \beta_t/6)$, gap $\gamma$, parameter $\Lambda_t$ gives us that there exists an integer $c_t$ such that for $\rho_t = \lamL{c_t}{\beta_t/6}/\pstar_t$, we have
\begin{enumerate}
    \item Halting: if $\lambda_t \geq \rho_t$ and $p_t \geq \pstar_t$, then $\Pr[K_t \leq c_t] \leq \beta_t/6$.
    \item Continuation: if $\lambda_t \leq \rho_t\Lambda_t$ and $p_t \leq \pstar_t/(\gamma\Lambda_t)$, then $\Pr[K_t \geq c_t + 1] \leq \beta_t/6$.
    \item Base rate: $\rho_t = O(\ln(6/\beta_t)/((\gamma-1)^2\pstar_t))$.
\end{enumerate}

Similarly, if $s_t = -1$, then \Cref{lem:ParameterSelection} (case~(\ref{item:belowThreshold})) applied to the negated mechanism with target $1-\pstar_t$, error parameters $(\beta_t/6, \beta_t/6)$, gap $\gamma$, parameter $\Lambda_t$ gives us that there exists an integer $c_t$ such that for $\rho_t = \lamR{c_t}{\beta_t/6}/(1-\pstar_t)$, we have
\begin{enumerate}
    \item Halting: if $\lambda_t \leq \rho_t$ and $1 - p_t \leq 1 - \pstar_t$, then $\Pr[K_t \geq c_t+1] \leq \beta_t/6$.
    \item Continuation: if $\lambda_t \geq \rho_t/\Lambda_t$ and $1 - p_t \geq \gamma\Lambda_t(1 - \pstar_t)$, then $\Pr[K_t \leq c_t] \leq \beta_t/6$.
    \item Base rate: $\rho_t = O(\ln(6/\beta_t)/((\gamma-1)^2(1-\pstar_t)))$.
\end{enumerate}

\medskip
\noindent\textbf{Flipping rule.}
The direct path guarantees continuation at $p_t \leq \pstar_t/(\gamma\Lambda_t)$ and the flipped path at $p_t \leq 1 - \gamma\Lambda_t(1-\pstar_t)$. When $\eps_t$ and $\pstar_t$ are constant, $\Lambda_t$ is time-independent, the flipping rule is the same at every step. Writing $x = \gamma\Lambda_t \geq 1$,
\[
  \frac{\pstar_t}{x} - \bigl(1 - x(1-\pstar_t)\bigr)
  \;=\; \frac{(1-\pstar_t)(x-1)(x - \pstar_t/(1-\pstar_t))}{x} = \frac{(1-\pstar_t)(x-1)}{x}\left(x - \frac{\pstar_t}{(1-\pstar_t)}\right).
\]

Since $x \geq 1$, the sign is determined by $x - \pstar_t/(1-\pstar_t)= \gamma \Lambda_t - \pstar_t/(1-\pstar_t)$. It follows that the direct path has a higher rejection threshold when $\gamma\Lambda_t(1-\pstar_t) \geq \pstar_t$, and the flipped path is better otherwise. The analyst sets $s_t = +1$ when $\gamma\Lambda_t(1-\pstar_t) > \pstar_t$ and $s_t = -1$ otherwise, achieving continuation at $p_t \leq \max(\pstar_t/(\gamma\Lambda_t),\; 1 - \gamma\Lambda_t(1-\pstar_t)) = \bar{p}_t$.

\medskip
\noindent \textbf{Sample complexity.}
Conditioned on $\calG$, $N_t \sim \Po(\lambda_t)$ with $\lambda_t \leq \bar\lambda_t$, where 
\[
\bar \lambda_t = \begin{cases}
    \rho_t\Lambda_t & s_t = +1 \\
    \rho_t & s_t =-1.
\end{cases}
\]

We claim that $\bar\lambda_t \geq \ln(6/\beta_t)$ in both cases. On the direct path, $$\bar\lambda_t \geq \lamL{c_t}{\beta_t/6}/\pstar_t \geq \lamL{c_t}{\beta_t/6} \geq c_t + \ln(6/\beta_t) \geq \ln(6/\beta_t),$$ where the third inequality is from \eqref{eq:lamL_def} and the fourth uses $c_t + 1 > 2\ln(6/\beta_t)$ from \Cref{lem:PoissonOverhead}(2). 

On the flipped path (selected only when $\pstar_t \geq 1/2$), the gap condition of \Cref{lem:PoissonOverhead}(1) gives $\lamR{c_t}{\beta_t/6} \geq \lamL{c_t}{\beta_t/6}/\gamma$, and since $\lamL{c_t}{\beta_t/6} \geq 2\ln(6/\beta_t)$ and $\gamma \leq 2$, we have 
    \[\bar\lambda_t = \lamR{c_t}{\beta_t/6}/(1-\pstar_t) \geq \ln(6/\beta_t)/(1-\pstar_t) \geq \ln(6/\beta_t).\]

Conditioned on $\calG$, $\lambda_t \leq \bar{\lambda}_t$, so by stochastic dominance of $\Po(\lambda_t)$ by $\Po(\bar{\lambda}_t)$, and the upper tail of the multiplicative Poisson Chernoff bound (\Cref{lem:PoissonChernoffMultiplicative}), $\Pr[N_t \geq e\,\bar\lambda_t \mid \calG] \leq e^{-\bar\lambda_t} \leq \beta_t/6$. Substituting the base rates, we get
\[
N_t = \begin{cases}
    O(\rho_t\Lambda_t) = \displaystyle O\left(\frac{\Lambda_t\ln(t/\beta)}{ \pstar_t(\gamma-1)^2}\right) & s_t = +1 \\
    \displaystyle  O\left(\frac{\ln(t/\beta) }{(1-\pstar_t)(\gamma-1)^2} \right) & s_t =-1.
\end{cases}
\]

\medskip
\noindent \textbf{Failure probability accounting.}
We apply a union bound across all points of failure - the probability of the event $\calG$ not holding, the probability of a type I or type II error at any time-step, and the probability of the sample complexity bound failing.
\[
  \underbrace{\Pr[\calG^c]}_{\leq\,\beta/2}
  \;+\; \underbrace{\sum_{t=1}^{\infty} \Pr[\text{Type I Error at }t \mid \calG]}_{\leq\,\beta/6}
  \;+\; \underbrace{\sum_{t=1}^{\infty}
    \Pr[\text{Type II Error at } t \mid \calG]}_{\leq\,\beta/6}
  \;+\; \underbrace{\sum_{t=1}^\infty \Pr[N_t \geq e\,\bar\lambda_t \mid \calG]}_{\leq\,\beta/6}
  \;\leq\; \beta.
\]
\end{proof}
\section{Purification}

In this section, we formally prove a simple purification lemma that allows us to convert approximate DP mechanisms to pure ones with minimal perturbation to their accuracy. As mentioned in the introduction, this technique of using randomized response as post-processing to purify approximate-DP mechanisms with finite co-domains is attributed to folklore by \cite{DBLP:journals/corr/abs-2211-11189}. We give complete proofs for ease of reference.

\begin{lem}[Purification]
\label{lem:ApproxToPure}
Let $\calM : \calX \to \{+1,-1\}$ be an $(\varepsilon_{1}, \delta_{1})$-DP mechanism, $X,X'\in\calX$ a pair of neighboring data sets, and 
\[ p = \Pr[\calM(X) = +1] \qquad p' = \Pr[\calM(X') = +1].\] 
Define the $\phi$-smoothed mechanism $\widetilde{\calM}:\calX \to \{+1,-1\}$ that for any input dataset $X''$ evaluates $a \gets \calM(X'')$ and passes $a$ through a binary symmetric channel with crossover probability $\phi$. That is, we draw an independent bit $b \sim \Ber(\phi)$ and output $-a$ if $b = 1$, and $a$ if $b = 0$. Let 
\[ \widetilde{p} = \Pr[\widetilde{\calM}(X) = +1] \qquad \widetilde{p}' = \Pr[\widetilde{\calM}(X') = +1].\]. The following statements are true:
\begin{enumerate}
    \item $\widetilde{p} = p + \phi(1 - 2p)$, and if $\phi\leq 1/2$, then $\widetilde{p} \geq \widetilde{p}' \Leftrightarrow p \geq p'$.
    \item Assuming $\phi \leq 1/2$, if for some  $c \ge 1$,
        \[ \phi \geq \frac{\delta_{1}}{e^{c\varepsilon_{1}} - 1 + 2\delta_{1}}\]
    then the $\phi$-smoothed mechanism $\widetilde{\calM}$ is $(c\cdot \varepsilon_{1},0)$-DP.
    \item In particular, setting 
    \[\phi=\frac{\delta_1}{e^{\eps_1} - 1 + 2 \delta_1}\] 
    we have that $\widetilde{\calM}$ is $\eps_1$-DP and that $|\widetilde{p} - p|  \leq \delta_1/\eps_1$.
\end{enumerate}
\end{lem}

\begin{proof}
\begin{enumerate}
    \item By the definition of $\widetilde{\calM}$, the output probabilities are convex combinations:
\begin{align}
    \begin{split}
        \widetilde{p} &= (1-\phi)p + \phi(1-p) = (1-2\phi)p + \phi.
    \end{split}
\label{eq:convex_combination_probabilities}
\end{align}
    The mapping $p\mapsto\widetilde{p}$ is affine. If $\phi\leq 1/2$, then the coefficient of $p$ is non-negative, and this mapping is order preserving, so $\widetilde{p} \geq \widetilde{p}' \Leftrightarrow p \geq p'$.
    \item Since $\calM$ is $(\eps_1,\delta_1)$-DP, $p \leq e^{\eps_1} p' + \delta_{1}$. Since $\phi \leq 1/2$, we are guaranteed that $1-2\phi\geq0$, and it follows that
    \begin{align*}
        \widetilde{p} &\leq (1-2\phi)(e^{\eps_1} p' + \delta_1) + \phi \\
        &= e^{\eps_1} \big((1-2\phi)p'\big) + \delta_1(1-2\phi) + \phi.
    \end{align*}
    Using the relation $(1-2\phi)p' = \widetilde{p}' - \phi$, we rewrite this entirely in terms of the smoothed probabilities:
    \begin{align*}
        \widetilde{p} &\leq e^{\eps_1} (\widetilde{p}' - \phi) + \delta_1(1-2\phi) + \phi \\
        &= e^{\eps_1}\widetilde{p}' + \delta_1(1-2\phi) - \phi(e^{\eps_1} - 1).
    \end{align*}
    To ensure strict $c\eps_1$-indistinguishability for any $c \ge 1$, we require $\widetilde{p} \leq e^{c\eps_1} \widetilde{p}'$. Thus, it is sufficient that the residual terms satisfy
    $$ \delta_1(1-2\phi) - \phi(e^{\eps_1} - 1) \leq (e^{c\eps_1} - e^{\eps_1})\widetilde{p}'. $$
    Since $c \geq 1$, the right-hand side is non-negative and monotonically increasing with respect to $\widetilde{p}'$. Therefore, it suffices for the inequality to hold at the minimum possible valid value, which is $\widetilde{p}' = \phi$.
    \begin{align*}
    \delta_1(1-2\phi) - \phi(e^{\eps_1} - 1) &\leq (e^{c\eps_1} - e^{\eps_1})\phi \\
    \Leftrightarrow \delta_1 - 2\phi\delta_1 - \phi e^{\eps_1} + \phi &\leq \phi e^{c\eps_1} - \phi e^{\eps_1} \\
    \Leftrightarrow \delta_1 &\leq \phi(e^{c\eps_1} - 1 + 2\delta_1).
    \end{align*}
    Rearranging for $\phi$ yields the generalized bound.
    \item Substituting $c=1$ gives us the relevant expression for $\phi$. The bound on the perturbation follows from a direct calculation:
        \[ |\tilde{p} - p| = |\phi(1-2p)| \leq \frac{\delta_1}{e^{\eps_1} - 1 + 2 \delta_1} \leq \delta_1/\eps_1.\]
\end{enumerate}
\end{proof}

\begin{cor}[Approximate DP inputs]\label{cor:approxDP}
Suppose each mechanism $\calM_t$ is $(\eps_t, \delta_t)$-DP with $\delta_t \geq 0$. For each $t$, apply \Cref{lem:ApproxToPure} with $c=1$ and crossover $\phi_t = \delta_t/(e^{\eps_t} - 1 + 2\delta_t)$ to obtain a pure $\eps_t$-DP mechanism $\widetilde{\calM}_t$. Run \Cref{alg:GTM} on the purified stream with targets $(\pstar_t)_{t\geq1}$. Then \Cref{thm:GlobalUtility,thm:ExPost} apply to the purified stream, and their guarantees transfer to the original success probabilities $p_t$ as follows. Write $\phi_t < \delta_t/\eps_t$ for the purification slack at step $t$.

\begin{enumerate}
\item \textbf{Halting} (\Cref{thm:GlobalUtility,thm:ExPost}). The algorithm halts at every step $t$ where $p_t \geq \pstar_t + \phi_t$.

\item \textbf{Continuation} (\Cref{thm:GlobalUtility,thm:ExPost}). The algorithm continues past every step $t$ where $p_t \leq \bar{p}_t - \phi_t$, where $\bar{p}_t$ is the rejection threshold from the relevant theorem.

\item \textbf{Privacy and sample complexity} are unchanged.
\end{enumerate}

Since $\phi_t \leq \delta_t/\eps_t$, the perturbation to the acceptance and rejection thresholds $\pstar_t$ and $\pbar_t$ is asymptotically negligible and all bounds match their pure-DP counterparts up to lower-order terms.
\end{cor}

\begin{proof}
By \Cref{lem:ApproxToPure}(3), the purified success probability satisfies $|\widetilde{p}_t - p_t| \leq \phi_t \leq \delta_t/\eps_t$. Items~1 and~2 follow from the triangle inequality: $p_t \geq \pstar_t + \phi_t$ implies $\widetilde{p}_t \geq p_t - \phi_t \geq \pstar_t$, and $p_t \leq \bar{p}_t - \phi_t$ implies $\widetilde{p}_t \leq p_t + \phi_t \leq \bar{p}_t$. Sample complexity is unchanged since each purified evaluation requires exactly one evaluation of $\calM_t$ plus one independent coin flip.
\end{proof}
\section{Lower Bounds}\label{sec:lower_bounds}

We establish lower bounds for the generalized private tester by reduction from the classical Above Threshold problem~\cite[Chapter 3]{dwork2014algorithmic}. We first prove a per-step lower bound for Above Threshold, then lift it to a multi-step bound by an averaging argument. We finally transfer it to the generalized testing problem via the Laplace mechanism. 

\medskip
\noindent\textbf{Standard threshold test.}
Let $\eps>0, \delta \ge 0$ and let $\calX$ denote the data universe with a fixed symmetric adjacency relation. Let $X \in \calX$ be the private input of the testing mechanism.
A \emph{(standard) threshold test} receives a sensitivity-$1$ query $f : \calX \to \mathbb{R}$, a public threshold $\tau \in \mathbb{R}$, and outputs $a \in \{-1, +1\}$.
We consider a sequence of such queries $f_1, f_2, \ldots$ presented in an online fashion, where $a_t = +1$ denotes halting and $a_t = -1$ denotes continuation.
We say an $\eps$-differentially private threshold test has \emph{global accuracy} $(\alpha_+, \alpha_-, \beta)$ if, with probability at least $1 - \beta$ over the algorithm's randomness, for all $t\geq1$:
\[
    f_t(X) \geq \tau + \alpha_+ \implies a_t = +1, \qquad \qquad
    f_t(X) \leq \tau - \alpha_- \implies a_t = -1.
\]

We first show that for two datasets $X_0$ and $X_1$ such that the probability of $M(X_0)=1$ is clearly less than 1/2 and the probability of $M(X_1)=1$ is clearly above 1/2, then $X_0$ and $X_1$ must have ``large enough'' distance $d$ in the adjacency relation.

\begin{defn}
    \begin{enumerate}
        \item Let $M : \calX \to \{-1, +1\}$ be an $(\eps, \delta)$-differentially private mechanism. 
        \item Let $X_0 = Y_0, Y_1, \ldots, Y_d = X_1$ be a chain of neighboring datasets with $Y_j$ and $Y_{j+1}$ neighboring for each $j$. 
        \item Define $h_j := \Pr[M(Y_j) = +1]$, $\delta' := \delta/(e^{\eps}-1)$, $u_j := h_j + \delta'$, and $v_j := (1 - h_j) + \delta'$. Note $u_j + v_j = 1 + 2\delta'$ for all $j$, and $\delta + \delta' = e^\eps\,\delta'$.
    \end{enumerate}    
\end{defn}

\begin{lem}[Acceptance probability propagation]\label{lem:acceptance_propagation}
For all $j \in \{0, \ldots, d-1\}$:
\begin{enumerate}
\item $e^{-\eps}\, u_j \leq u_{j+1} \leq e^{\eps}\,u_j$.
\item $e^{-\eps}\, v_j \leq \,v_{j+1} \leq e^{\eps}\, v_j$.
\end{enumerate}
\end{lem}

\begin{proof}
\emph{Item 1.} Since $Y_{j}$ and $Y_{j+1}$ are neighboring, and $M$ is $(\eps,\delta)$-DP, $h_{j+1} \leq e^\eps\,h_{j} + \delta$, so $u_{j+1} = h_{j+1} + \delta' \leq e^\eps\,h_{j} + \delta + \delta' = e^\eps(h_{j} + \delta') = e^\eps\,u_{j}$. The other bound follows by flipping the role of $Y_j$ and $Y_{j+1}$.

\emph{Item 2.} Similarly, by reasoning about the $-1$ output, we have that $(1 - h_j) \leq e^\eps(1 - h_{j+1}) + \delta$, so $v_j = (1-h_j) + \delta' \leq e^\eps(1-h_{j+1}) + \delta + \delta' = e^\eps((1-h_{j+1}) + \delta') = e^\eps\,v_{j+1}$. As before, the other bound follows by flipping the role of $Y_j$ and $Y_{j+1}$.
\end{proof}

\begin{lem}\label{lem:chain_LB}
If $\Pr[M(X_0) = +1] \leq \beta^{\tI}$ and $\Pr[M(X_1) = +1] \geq 1 - \beta^{\tII}$ for some $\beta^{\tI}, \beta^{\tII} \geq 0$ with $\beta^{\tI} + \delta' < 1/2$ and $\beta^{\tII} + \delta' < 1/2$, then
\begin{align}\label{eq:chain_LB}
    d \;\geq\; \frac{1}{\eps}\left(\log\frac{1}{2(\beta^{\tI} + \delta')} + \log\frac{1}{2(\beta^{\tII} + \delta')}\right) - 1.
\end{align}
In particular, when $\delta = 0$ this reduces to $d \geq \frac{1}{\eps}(\log\frac{1}{2\beta^{\tI}} + \log\frac{1}{2\beta^{\tII}}) - 1$.
\end{lem}
\begin{proof}
By assumption, $u_0 \leq \beta^{\tI} + \delta' < 1/2$.
Since $\beta^{\tII} + \delta' < 1/2$, we have $u_d = h_d + \delta' \geq (1 - \beta^{\tII}) + \delta' > 1/2$.
Since $u_{j+1} \leq e^{\eps}\,u_j$, there exists $k \in \{0, \ldots, d-1\}$ with $u_k \leq 1/2 \leq u_{k+1}$.
From $u_k + v_k = 1 + 2\delta'$, it follows that $v_k \geq 1/2 + 2\delta' \geq 1/2$.

Applying $u_{j+1} \leq e^{\eps}\,u_j$ from $j = 0$ to $j = k$ gives
\[
    \frac{1}{2} \;\leq\; u_{k+1} \;\leq\; e^{(k+1)\eps}\,u_0 \;\leq\; e^{(k+1)\eps}(\beta^{\tI} + \delta'),
\]
so $k + 1 \geq \frac{1}{\eps}\ln\frac{1}{2(\beta^{\tI} + \delta')}$.

By assumption, $v_d \leq \beta^{\tII} + \delta'$.
Applying $v_j \leq e^{\eps}\,v_{j+1}$ from $j = k$ to $j = d-1$ gives
\[
    \frac{1}{2} \;\leq\; v_k \;\leq\; e^{(d-k)\eps}\,v_d \;\leq\; e^{(d-k)\eps}(\beta^{\tII} + \delta'),
\]
so $d - k \geq \frac{1}{\eps}\ln\frac{1}{2(\beta^{\tII} + \delta')}$.
Adding the two gives $d + 1 = (k+1) + (d-k) \geq \frac{1}{\eps}(\log\frac{1}{2(\beta^{\tI} + \delta')} + \log\frac{1}{2(\beta^{\tII} + \delta')})$.
\end{proof}

We next show a lower bound on any $(\eps, \delta)$-differentially private standard threshold test
\begin{thm}[Lower bound for standard private threshold testers]\label{thm:AT_global_LB}
Let $\calA$ be an $(\eps, \delta)$-differentially private standard threshold test with global accuracy $(\alpha_+, \alpha_-, \beta)$ over $T$ steps.
Define $\delta' := \delta/(e^{\eps} - 1)$.
If $\beta + \delta' < 1/2$, then
\begin{align}\label{eq:AT_global_LB}
    \alpha_+ + \alpha_- \;\geq\; \frac{1}{\eps}\left(\log\frac{T}{2(\beta + T\delta')} + \log\frac{1}{2(\beta + \delta')}\right) - 2.
\end{align}
\end{thm}

\begin{proof}
Fix queries $f_1, \ldots, f_T$ of sensitivity $1$ with disjoint support in the dataset, i.e. the entries of $X$ that affect $f_t$ are disjoint from those affecting $f_{t'}$ for $t \neq t'$. Recall that $\tau$ is the public threshold.
Let $X_0$ be a dataset with $f_t(X_0) = \tau - \alpha_-$ for all $t$.
For each $t \in [T]$, write $H_t (X)$ for the event that the algorithm when run on a dataset $X\in\calX$ reaches step $t$ and halts there, i.e., $a_{[t]} = (-1, \ldots, -1, +1)$.

When $\calA$ is run on $X_0$, global accuracy gives $\sum_{t=1}^T \Pr[H_t(X_0)] \leq \beta$, so by averaging there exists $t^* \in [T]$ with $\Pr[H_{t^*}(X_0)] \leq \beta/T$.

Let $X_1$ be the dataset obtained from $X_0$ by modifying only the entries in the support of $f_{t^*}$ so that $f_{t^*}(X_1) = \tau + \alpha_+$.
Since $f_{t^*}$ has sensitivity $1$, the datasets $X_0$ and $X_1$ are at distance $\lceil \alpha_+ + \alpha_- \rceil$ in the adjacency relation.
By the disjointness of the support, $f_t(X_1) = \tau - \alpha_-$ for all $t \neq t^*$.
When $\calA$ is run on $X_1$, global accuracy requires that no error occurs at any step with probability at least $1 - \beta$; since the only step with $f_t \geq \tau + \alpha_+$ is $t = t^*$, this implies $\Pr[H_{t^*}(X_1)] \geq 1 - \beta$.

Define the mechanism $M(X) := +1$ if $H_{t^*}(X)$ occurs when $\calA$ is run with dataset $X$ on the query sequence $f_1, \ldots, f_T$, and $M(X) := -1$ otherwise.
This mechanism is $(\eps, \delta)$-DP by post-processing.
Applying \Cref{lem:chain_LB} to $M$ with $\beta^{\tI} = \beta/T$, $\beta^{\tII} = \beta$, and $d = \lceil \alpha_+ + \alpha_- \rceil \leq \alpha_+ + \alpha_- + 1$ gives
\[
    \alpha_+ + \alpha_-
    \;\geq\; \frac{1}{\eps}\left(\log\frac{T}{2(\beta + T\delta')} + \log\frac{1}{2(\beta + \delta')}\right) - 2. \qedhere
\]
\end{proof}

\begin{cor}[Lower bound for generalized private testers]\label{cor:GTM_LB}
Fix a target threshold $\pstar \in (0,1)$, a constant input privacy parameter $\eps_1 > 0$, and an output privacy guarantee $(\eps, \delta)$.
Let $\calA$ be a generalized private tester (\Cref{def:tester}) for $\eps_1$-DP input mechanisms with global accuracy~$\beta \leq 1/2$ over $T$ steps.
Define $\delta' := \delta/(e^{\eps} - 1)$, suppose $\beta + \delta' < 1/2$, and let $\bar{p}$ denote the background threshold of the tester.
Define
\[
    \Lambda_0 \;:=\; e^{-2\eps_1} \cdot \left(\frac{T}{2(\beta + T\delta')} \cdot \frac{1}{2(\beta + \delta')}\right)^{\!\eps_1/\eps}.
\]
Then:
\begin{enumerate}
    \item If $\pstar, \bar{p} \leq 1/2$, then $\pstar/\bar{p} \geq \Lambda_0$.
    \item If $\pstar, \bar{p} \geq 1/2$, then $(1-\bar{p})/(1-\pstar) \geq \Lambda_0$.
    \item If $\bar{p} < 1/2 < \pstar$, then $1/(4\bar{p}(1-\pstar)) \geq \Lambda_0$.
\end{enumerate}
These bounds imply a simplified unified description in terms of the \emph{odds ratios} of $\pstar$ and $\pbar$ unconditionally.
\[ \frac{\pstar/(1-\pstar)}{\bar{p}/(1-\bar{p})} \;\geq\; \Lambda_0. \]
\end{cor}
\begin{proof}
Fix sensitivity-$1$ queries $f_1, \ldots, f_T$ with disjoint support, a dataset $X$, and a parameter $\tau_L \in \mathbb{R}$.
For each $t$, define $\calM_t$ by the mechanism that when given input $X$, evaluates $f_t(X)$, draws $\nu\sim \Lap(1/\eps_1)$, and outputs $+1$ if $f_t(X) + \nu \geq \tau_L$, and $-1$ otherwise.
Each $\calM_t$ is $\eps_1$-DP, and disjoint support ensures that changing the entries affecting $f_{t^*}$ does not alter $\calM_t$ for $t \neq t^*$. 
The success probability $p_t = \Pr[\calM_t(X) = +1]$ satisfies
\[
    p_t \;=\;
    \begin{cases}
        \tfrac{1}{2}\,e^{-\eps_1(\tau_L - f_t(X))} & \text{if } f_t(X) \leq \tau_L, \\[4pt]
        1 - \tfrac{1}{2}\,e^{-\eps_1(f_t(X) - \tau_L)} & \text{if } f_t(X) > \tau_L.
    \end{cases}
\]
Note that $p_t \le 1/2$ iff $f_t(X) \le \tau_L$.
For $p \in (0,1)$, define the quantile function $Q(p) := \tau_L + q(p)$ where
\[
    q(p) \;:=\;
    \begin{cases}
        \frac{1}{\eps_1}\log(2p) & \text{if } p \leq 1/2, \\[4pt]
        \frac{1}{\eps_1}\log\frac{1}{2(1-p)} & \text{if } p > 1/2.
    \end{cases}
\]
It follows that $p_t = p$ exactly when $f_t(X) = Q(p)$.
Note that $q(p)$ is continuous, strictly increasing, and satisfies $q(1/2) = 0$.

Note that $f_t(X) = Q(\pstar)$ gives $p_t = \pstar$
and $f_t(X) = Q(\bar{p})$ gives $p_t = \bar{p}$.
With $\alpha_+ = 0$ and $\alpha_- = Q(\pstar) - Q(\bar{p}) = q(\pstar) - q(\bar{p})$, feeding the mechanisms $\calM_t$ into the generalized tester $\calA$ yields an $(\eps, \delta)$-DP (standard) threshold test with global accuracy $\beta$ over $T$ steps.
Applying \Cref{thm:AT_global_LB}:
\begin{align}\label{eq:GTM_LB_q}
    q(\pstar) - q(\bar{p})  = \alpha_+ + \alpha_- \;\geq\; \frac{1}{\eps_1}\left(\log\frac{T}{2(\beta + T\delta')} + \log\frac{1}{2(\beta + \delta')}\right) - 2.
\end{align}
Multiplying both sides by $\eps_1$ and exponentiating, the right-hand side becomes $\Lambda_0$.
We now evaluate $\eps_1(q(\pstar) - q(\bar{p}))$ in each case. Note that by the definition of the generalized private testing  the case $\pstar < \bar{p}$ cannot happen.

\medskip
\noindent\emph{Case 1: $\pstar, \bar{p} \leq 1/2$.}
\[
    \eps_1(q(\pstar) - q(\bar{p})) \;=\; \log(2\pstar) - \log(2\bar{p}) \;=\; \log\frac{\pstar}{\bar{p}}.
\]
Therefore $\pstar/\bar{p} \geq \Lambda_0$.

\medskip
\noindent\emph{Case 2: $\pstar, \bar{p} \geq 1/2$.}
\[
    \eps_1(q(\pstar) - q(\bar{p})) \;=\; \log\frac{1}{2(1-\pstar)} - \log\frac{1}{2(1-\bar{p})} \;=\; \log\frac{1-\bar{p}}{1-\pstar}.
\]
Therefore $(1-\bar{p})/(1-\pstar) \geq \Lambda_0$.

\medskip
\noindent\emph{Case 3: $\bar{p} < 1/2 < \pstar$.}
\[
    \eps_1(q(\pstar) - q(\bar{p})) \;=\; \log\frac{1}{2(1-\pstar)} - \log(2\bar{p}) \;=\; \log\frac{1}{4\bar{p}(1-\pstar)}.
\]
Therefore $1/(4\bar{p}(1-\pstar)) \geq \Lambda_0$.

\medskip
\noindent\emph{Unified bound via the odds ratio.}
In all three cases, the per-case bound implies a lower bound on the odds ratio $\frac{\pstar/(1-\pstar)}{\bar{p}/(1-\bar{p})}$. In Case 1, since $\pstar \geq \bar{p}$ and both are at most $1/2$, we have $\frac{1-\bar{p}}{1-\pstar} \geq 1$, so
\[
  \frac{\pstar/(1-\pstar)}{\bar{p}/(1-\bar{p})} \;=\; \frac{\pstar}{\bar{p}} \cdot \frac{1-\bar{p}}{1-\pstar} \;\geq\; \frac{\pstar}{\bar{p}} \;\geq\; \Lambda_0.
\]
In Case 2, since $\pstar \geq \bar{p}$ and both are at least $1/2$, we have $\frac{\pstar}{\bar{p}} \geq 1$, so
\[
  \frac{\pstar/(1-\pstar)}{\bar{p}/(1-\bar{p})} \;=\; \frac{\pstar}{\bar{p}} \cdot \frac{1-\bar{p}}{1-\pstar} \;\geq\; \frac{1-\bar{p}}{1-\pstar} \;\geq\; \Lambda_0.
\]
In Case 3, since $\pstar > 1/2$ and $\bar{p} < 1/2$, we have $\pstar(1-\bar{p}) > 1/4$, so
\[
  \frac{\pstar/(1-\pstar)}{\bar{p}/(1-\bar{p})} \;=\; \frac{\pstar(1-\bar{p})}{\bar{p}(1-\pstar)} \;>\; \frac{1}{4\bar{p}(1-\pstar)} \;\geq\; \Lambda_0.
\]
Therefore, in all cases,
\begin{align}\label{eq:odds_ratio_LB}
  \frac{\pstar/(1-\pstar)}{\bar{p}/(1-\bar{p})} \;\geq\; \Lambda_0.
\end{align}
\end{proof}

\subsection{Sample complexity lower bounds}

We now turn to lower bounds on the number of evaluations that any generalized private tester must perform.

\begin{thm}[Expected sample complexity lower bound]\label{thm:sample_complexity_LB}
Fix $\pstar \in (0, 1/2]$, a constant input privacy parameter $\eps_1 > 0$, and output privacy parameters $(\eps, \delta)$ with $\delta' := \delta/(e^\eps - 1)$. Let $\calA$ be a generalized private tester (\Cref{def:tester}) with per-step Type~I and Type~II error probabilities $\beta_t^{\tI}$ and $\beta_t^{\tII}$ at step $t$, with $\beta_t^{\tI} \leq 1/4$ and $\beta_t^{\tII} + \delta' < 1/2$.

Then there exists an input instance on which the expected number of evaluations of $\calM_t$ at step $t$ satisfies
\[
  \Ex[N_t] \;\geq\; \frac{e^{-\eps_1}}{4\,\pstar\,(2(\beta_t^{\tII} + \delta'))^{\eps_1/\eps}}.
\]
When $\delta = 0$, this simplifies to $\Ex[N_t] \geq \frac{e^{-\eps_1}}{4\,\pstar\,(2\beta_t^{\tII})^{\eps_1/\eps}}$.
\end{thm}

\begin{proof}
We use the Laplace mechanism construction from the proof of \Cref{cor:GTM_LB}. Fix sensitivity-$1$ queries $f_1, \ldots, f_T$ with disjoint support and a threshold $\tau_L$, and define $\calM_t(X) = \sign(f_t(X) + \Lap(1/\eps_1) - \tau_L)$. As in that proof, construct a chain of datasets $X^{(0)}, X^{(1)}, \ldots, X^{(k)}$ differing only in the support of $f_t$, with $X^{(j)}$ and $X^{(j+1)}$ neighboring, $f_t(X^{(j)}) = f_t(X^{(0)}) - j$, and $f_{t'}(X^{(j)}) = f_{t'}(X^{(0)})$ for all $t' \neq t$. Choosing $f_t(X^{(0)}) = \tau_L + \frac{1}{\eps_1}\ln(2\pstar)$ gives $p_t^{(j)} = \pstar\,e^{-j\eps_1} \leq 1/2$ for all $j \geq 0$. Write $h_j := \Pr[a_t = s_t \mid X_t = X^{(j)}]$ and $v_j := (1-h_j) + \delta'$.

\medskip\noindent\textbf{Step 1: propagation of acceptance probability.}
The mapping $X_t \mapsto a_t$ is $(\eps, \delta)$-DP by the definition of a generalized private tester (\Cref{def:tester}), and post-processing. The Type~II accuracy guarantee gives $h_0 \geq 1 - \beta_t^{\tII}$, so $v_0 \leq \beta_t^{\tII} + \delta'$. Set $k = \lfloor \frac{1}{\eps}\ln\frac{1}{2(\beta_t^{\tII}+\delta')}\rfloor$, so that $e^{k\eps} \leq \frac{1}{2(\beta_t^{\tII}+\delta')}$. Iterating Item~2 of \Cref{lem:acceptance_propagation} gives $v_k \leq e^{k\eps}\,v_0 \leq e^{k\eps}(\beta_t^{\tII}+\delta') \leq 1/2$. It follows that $h_k = 1 + \delta' - v_k \geq 1 + \delta' - 1/2 \geq 1/2$.

\medskip\noindent\textbf{Step 2: bounding the joint probability of halting with no positive evaluations.}
Let $K_t := \sum_{i=1}^{N_t} \mathbf{1}[y_{t,i} = +1]$ denote the number of $+1$ outcomes among the $N_t$ evaluations. Let $R$ denote all internal randomness of the algorithm at step $t$. Conditioned on $R = r$ and the event that all evaluation outcomes are $-1$, the algorithm follows a deterministic execution path: it performs some number of evaluations $n_\bot(r)$ and produces an output $a_\bot(r) \in \{-s_t, s_t\}$, both determined by $r$ alone. We can write
\begin{align*}
  \Pr[a_t = s_t,\, K_t = 0 \mid p_t,\, R = r] &\;=\; \Pr[ K_t = 0 \mid p_t,\, R = r] \cdot \Pr[a_t = s_t\mid p_t,\,R = r,\, K_t = 0]
\end{align*}
Since $\Pr[ K_t = 0 \mid p_t,\, R = r] = (1 - p_t)^{n_\bot(r)}$ and $\Pr[a_t = s_t\mid p_t,\,R = r,\, K_t = 0] = \mathbf{1}[a_\bot(r) = s_t]$, by the law of total probability over $R$,
\begin{align*}
  \Pr[a_t = s_t,\, K_t = 0 \mid p_t] &\;=\; \Ex_R\!\left[(1 - p_t)^{n_\bot(R)}\,\mathbf{1}[a_\bot(R) = s_t]\right].
\end{align*}
Since $(1-p_t)^{n_\bot(R)}$ is non-increasing in $p_t$ for each $R$, the joint probability is non-increasing in $p_t$. In the Laplace construction, $p'$ can be made arbitrarily small (and in particular smaller than $\bar{p}_t$) by taking $f_t(X)$ sufficiently small. For any such $p' \leq \bar{p}_t$, the Type~I accuracy guarantee gives $\Pr[a_t = s_t \mid p'] \leq \beta_t^{\tI}$. Therefore:
\begin{align}\label{eq:joint_bound}
  \Pr[a_t = s_t,\, K_t = 0 \mid p_k] \;\leq\; \Pr[a_t = s_t,\, K_t = 0 \mid p'] \;\leq\; \Pr[a_t = s_t \mid p'] \;\leq\; \beta_t^{\tI}.
\end{align}

\medskip\noindent\textbf{Step 3: lower bound on $\Ex[K_t \mid p_k]$.}
Decomposing the acceptance probability at $p_t = p_k$:
\begin{align*}
  \frac{1}{2} \;\leq\; h_k
  &\;=\; \Pr[a_t = s_t,\, K_t = 0 \mid p_k] + \Pr[a_t = s_t,\, K_t \geq 1 \mid p_k] \\
  &\;\leq\; \beta_t^{\tI} + \Pr[K_t \geq 1 \mid p_k].
\end{align*}
Therefore $\Pr[K_t \geq 1 \mid p_k] \geq 1/2 - \beta_t^{\tI} \geq 1/4$, and since $K_t$ is a non-negative integer-valued random variable:
\[
  \Ex[K_t \mid p_k] \;\geq\; \Pr[K_t \geq 1 \mid p_k] \;\geq\; \frac{1}{4}.
\]

\medskip\noindent\textbf{Step 4: from $\Ex[K_t]$ to $\Ex[N_t]$.}
Conditioned on the algorithm's internal randomness $R=r$, the evaluations of $\calM_t(X_t)$ are i.i.d.\ $\Ber(p_k)$, and $N_t$ is a stopping time with respect to the filtration generated by the evaluation outcomes. By Wald's equation (see, e.g., \cite[Theorem~12.3]{DBLP:books/daglib/0012859}),
\[
  \Ex[K_t \mid p_k,\,R = r] \;=\; p_k \cdot \Ex[N_t \mid p_k,\,R = r].
\]
Taking expectations over $R$:
\[
  \Ex[K_t \mid p_k] \;=\; p_k \cdot \Ex[N_t \mid p_k].
\]
Combining with Step 3:
\[
  \Ex[N_t \mid p_k] \;=\; \frac{\Ex[K_t \mid p_k]}{p_k} \;\geq\; \frac{1}{4\,p_k}.
\]
Since $k \geq \frac{1}{\eps}\ln\frac{1}{2(\beta_t^{\tII}+\delta')} - 1$:
\[
  p_k \;=\; \pstar\,e^{-k\eps_1} \;\leq\; \pstar\,e^{\eps_1}\,(2(\beta_t^{\tII}+\delta'))^{\eps_1/\eps},
\]
and the result follows.
\end{proof}

\begin{remark}[Comparison with the amortized upper bound]
With $\delta = 0$ and global accuracy $\beta$, the Type~II error satisfies $\beta_t^{\tII} \leq \beta$ at every step, so \Cref{thm:sample_complexity_LB} gives $\Ex[N_t] = \Omega(1/(\pstar\,\beta^{\eps_1/\eps}))$ per step. The amortized sample complexity upper bound (\Cref{lem:sample_complexity}), for $\sigma \geq 2$) gives $N_{\leq T}/T = O(\frac{\sigma}{\sigma-1}\cdot\frac{(1/\beta)^{\eps_1/(\theta\eps)}\,\ln(T/\beta)}{\pstar\,(\gamma-1)^2})$. Both scale as $(1/\beta)^{\Theta(\eps_1/\eps)}/\pstar$, matching up to the $\ln(T/\beta)/(\gamma-1)^2$ factor and constants in the exponent on $1/\beta$.
\end{remark}
\section{Reduction from Optimization under Continual Observation to the Batch Setting}
\label{sec:batchtocontinual}

In this section we present the meta-algorithm \BtoCO that converts any batch differentially private optimization algorithm into one that operates under continual observation. We recall the problem setting formally in \Cref{sec:co-setup}, some technical preliminaries in \Cref{sec:co-prelims}, and the algorithm \BtoCO in \Cref{sec:co-algorithm}. We prove its privacy and accuracy guarantee in \Cref{sec:co-privacy} and \Cref{sec:co-accuracy}, and then show in \Cref{sec:co-applications} how to construct various CO algorithms for problems of interest using \BtoCO. 

\subsection{Problem Setup}\label{sec:co-setup}

\medskip \noindent \textbf{Maximization problems.}
Let $\calX$ denote the data universe. An (insertion-only) \emph{data stream} is a sequence of datasets $X_1, X_2, \ldots \in \calX$ such that $X_{t+1}$ is obtained from $X_t$ by inserting a single data element at time $t+1$. Let $\calY$ denote a set of candidate solutions. A maximization problem is specified by a function $f : \calX \times \calY \to \R_{\geq 0}$. We write $\OPT(X) := \max_{Y\in\calY} f(X,Y)$. We consider problems $f$ that have \emph{bounded sensitivity}, i.e. for all neighboring datasets $X, X'$ (differing in one element) and all $Y \in \calY$, 
\[ |f(X,Y) - f(X',Y)| \leq 1.\]

\medskip \noindent \textbf{Private maximization under continual observation.} We study optimization in the continual observation setting for bounded sensitivity maximization problems $f$ that in addition satisfy \emph{data monotonicity}, i.e. for all $Y \in \calY$ and all $t \geq 1$, $f(X_t, Y) \leq f(X_{t+1}, Y)$. We write $\OPT_t := \max_{Y \in \calY} f(X_t, Y)$ for the optimum at time $t$. Data monotonicity implies that $\OPT_t$ is non-decreasing. By bounded sensitivity, $\OPT_{t+1} \leq \OPT_t + 1$.

The goal is to produce, at every time $t$, a solution $Y_t \in \calY$ such that, with probability at least $1 - \beta$ simultaneously for all $t \geq 1$,
\[
f(X_t, Y_t) \geq \Phi \cdot \OPT_t - E_{\mathrm{total}}(t),
\]
where $\Phi \in (0,1]$ is a constant multiplicative approximation factor and $E_{\mathrm{total}}(t) \geq 0$ is additive error incurred at time $t$. The output sequence $(Y_t)$ must satisfy $(\eps, \delta)$-differential privacy, where $\eps$ and $\delta$ are prescribed global budgets. Here two streams $(X_t)_{t\geq1}$ and $(X'_t)_{t\geq1}$ are considered neighboring if for all $t$, $X_t \sim X'_t$. 

A multiplicative approximation is a standard relaxation introduced for dealing with optimization problems $f$ that are NP-hard, and additive error is a consequence of working in the private setting. Prior work \cite{gupta2010differentially, DBLP:conf/icml/ChaturvediNN23} proves lower bounds showing that for problems like submodular maximization additive error is unavoidable. 

\medskip \noindent \textbf{Batch DP algorithm guarantee.}
We will construct algorithms for private maximization under continual observation by starting with a DP algorithm that works in the static i.e. \emph{batch} setting. A batch DP algorithm $A$ takes as input a dataset $X \in \calX$, a privacy parameter $\eps > 0$, an approximate-DP parameter $\delta \geq 0$, and a failure probability $\beta_A \in (0,1)$, and returns a solution $\hat{Y} \in \calY$. The algorithm $A$ is $(\eps, \delta)$-DP. We consider two kinds of accuracy guarantees:
\begin{itemize}
\item \emph{High-probability:} There exists some $\alpha\in (0,1]$ and some function $E_A$ such that with probability at least $1-\beta_A$, $f(X, A(X)) \geq \alpha \cdot \OPT(X) - E_A(X, \eps, \delta, \beta_A)$.
\item \emph{In-expectation:} $\E[f(X, A(X))] \geq \alpha \cdot \OPT(X) - E_A(X, \eps, \delta)$.
\end{itemize}
We write $E_A(t)$ as shorthand for $E_A$ evaluated at the parameters in effect at time $t$.

\subsection{Preliminaries}\label{sec:co-prelims} \medskip \noindent \textbf{1. Converting in-expectation guarantees to bounded error probability guarantees.} We use the following lemma to reduce reasoning about algorithms with in-expectation guarantees, to reasoning about algorithms with constant bounded error probabilities.

\begin{lem}[In-expectation to bounded-error conversion]\label{lem:exp-to-prob}
Let $f : \calX \times \calY \to \R_{\geq 0}$ and let $A$ be a randomized algorithm
satisfying $\E[f(X, A(X))] \geq \alpha \cdot \OPT(X) - E$ for some $\alpha \in (0,1]$
and $E \geq 0$.
Then for any $c \in (0,1)$, writing $\alpha' := (1-c)\alpha$
and $\beta_c := 1 - c\alpha$,
\[
  \Pr\!\bigl[f(X, A(X)) \geq \alpha' \cdot \OPT(X) - E\bigr]
  \;\geq\; 1 - \beta_c.
\]
\end{lem}

\begin{proof}
Set $\tau := \alpha' \cdot \OPT(X) - E$.
If $\tau \leq 0$, the bound holds since $f \geq 0$.
For $\tau>0$, we can write
\begin{align*}
  \alpha \cdot \OPT(X) - E
  &\;\leq\; \E[f(X,A(X))]\\
  &\;\leq\;
  \tau\,\Pr[f(X,A(X)) < \tau] + \OPT(X)\,\Pr[f(X,A(X)) \geq \tau] \\
  &\;=\;
  \tau + (\OPT(X) - \tau)\,\Pr[f(X,A(X)) \geq \tau].
\end{align*}
In the above, we use that $f(X,A(X)) \in [0, \OPT(X)]$ unconditionally. Rearranging:
\begin{align*}
  \Pr[f(X,A(X)) \geq \tau]
  &\;\geq\;
  \frac{\alpha \cdot \OPT(X) - E - \tau}{\OPT(X) - \tau} \\
  &\;=\;
  \frac{\alpha \cdot \OPT(X) - E - \alpha' \cdot \OPT(X) + E}
    {\OPT(X) - \alpha' \cdot \OPT(X) + E} \\
  &\;=\;
  \frac{c\alpha \cdot \OPT(X)}{(1 - \alpha') \cdot \OPT(X) + E}.
\end{align*}
Since $\tau > 0$ implies $E < \alpha' \cdot \OPT(X)$,
\[
  \frac{c\alpha \cdot \OPT(X)}{(1 - \alpha') \cdot \OPT(X) + E}
  \;\geq\;
  c\alpha
  \;=\; 1 - \beta_c.
  \qedhere
\]
\end{proof}

\medskip \noindent \textbf{2. Generalized Private Selection via Random Stopping \citep{DBLP:conf/stoc/0001T19}.} We recall one of the variants of the algorithm introduced by \cite{DBLP:conf/stoc/0001T19} for Generalized Private Selection, that we refer to as \TSelect for convenience. We refer the reader to the arxiv version of their paper \citep{liu2018private} for the statements and proofs that we refer to here.

\TSelect takes as input query access to an $(\eps_1,\delta_1)$-DP mechanism $\calQ$ that for datasets $X\in\calX$ outputs pairs of elements and scores $(Y,v)\in\calY\times\R$. It is also given a threshold $\tau\in\R$, and parameters $\xi\in(0,1]$ and $\eps_0\in[0,1]$. It computes a number of iterations $T$ that is the maximum number of times it will query $A$. It then proceeds to iterate over a counter $j\in[T]$, and draw values $(Y_j,v_j) \gets \calQ(X_t)$. It checks whether $v_j \geq \tau$, i.e. whether the score beats the given threshold, and outputs $(Y_j,v_j)$ if so. If this check fails, then it performs a randomized check by drawing $b\sim \Ber(\xi)$. If $b=+1$, then this check is considered to have failed, and \TSelect halts immediately, outputting $-1$ indicating failure. If $b=-1$, then this check has passed and the algorithm continues to iterate over $j$. If all $T$ iterations result in no successful threshold test, then the algorithm outputs $-1$.

\TSelect can be used as a success probability amplification algorithm. When it is known that the probability of $\calQ$ outputting a value $(Y,v)$ such that $v$ beats $\tau$ with probability at least $p_1$, then by setting $\xi$ to scale roughly as $\beta p_1$, we are guaranteed that $\calQ$ will output a value that beats the threshold $\tau$ with probability $1-\beta$. Further, the privacy cost of this procedure is roughly $(O(\eps_1),O(\delta_1/\xi))$, in other words the privacy cost of success probability amplification is borne almost entirely by the additive privacy loss parameter in the approximate DP setting, and just the constant factor overhead in the multiplicative privacy loss parameter in the pure DP setting. Further, the sample complexity can in fact be bounded in terms of a Geometric distribution with stopping probability at least $p_1$.

\begin{algorithm}[t]
\caption{\TSelect (Algorithm 1, Thresholding with a known threshold $\tau$ in the arXiv version of \cite{liu2018private})}
\label{alg:threshold-select}
\SetAlgoLined
\KwIn{Dataset $X\in\calX$, Mechanism $Q: \calX \to \calY \times \R$;
  threshold $\tau \in \R$;
  $\xi \in (0,1]$, parameter $\eps_0\in[0,1]$.}
\KwOut{A sample $(Y, v) \in \calY \times \R$, or $-1$.}
\BlankLine
$T \gets \lceil\max\{(1/\xi)\ln(2/\eps_0),\; 1+1/(e\xi)\}\rceil$\;
\For{$j = 1, \ldots, T$}{
  $(Y_j, v_j) \gets Q(X)$\;
  \lIf{$v_j \geq \tau$}{\Return $(Y_j, v_j)$}
  Draw $b\sim\Ber(\xi)$, if $b = +1$, \Return $-1$\;
}
\Return $-1$\;
\end{algorithm}

\begin{lem}[Thresholded Private Selection {\cite[Theorem~3.1]{liu2018private}}]\label{lem:threshold-select}
Let $Q : \calX \to \calY\times\R$ be an $(\eps_1,\delta_1)$-DP mechanism.
Fix $\xi \in (0,1]$, $\eps_0 \in (0,1]$, and
$T \geq \max\!\bigl\{\frac{1}{\xi}\ln\frac{2}{\eps_0},\; 1 + \frac{1}{e\xi}\bigr\}$.

\begin{enumerate}
\item \textbf{Privacy.}
The output of \Cref{alg:threshold-select} is
$\bigl(2\eps_1 + \eps_0,\;\; 3e^{2\eps_1 + \eps_0} \cdot \delta_1/\xi\bigr)$-differentially private.
When $\delta_1 = 0$, this simplifies to $(2\eps_1 + \eps_0, 0)$-DP.

\item \textbf{Failure probability.}
Let $p_1 := \Pr_{(Y,v) \sim Q(X)}[v \geq \tau]$. Then
\[
  \Pr[\emph{output} = -1] \;\leq\; \frac{(1-p_1)(1 + \eps_0/2)}{p_1}\,\xi.
\]

\item \textbf{Sample complexity.} The number of draws from $Q$ is at most $T$. Let $p_1 := \Pr_{(Y,v) \sim Q(X)}[v \geq \tau]$. The number of draws is stochastically dominated by a $\Geo(p_1(1-\xi) + \xi)$-distributed random variable.
\end{enumerate}
\end{lem}

\begin{proof}
Items~1 and~2 are parts~(b), (c), and~(e) of \cite[Theorem~3.1]{liu2018private} respectively.
Item~3 follows directly from the proof of part (d) of \cite[Theorem~3.1]{liu2018private}; in each iteration, the algorithm makes an evaluation of $\calQ$ that succeeds with probability at least $p_1$. If this evaluation fails, then the algorithm halts with probability $\xi$. These two Bernoulli draws are independent of each other. Therefore, the probability of halting in any one iteration is exactly $p_1(1-\xi) + \xi$. It follows that the total number of evaluations is dominated by $\Geo(p_1(1-\xi) + \xi)$.
\end{proof}

\medskip \noindent \textbf{3. Simplified GTM.} We briefly recall a simplified version of \Cref{thm:ExPost} that we will refer to in our reduction from optimization in the CO setting to the batch setting.

\begin{cor}[GTM with ExPost parameters]\label{cor:GTM-ExPost}
Fix a threshold probability $\pstar \in (0,1)$, a gap parameter $\gamma \in (1,2]$,
and suppose $\beta_{\GTM} \leq 1/e$ and $C_1, C_H > 0$.
There exists a GTM instance (\Cref{alg:GTM}) such that
if $\eps_t \leq \eps/(C_1 \ln(4/(\omega_t\beta_{\GTM})))$
for all $t \geq 1$, then the following hold
with $\bar{p} := \pstar/(\gamma\Lambda)$,
where
\[
  \Lambda \;:=\; e^{1/C_1 + 2/C_H}
\]
is a universal constant (independent of $t$ and $\beta_{\GTM}$).

\begin{enumerate}
\item \textbf{Privacy.}
The GTM transcript is $(1 + C_H/C_1)\eps$-differentially private.

\item \textbf{Accuracy.}
With probability at least $1 - \beta_{\GTM}$,
simultaneously for all $t \geq 1$:
\begin{enumerate}
\item the algorithm halts at every step $t$
  where $p_t \geq \pstar + \delta_t/\eps_t$, and
\item the algorithm continues past every step $t$
  where $p_t \leq \bar{p} - \delta_t/\eps_t$.
\end{enumerate}
When $\delta_t = 0$, the thresholds simplify to $\pstar$ and $\bar{p}$.

\item \textbf{Sample complexity.}
With probability at least $1 - \beta_{\GTM}$,
the number of evaluations of $\calM_t$ at step $t$ is at most
$O\!\left(\ln(t/\beta_{\GTM})/(\pstar(\gamma-1)^2)\right)$.
\end{enumerate}
\end{cor}

\begin{proof}
Apply \Cref{thm:ExPost} with $\eps_C = \eps$, $\beta = \beta_{\GTM}$, and $s_t = +1$ for all $t$. The hypothesis on $\eps_t$ gives $\Lambda_t = (4/\beta_{\GTM})^{\eps_t/\eps_C}\cdot e^{2/C_H} \leq (4/\beta_{\GTM})^{1/(C_1\ln(4/(\omega_t\beta_{\GTM})))}\cdot e^{2/C_H} \leq e^{1/C_1 + 2/C_H}$.
It follows that $\bar{p}_t = \pstar/(\gamma\Lambda_t) \geq \pstar/(\gamma\Lambda)$ so it follows that $\bar{p} : = \pstar/(\gamma\Lambda)$ suffices for the rejection threshold. 
Items~1 and~3 follow directly from \Cref{thm:ExPost}; the sample complexity uses $\Lambda_t \leq \Lambda = O(1)$.
The $\delta_t/\eps_t$ perturbation in Item~2 follows from \Cref{cor:approxDP}.
\end{proof}

\subsection{Algorithm description}\label{sec:co-algorithm} We can now introduce \BtoCO (\Cref{alg:batch-to-co}) that takes as input black-box access to a mechanism $A$ that solves a maximization problem $f:\calX\times\calY\to\R_{\geq0}$ in the batch DP setting, and constructs a mechanism that solves the same problem in the CO setting. We will show in \Cref{thm:batch-to-co} that \BtoCO fulfills the same sort privacy guarantee as $A$, i.e. pure or approximate. Further, it fulfills an error bound for the dataset $X_t$ that scales with the error incurred by $A$ on $X_t$ when initialized with privacy parameters $\eps_t = \Omega(\eps/\ln^2 (t/\beta))$, $\delta_t = \tilde{\Omega}(\delta/\ln^2 t)$, plus some slow scaling polylogarithmic error terms. Additionally, with probability $1-\beta$, for all $t\geq1$, the per-round sample complexity, i.e. evaluations of $A$, is at most $\tilde{O}(\ln (t/\beta))$.

\medskip
\noindent\textbf{Inputs and Outputs.} A run of \BtoCO proceeds as follows. It is given as input black-box access to a batch DP mechanism $A:\calX\to\calY$ to which it can feed the input dataset $X_t$ at any point $t$ in the input data stream, and privacy parameters $\eps_t$ and $\delta_t$. Every time it evaluates $A$, it is guaranteed that the evaluation is $(\eps_t,\delta_t)$-DP. It will output a sequence of values $(Y_t)_{t\geq 1} \in \calY^*$.

\medskip
\noindent\textbf{Checkpoints.} \Cref{alg:batch-to-co} defines a sequence of increasing time-steps called \emph{checkpoints}, denoted $(t(\ell))_{\ell\geq1} \in\N^*$ over the course of the stream. Checkpoints are time-steps at which \Cref{alg:batch-to-co} will update its output solution; at time-steps which are not checkpoints, the solution output will be the solution output at the previous checkpoint. $t(1) := 1$, i.e. the first time-step is always the first checkpoint, and at the very beginning of the stream, with the dataset $X_1$, we make a call to the batch algorithm $A$ to generate the very first output $Y_1$. All subsequent outputs will be determined as explained below.

\medskip
\noindent\textbf{Calls to \GTM.} After having identified the $(\ell-1)$-th checkpoint, \Cref{alg:batch-to-co} identifies the $\ell$-th checkpoint via an invocation of \GTM (\Cref{alg:GTM}), with the parameter setting indicated by \Cref{cor:GTM-ExPost}. We refer to this instance as $\GTM_\ell$. It defines a \emph{threshold} $h_\ell = (1 + \kappa)^\ell$, and conducts the threshold test:
    \[ 1(f(X_t,\hat{Y}) + \Lap(2/\eps_t) - h_{\ell} \geq 0), \qquad \hat{Y}\gets A(X_t,\eps_t/2,\delta_t). \]
In other words, it checks whether the utility obtained from a value $\hat{Y}$ drawn from the batch algorithm $A$ for the dataset $X_t$ beats the $\ell$-th threshold $h_\ell$. Since $A$ is randomized, this is a probabilistic check, and by adding a draw from the Laplace distribution $\Lap(2/\eps_t)$, we will have the promise that $f(X_t,\hat{Y}) + \Lap(2/\eps_t)$ is $(\eps_t,\delta_t)$-DP. It follows that we can define a single evaluation of this probabilistic test as the $t$-th input $\calM_t(X_t)$ to $\GTM_\ell$. Since neighboring data streams $(X_t)_{t\geq1}$ and $(X'_t)_{t\geq1}$ fulfill the promise that there is at most one time-step $t$ such that the new elements $x_t$ and $x'_t$ inserted at that time-step are unequal, it follows that $X_t$ and $X'_t$ are neighboring at all time-steps.

It follows from the privacy guarantee of $\GTM_\ell$ that this invocation fulfills a pure privacy guarantee. We set the output privacy parameter to be a value $\eps_\ell$ such that $\sum_{\ell=2}^\infty \eps_\ell < \eps$, and the acceptance threshold $\pstar = \Theta(1)$, i.e. some constant to be determined later. It follows from \Cref{cor:GTM-ExPost}, that by setting $\eps_t = \Theta(\eps_{\ell(t)}/\ln (t/\beta))$, we are guaranteed that all invocations $(\GTM_\ell)_{\ell\geq2}$ are $\sum_{\ell\geq2}\eps(\ell)$-DP, and it will suffice to set $\eps(\ell) = \Theta(\eps/\ell \ln^2\ell)$. This holds even for approximate DP $A$ because of the internal purification by \GTM (see \Cref{cor:GTM-ExPost}); the additive privacy loss parameters affect only the accuracy of the threshold tests. 

With high probability, $\GTM_\ell$ will halt if the test passes with probability at least $\pstar + \delta_t/\eps_t$, and it will not halt if $p_t \leq \bar{p} -  \delta_t/\eps_t= \pstar/\gamma\Lambda -  \delta_t/\eps_t$, for some user-defined constants $\gamma$ and $\Lambda$. Here $p_t$ is the probability of the $t$-th threshold test passing, i.e.
    \[ p_t \,:=\, \Pr[f(X_t,A(X_t,\eps_t/2,\delta_t)) + \Lap(2/\eps_t) \geq (1+\kappa)^\ell]. \]
To absorb the $\delta_t/\eps_t$ terms, we introduce a constraint on $\delta_t$ to be at most $\eps_t\cdot \bar{p}/2$.

Modulo the additive error introduced by the Laplace noise, this implies that if $\GTM_\ell$ outputs $a_t = -1$, then $f(X_t,\hat{Y}) + \Lap(2/\eps_t)$ cannot be greater than $(1+\kappa)^\ell$ with probability greater than $\pstar + \delta_t/\eps_t$. Similarly, if $a_t = +1$, then $f(X_t,\hat{Y}) + \Lap(2/\eps_t)$ cannot be smaller than $(1+\kappa)^\ell$ with probability greater than $1 - \bar{p} + \delta_t/\eps_t$. Since the $\delta_t/\eps_t$ terms are negligible, this guarantee promises that if $\GTM_\ell$ indicates continuation, then it is safe to repeat the previous checkpoint's solution whilst ensuring accuracy competitive with what the batch algorithm $A$ would achieve on the dataset $X_t$ with privacy parameter $(\eps_t,\delta_t)$. Further, since $a_t = +1$ only when $p_t \geq \bar{p} - \delta_t/\eps_t$, and since the thresholds $(1+\kappa)^\ell$ increase exponentially with $\ell$ whereas $f(X_t,\cdot)$ can only increase at most linearly with $t$ (due to $1$-sensitivity), it will follow that with high probability, only some $O(\ln (t/\beta))$-many checkpoints are defined by time-step $t$. Since checkpoints are declared at a low rate, the privacy parameters passed to the active instance of $\GTM_\ell$ scale slowly, and the invocations of the batch algorithm also get slowly-decaying privacy-loss parameters $(\eps_t,\delta_t)$ as $t$ increases. Putting everything together, we are guaranteed high accuracy in the threshold tests.

\medskip
\noindent\textbf{Calls to \TSelect.} When a new checkpoint is declared, we know that the previous checkpoint's solution is no longer competitive with $(1+\kappa)^\ell$, and a new solution must be generated. If we make a single invocation to $A$ and set $Y_t$ to equal its output, we are only guaranteed that it beats $(1+\kappa)^\ell$ with probability $\bar{p}-\delta_t/\eps_t$. To amplify the success probability, we make a call to the generalized private selection algorithm \TSelect (\Cref{alg:threshold-select}) to generate a new solution for release, with threshold set equal to $(1+\kappa)^\ell$. We will refer to this invocation of \TSelect as $\TSelect_\ell$. It will follow that as long as $\TSelect_\ell$ makes at least $\approx 1/\bar{p}$ many evaluations, at least one draw will exceed the updated threshold $(1+\kappa)^\ell$, and this solution is selected and output by \Cref{alg:threshold-select}. Again, there is some error introduced by Laplace noise added to privatize the score of each draw, but this again turns out to give us a lower order error term.

\medskip
\noindent\textbf{\BtoCO guarantees.} To sum up, the privacy guarantee of \BtoCO follows by basic composition over the very first call to the batch algorithm $A$, and the subsequent invocations of $(\GTM_\ell)_{\ell\geq2}$ and $(\TSelect_\ell)_{\ell\geq2}$. The accuracy guarantee of \BtoCO follows from the accuracy of the tests conducted by the invocations of \GTM, and the quality of the solutions generated by the invocations of \TSelect. The total sample complexity of \BtoCO turns out to be essentially $O(\ln (t/\beta))$, due to the sample complexity of the ex-post parameter setting with which GTM is used, and the relatively mild stopping probability $\xi = \tilde\Theta(\bar{p})$ with which \TSelect is invoked; since $\bar{p} = \Omega(1)$, this implies a relatively small number of calls until the threshold check passes and a solution may be output.

The outline of the rest of this section is as follows. In \Cref{def:mechanisms} and \Cref{lem:mechanism-properties} we define the mechanisms constructed from $A$ that we pass to the instances of \GTM and \TSelect, and prove accuracy and privacy guarantees for them. \Cref{lem:co-privacy} proves that \BtoCO is $(\eps,\delta)$-DP. \Cref{lem:translation} shows that when $\GTM_\ell$ outputs $a_t = -1$, then the quality of values generated by $A$ cannot be much more than $(1+\kappa)^\ell$, and that if $a_t = +1$, then $\OPT_t$ cannot be much less than $(1+\kappa)^\ell$. \Cref{lem:checkpoint-count} uses the latter guarantee to show that the number of checkpoints scales logarithmically in the length of the stream. \Cref{lem:checkpoint-quality} shows that when a checkpoint is reached, and a value is generated via $\TSelect_\ell$, then its quality is not too much less than $(1+\kappa)^\ell$. Finally, \cref{thm:batch-to-co} puts everything together to formally prove the privacy, accuracy, and sample complexity guarantees of \BtoCO, and \Cref{cor:batch-to-co-hp} and \Cref{cor:batch-to-co-exp} show how to derive concrete accuracy bounds when the batch algorithm $A$ is known to satisfy an accuracy guarantee with bounded error probability $\beta_A$, or just an in-expectation accuracy guarantee, respectively.

We collect the parameters and derived quantities used
in \BtoCO and its analysis.
The pseudocode appears in \Cref{alg:batch-to-co};
the quantities below are referenced throughout
Sections~\ref{sec:co-privacy} and \ref{sec:co-accuracy}.

\medskip
\noindent\textbf{Input parameters.}
\begin{itemize}
\item $\eps < 1$, $\delta \in [0,1)$: global privacy budget.
\item $A$: batch DP algorithm, $(\eps,\delta)$-DP for any given $(\eps,\delta)$.
\item $\pstar \in (0,1/2]$: target halting threshold for the GTM.
  In the high-probability setting, $\pstar = (1-\beta_A)/2$
  where $\beta_A$ is the failure probability of $A$;
  in the in-expectation setting with slack $c$, $\pstar = c\alpha/2$.
\item $\kappa > 0$: parameterizes the scaling of the geometric sequence of thresholds $h_i = (1+\kappa)^i$.
\item $\gamma \in (1,2]$: parameter in calls to \GTM that controls the trade-off between the rejection threshold and per-step sample complexity.
\item $\beta \in (0,1/e)$: global failure probability.
\end{itemize}

\medskip
\noindent\textbf{Weight sequence.}
We fix a sequence $(\omega_i)_{i \geq 1}$ with $\sum_{i=1}^\infty \omega_i \leq 1$.
We use $\omega_k = C_\omega/(k\ln^2(k+1))$, for an appropriate normalization constant $C_\omega$.

\medskip
\noindent\textbf{Constants.}
The following constants are set to satisfy the privacy
and accuracy constraints established in
\Cref{lem:co-privacy} and \Cref{lem:checkpoint-quality}:
\begin{itemize}
\item $C_1 > 0, C_H > 0$: control the ex-post privacy decay rate in \Cref{cor:GTM-ExPost}.The per-step privacy parameter within GTM instance $\ell$ is $\eps_t = \eps(\ell)/(C_1 \ln(4/(\omega_t\beta)))$.
\item $C_\eps = C_1/(C_1 + C_H + 3) < 1$: scaling parameter for the per-checkpoint privacy budget, so that $\eps(\ell) = C_\eps \cdot \omega_\ell \cdot \eps$ (see per-checkpoint quantities below). Chosen so that the combined privacy cost of the $\ell$-th GTM instance and the corresponding \TSelect call is at most $\omega_\ell\,\eps$ (see \Cref{lem:co-privacy}).
\item $C_\beta \leq 1/3$: scaling parameter for the per-instance failure probability, so that $\beta_\ell = C_\beta \cdot \omega_\ell \cdot \beta$ (see per-checkpoint quantities below). The constraint $C_\beta \leq 1/3$ ensures that \TSelect succeeds with sufficiently high probability at each checkpoint (see \Cref{lem:checkpoint-quality}).
\item $C_\delta = 1 + 3e^3/C_\beta$: scaling parameter for the per-instance $\delta$ budget, ensuring that the total additive privacy loss across all checkpoints is at most $\delta$ (see \Cref{lem:co-privacy}).
\end{itemize}

\medskip
\noindent\textbf{Derived quantities from the GTM.}
From \Cref{cor:GTM-ExPost}:
\begin{itemize}
\item $\Lambda = e^{1/C_1 + 2/C_H}$. Depends only on $C_1$ and $C_H$; independent of $t$, $\beta$, and the data.
\item $\bar{p} = \pstar/(\gamma\Lambda)$: the rejection threshold. The GTM continues past step $t$ whenever $p_t \leq \bar{p} - \delta_t/\eps_t$. Since $\pstar \leq 1/2$ and $\gamma\Lambda > 1$, we have $\bar{p} < 1/2$.
\end{itemize}

\medskip
\noindent\textbf{Per-checkpoint quantities.}
For each GTM instance $\ell \geq 2$:
\begin{itemize}
\item $\eps(\ell) = C_\eps \cdot \omega_\ell \cdot \eps$: the privacy budget allocated to instance $\ell$.
\item $\beta_\ell = C_\beta \cdot \omega_\ell \cdot \beta$: the failure probability allocated to instance $\ell$. By definition of $\omega_\ell$, $\sum_{\ell \geq 2} \beta_\ell \leq C_\beta\,\beta$.
\item $\xi_\ell = \bar{p} \cdot \beta_\ell$: the stopping probability parameter for \TSelect at checkpoint $\ell$.
\end{itemize}

\medskip
\noindent\textbf{Per-step quantities.}
For GTM instance $\ell$ at time $t$:
\begin{itemize}
\item $\eps(\ell,t) = \eps(\ell)/(C_1 \cdot \ln(4/(\omega_t\beta)))$: the privacy parameter for evaluations of the test and scored mechanisms at time $t$. The decay in $t$ is required by \Cref{cor:GTM-ExPost} to keep $\Lambda_t \leq \Lambda$.
\item $\delta(\ell,t) = \min(\bar{p}\,\omega_\ell^2\,\beta\,\delta/C_\delta,\; (\bar{p}/2)\cdot\eps(\ell,t))$: the approximate-DP parameter at time $t$. The first term in the min-expression ensures the total additive privacy loss across checkpoints is at most $\delta$; the second ensures $\delta_t/\eps_t \leq \bar{p}/2$, which absorbs the $\delta_t/\eps_t$ perturbation in the GTM accuracy guarantee.
\end{itemize}

\medskip
\noindent\textbf{Batch algorithm parameters at time $t$.}
Suppose there are $\ell-1$ checkpoints before time-step $t$. The batch algorithm $A$ is called with privacy parameters
$\eps_t/2 = \eps(\ell,t)/2$ and $\delta_t = \delta(\ell,t)$.
Substituting, the effective batch privacy parameter is
\[
  \eps_{\mathrm{batch}}(t) = \frac{C_\eps\,\omega_{\ell}\,\eps}{2C_1\,\ln(4/(\omega_t\beta))},
\]
and $\delta_{\mathrm{batch}}(t) = \delta(\ell,t)
\leq \bar{p}\,\omega_\ell^2\,\beta\,\delta/C_\delta
= \Theta(\omega_\ell^2\,\beta\,\delta)$.
We write $E_A(t)$ as shorthand for
$E_A(X_t,\, \eps_{\mathrm{batch}}(t),\, \delta_{\mathrm{batch}}(t))$.

\begin{algorithm}[t]
\caption{Batch-to-Continual-Observation Meta-Algorithm}
\label{alg:batch-to-co}
\SetAlgoLined

\KwIn{%
  Output privacy parameters $\eps<1$, $\delta\in [0,1)$; batch-DP algorithm $A$;
  target threshold $\pstar \in (0,1/2]$;
  parameters $\kappa > 0$, $\gamma \in (1,2]$;
  failure probability $\beta \in (0,1/e)$.}

\KwOut{At each time $t$, a solution $Y_t \in \calY$.}

Define for all $i, t \geq1$:\;
\Indp
    $h_i := (1+\kappa)^{i}$\;
    $\eps(i) := C_\eps \cdot \omega_i \cdot \eps$\;
    $\beta_i := C_\beta \cdot \omega_i \cdot \beta$\;
    $\xi_i := \bar{p} \cdot \beta_i$\;
    $\eps(i,t) := \eps(i) \,/\, \bigl(C_1 \cdot \ln(4/(\omega_t\beta))\bigr)$\;
    $\delta(i,t) := \min(\bar{p} \cdot \omega_i^2 \cdot \beta \cdot \delta/C_\delta,\, (\bar{p}/2)\cdot \eps(i,t))$\;
\Indm
$\eps_1 \gets \eps(1,1)$\;
$\delta_1 \gets \delta(1,1)$\;
$Y_{\mathrm{cur}} \gets A(X_1,\, \eps_1,\, \delta_1)$\;
$Y_1 \gets Y_{\mathrm{cur}}$\;
$\ell \gets 2$\;
$\GTM_\ell.\mathrm{Init}(\eps(\ell),\, \pstar,\, \beta_\ell)$\;

\For{$t = 2, 3, \ldots$}{

  $\eps_t \gets \eps(\ell,t)$\;
  $\delta_t \gets \delta(\ell,t)$\;
  $a_t \gets \GTM_\ell.\mathrm{Step}(\calM_t^{h_\ell},\, X_t,\, \eps_t,\, \delta_t)$\;
  \lIf{$a_t = -1$}{$Y_t \gets Y_{\mathrm{cur}}$}
  \Else{
    $\hat{Y} \gets \TSelect(\calQ_t,\; h_\ell,\; \xi_\ell,\; \min(\eps_t,1))$\;
    \lIf{$\hat{Y} \neq -1$}{$Y_{\mathrm{cur}} \gets \hat{Y}$}
    $Y_t \gets Y_{\mathrm{cur}}$\;
    $\ell \gets \ell + 1$\;
    $\GTM_\ell.\mathrm{Init}(\eps(\ell),\, \pstar,\, \beta_\ell)$\;
  }
}
\end{algorithm}

\begin{defn}[Test and scored mechanisms]\label{def:mechanisms}
Given a batch DP algorithm $A$, privacy parameters $\eps_t > 0$, $\delta_t \in [0,1)$, and a threshold $h > 0$, define:
\begin{enumerate}
    \item The \emph{test mechanism} $\calM_t^h : \calX \to \{-1,+1\}$: on input $X_t$, draw $\hat{Y} \gets A(X_t,\, \eps_t/2,\, \delta_{t})$
and return
\[
  \calM_t^h(X_t) \;:=\; \sign\!\bigl(f(X_t, \hat{Y}) + \Lap(2/\eps_t) - h\bigr).
\]

\item The \emph{scored mechanism}
$\calQ_t : \calX \to \calY \times \R$:
on input $X_t$, draw $\hat{Y} \gets A(X_t,\, \eps_t/2,\, \delta_t)$
and return
\[
  \calQ_t(X_t) \;:=\; \bigl(\hat{Y},\;\; f(X_t, \hat{Y}) + \Lap(2/\eps_t)\bigr).
\]
\end{enumerate}
\end{defn}

\begin{lem}[Properties of the test and scored mechanisms]\label{lem:mechanism-properties}
Let $\calM_t^h$ and $\calQ_t$ be as in \Cref{def:mechanisms}. The following statements hold:

\begin{enumerate}
\item \textbf{Privacy.}
Both $\calM_t^h$ and $\calQ_t$ are $(\eps_t, \delta_t)$-differentially private.

\item \textbf{Test success probability.}
Let $p_t := \Pr[\calM_t^h(X_t) = +1]$.
\begin{enumerate}
\item For any $v \geq h$ and $\beta_A\in (0,1)$ such that $\Pr[f(X_t, \hat{Y}) \geq v] \geq 1 - \beta_A$ for $\hat{Y}\gets A(X_t,\eps_t/2,\delta_t)$,
  \[
    p_t \;\geq\;
    (1-\beta_A)\!\left(1 - \tfrac{1}{2}\exp\!\left(
      -\tfrac{\eps_t}{2}(v - h)\right)\right).
  \]
\item If $h \geq \OPT_t$, then
  \[
    p_t \;\leq\;
    \tfrac{1}{2}\exp\!\left(
      -\tfrac{\eps_t}{2}(h - \OPT_t)\right).
  \]
\end{enumerate}

\item \textbf{Score accuracy.}
Let $(\hat{Y}, v) \sim \calQ_t(X_t)$. Then for any $\beta' \in (0,1)$,
\[
  \Pr\!\left[\bigl|v - f(X_t, \hat{Y})\bigr|
    > \frac{2}{\eps_t}\ln\frac{1}{\beta'}\right]
  \;=\; \beta'.
\]
\end{enumerate}
\end{lem}

\begin{proof}
\textbf{Item 1.}
Each mechanism composes two operations on $X_t$: (i)~$A(X_t, \eps_t/2, \delta_t)$, which is $(\eps_t/2, \delta_t)$-DP; and (ii)~$f(X_t, \hat{Y}) + \Lap(2/\eps_t)$, which for fixed $\hat{Y}$ is $\eps_t/2$-DP since $X_t \mapsto f(X_t, \hat{Y})$ has sensitivity $1$.
By basic composition, both mechanisms are $(\eps_t, \delta_t)$-DP.

\medskip
\noindent\textbf{Item 2.}
Write $V := f(X_t, \hat{Y})$ and $L \sim \Lap(2/\eps_t)$ independent of $\hat{Y}$.
The test returns $+1$ if and only if $V + L \geq h$.

For part~(a): with probability at least $1 - \beta_A$, $V \geq v \geq h$.
Conditioned on this event, fixing a draw of $V$, 
\[ \Pr[V + L \geq h | V] = 1 - \tfrac{1}{2}\exp({-\eps_t(V-h)/2}) \geq 1 - \tfrac{1}{2}\exp({-\eps_t(v-h)/2}).\]
Taking expectation over $V$, and dropping the conditioning on the event of probability $\geq 1-\beta_A$, the stated bound follows.

For part~(b): since $f(X_t, \hat{Y}) \leq \OPT_t$ unconditionally, $V \leq \OPT_t \leq h$. 
For any draw of $V$, 
\[ \Pr[V + L \geq h | V] = \tfrac{1}{2}\exp({-\eps_t(h - V)/2})
\leq \tfrac{1}{2}\exp({-\eps_t(h - \OPT_t)/2}).\]
Taking expectation over $V$ gives the stated bound.

\medskip
\noindent\textbf{Item 3.}
The score is $v = f(X_t, \hat{Y}) + L$
where $L \sim \Lap(2/\eps_t)$ is independent of $\hat{Y}$.
The claim follows from the Laplace tail bound:
$\Pr[|L| > s] = \exp(-\eps_t s/2)$ for all $s \geq 0$.
\end{proof}

\subsection{Privacy}\label{sec:co-privacy}

\begin{lem}\label{lem:co-privacy}
If $\beta<1/e$, $\eps<1$, $C_\eps \leq C_1/(C_1 + C_H + 3)$, and $C_\delta \geq 1 + 3e^3/C_\beta$, then \Cref{alg:batch-to-co} satisfies $(\eps, \delta)$-differential privacy.
When $\delta = 0$, the algorithm is $\eps$-DP.
\end{lem}

\begin{proof}
The observable output is determined by:
(a)~the initialization call to $A$ at time-step $t=1$;
(b)~the GTM transcript within each instance;
(c)~the \TSelect output at each checkpoint; and
(d)~continuation outputs $Y_t = Y_{\mathrm{cur}}$,
which are post-processing of (a) and (c).
By basic composition, the total privacy cost is
the sum of the costs of (a)--(c).

\medskip
\noindent\textbf{Privacy loss per GTM Instance.}
By \Cref{cor:GTM-ExPost}, the GTM transcript for instance $\ell$ is $(1 + C_H/C_1)\eps(\ell)$-DP with $\delta = 0$, since purification is handled internally.
Substituting $\eps(\ell) = C_\eps \omega_\ell \eps$, we get that each invocation is pure DP with privacy parameter
\begin{equation}\label{eq:gtm-eps-cost}
  (1 + C_H/C_1)\,C_\eps\,\omega_\ell\,\eps.
\end{equation}

\medskip
\noindent\textbf{Privacy loss per \TSelect instance.} 
At checkpoint $\ell$, for $t = t(\ell)$, \TSelect (\Cref{alg:threshold-select}) draws from $\calQ_t$, which is $(\eps_t, \delta_t)$-DP by \Cref{lem:mechanism-properties}.
By \Cref{lem:threshold-select}(1) with $\eps_1 = \eps_t$, $\eps_0 = \min(\eps_t, 1)$, $\delta_1 = \delta_t$, and coin-flip probability $\xi_\ell$, the output is
\[
  \bigl(2\eps_t + \eps_0,\;\; 3e^{2\eps_t + \eps_0} \cdot \delta_t / \xi_\ell\bigr)\text{-DP}.
\]
Since $\eps_t \leq 1$, we have $2\eps_t + \eps_0 \leq 3\eps_t$ and $3e^{2\eps_t + \eps_0} \leq 3e^3$.

For the $\eps$-term, we note the checkpoint occurs at the halting time $t = t(\ell)$, where $\eps_{t(\ell)} = \eps(\ell,t) = \eps(\ell)/(C_1 \ln(4/(\omega_{t(\ell)}\beta)))$.
Since $\ln(4/(\omega_{t(\ell)}\beta)) \geq \ln(1/\beta) \geq 1$ (assuming $\beta \leq 1/e$), we get that the $\eps$-term is at most
\begin{equation}\label{eq:checkpoint-eps-cost}
  \frac{3 \eps(\ell)}{C_1} \;=\; \frac{3\,C_\eps\,\omega_\ell\,\eps}{C_1}.
\end{equation}
For the $\delta$-term, we have that by definition $\delta_t \leq \bar{p}\,\omega_\ell^2\,\beta\,\delta/C_\delta$, and $\xi_\ell = \bar{p}\,C_\beta\,\omega_\ell\,\beta$. It follows that the $\delta$-term is at most 
\begin{equation}\label{eq:checkpoint-delta-cost}
  3e^3 \cdot \frac{\delta_t}{\xi_\ell}
  \;=\; 3e^3 \cdot \frac{\bar{p}\,\omega_\ell^2\,\beta\,\delta / C_\delta}
    {\bar{p}\,C_\beta\,\omega_\ell\,\beta}
  \;=\; \frac{3e^3}{C_\beta\,C_\delta}\,\omega_\ell\,\delta.
\end{equation}

\medskip
\noindent\textbf{Privacy loss for first checkpoint.}
The call $A(X_1, \eps_1, \delta_1)$
costs $(\eps_1, \delta_1)$, where $\eps_1 = \eps(1,1)$ is at most
\begin{equation}\label{eq:init-eps-cost}
  C_\eps\,\omega_1\,\eps/(C_1\ln(4/(\omega_1\beta))) \leq C_\eps\,\omega_1\,\eps/C_1.
\end{equation}
Similarly $\delta_1 = \delta(1,1)$ which is at most
\begin{equation}\label{eq:init-delta-cost}
  \bar{p}\,\omega_1^2\,\beta\,\delta/C_\delta
\leq \omega_1\,\delta/C_\delta,
\end{equation}
using $\bar{p}\,\omega_1\,\beta \leq 1$. If $A$ is pure DP and $\delta$ has been set equal to $0$, then this equals $0$.

\medskip
\noindent\textbf{Total multiplicative privacy loss parameter.}
Summing~\eqref{eq:gtm-eps-cost},~\eqref{eq:checkpoint-eps-cost} over $\ell \geq 2$,
and adding~\eqref{eq:init-eps-cost}, the net multiplicative privacy loss parameter is at most
\begin{align*}
  \frac{C_\eps\,\omega_1\,\eps}{C_1}
  + \sum_{\ell=2}^{\infty}
    C_\eps\,\omega_\ell\,\eps\!\left(1 + \frac{C_H}{C_1} + \frac{3}{C_1}\right) \;\leq\; C_\eps\,\eps\!\left(1 + \frac{C_H}{C_1} +  \frac{3}{C_1}\right).
\end{align*}
This is at most $\eps$ when $C_\eps \leq C_1/(C_1 + C_H + 3)$.

\medskip
\noindent\textbf{Total additive privacy loss parameter.}
GTM instances are pure DP, so by summing~\eqref{eq:checkpoint-delta-cost} over $\ell \geq 2$
and adding~\eqref{eq:init-delta-cost}, we have that the additive privacy loss parameter is at most
\begin{align*}
  \frac{\omega_1\,\delta}{C_\delta}
  + \sum_{\ell=2}^{\infty} \frac{3e^3}{C_\beta\,C_\delta}\,\omega_\ell\,\delta \;\leq\;
  \frac{\delta}{C_\delta}\!\left(1 + \frac{3e^3}{C_\beta}\right).
\end{align*}
This is at most $\delta$ when $C_\delta \geq 1 + 3e^3/C_\beta$, and this is true even when $\delta = 0$, as mentioned above.
\end{proof}

\subsection{Accuracy}\label{sec:co-accuracy}

We establish the accuracy of \Cref{alg:batch-to-co}.
The following lemma translates the outcome of the GTM at each step into bounds on $\OPT_t$ relative to the threshold $h$.

\begin{lem}[Threshold test accuracy]\label{lem:translation}
Let $\pstar \in (0,1/2]$
and let $\bar{p} := \pstar/(\gamma\Lambda)$
as in \Cref{cor:GTM-ExPost}.
With probability at least $1 - \beta_\ell$ over the randomness
of the $\ell$-th invocation of \Cref{alg:GTM}:

\begin{enumerate}
\item \textbf{If $a_t = -1$ (continuation),} then
for any $v$ satisfying
$\Pr[f(X_t, \hat{Y}) \geq v] \geq 1 - \beta_A$ for some $\beta_A \leq 1 - 2\pstar$,
\[
  v \;<\; h + \frac{2}{\eps_t}\ln\frac{1}{\bar{p}}.
\]
\item \textbf{If $a_t = +1$ (halting),} then
\[
  \OPT_t + \frac{2}{\eps_t}\ln\frac{1}{\bar{p}} \;>\; h.
\]
\end{enumerate}
\end{lem}

\begin{proof}
By \Cref{cor:GTM-ExPost}, with $\beta_{\GTM} = \beta_\ell$,
with probability at least $1 - \beta_\ell$
the halting and continuation guarantees hold simultaneously
for all $t$.
We condition on this event.
By the definition of $\delta(i,t)$ in \Cref{alg:batch-to-co},
$\delta_t/\eps_t \leq \bar{p}/2$.

\medskip
\noindent\textbf{Item 1.}
If $a_t = -1$, then by \Cref{cor:GTM-ExPost}(2a)
and $\delta_t/\eps_t \leq \bar{p}/2$,
we have $p_t < \pstar + \bar{p}/2$.
On the other hand, if $v \geq h$,
\Cref{lem:mechanism-properties}(2a) gives
\[
  p_t \;\geq\;
  (1-\beta_A)\!\left(1 - \tfrac{1}{2}
    \exp\!\left(-\tfrac{\eps_t}{2}
      (v - h)\right)\right).
\]
For both bounds to hold simultaneously,
since $(1-\beta_A) \geq 2\pstar$ and $1 - \frac{1}{2}e^{-x}$ is increasing in $x$,
setting $x' := v - h$ it must be the case that
\begin{align*}
    2\pstar\left(1 - \frac{1}{2}e^{-\eps_tx' /2}\right) &< \pstar + \bar{p}/2 \\
    \Leftrightarrow 1 - \frac{1}{2}e^{-\eps_tx' /2} &< \frac{1}{2} + \frac{\bar{p}}{4\pstar} \\
    \Leftrightarrow e^{-\eps_tx'/2} &> 1 - \frac{\bar{p}}{2\pstar} \quad \geq \bar{p},
\end{align*}
where in the above we use that $\bar{p}\leq \pstar \leq 1/2$. Taking logarithms, $x' < (2/\eps_t)\ln(1/\bar{p})$.
If $v < h$, the bound holds trivially.

\medskip
\noindent\textbf{Item 2.}
By the contrapositive of \Cref{cor:GTM-ExPost}(2b):
if $a_t = +1$, then $p_t > \bar{p} - \delta_t/\eps_t \geq \bar{p}/2$.
If $h < \OPT_t$, the conclusion holds trivially since $(2/\eps_t)\ln(1/\bar{p}) > 0$. If $h \geq \OPT_t$, then by \Cref{lem:mechanism-properties}(2b),
$p_t \leq \tfrac{1}{2}\exp(-\tfrac{\eps_t}{2}(h - \OPT_t))$.
Combining with $p_t > \bar{p}/2$
and taking logarithms:
\[
  h - \OPT_t \;<\; \frac{2}{\eps_t}\ln\frac{1}{\bar{p}}.
  \qedhere
\]
\end{proof}

\begin{lem}[Checkpoint count]\label{lem:checkpoint-count}
Under the hypotheses of \Cref{lem:translation},
with probability at least $1 - \sum_\ell \beta_\ell \geq 1 - \beta$
(by a union bound over the accuracy guarantees
of \Cref{cor:GTM-ExPost} across all instances),
the number of checkpoints $\ell(t)$ declared by time $t$ satisfies
\[
  \ell(t) \;=\; O\!\left(\frac{1}{\ln(1+\kappa)}\left(
    \ln\frac{t}{\eps}
    + \ln\ln\frac{1}{\beta}
    + \ln\ln\frac{1}{\bar{p}}\right)\right).
\]
\end{lem}

\begin{proof}
Suppose instance $\ell$ halts at step $t(\ell)$. By \Cref{lem:translation}~Item~2,
\[
  h_\ell \;<\; \OPT_{t(\ell)} + \frac{2}{\eps_{t(\ell)}}\ln\frac{1}{\bar{p}}.
\]
By bounded sensitivity, $\OPT_{t(\ell)} \leq t(\ell)$. Substituting $\eps_{t(\ell)} = C_\eps \omega_\ell \eps / (C_1 \ln( 4/(\omega_{t(\ell)}\beta)))$,
\[
  \frac{2}{\eps_{t(\ell)}}\ln\frac{1}{\bar{p}}
  \;=\; \frac{2C_1}{C_\eps\,\omega_\ell\,\eps}\cdot\ln\frac{4}{\omega_{t(\ell)}\beta}\cdot \ln\frac{1}{\bar{p}}.
\]
We can write $\omega_\ell = C_\omega/(\ell\ln^2 (\ell+1) )$,
$\ell\ln^2 (\ell+1) \leq \ell^3$ for $\ell \geq 1$,
and $t(\ell) \leq t$. It follows that
\[
  (1+\kappa)^{\ell}
  \;=\; h_\ell
  \;<\; t + C_3 \cdot \frac{\ell^3 \ln(t/\beta)}{\eps} \cdot \ln\frac{1}{\bar{p}}
\]
for an absolute constant $C_3 > 0$.
Taking logarithms:
\[
  \ell\ln(1+\kappa)
  \;<\; \ln\!\left(t + C_3 \cdot \frac{\ell^3\ln(t/\beta)}{\eps}\cdot\ln\frac{1}{\bar{p}}\right).
\]
Since $3\ln\ell$ is eventually dominated by the linear growth of $\ell\ln(1+\kappa)$,
\[
  \ell \;=\; O\!\left(\frac{1}{\ln(1+\kappa)}\left(
    \ln\frac{t}{\eps}
    + \ln\ln\frac{1}{\beta}
    + \ln\ln\frac{1}{\bar{p}}\right)\right).
  \qedhere
\]
\end{proof}

\begin{lem}[Checkpoint quality]\label{lem:checkpoint-quality}
Under the hypotheses of \Cref{lem:translation}, suppose instance $\ell$ halts at step $t(\ell)$ and \TSelect (\Cref{alg:threshold-select}) is run with the scored mechanism $\calQ_{t(\ell)}$ (\Cref{def:mechanisms}), stopping probability parameter $\xi_\ell$, and threshold $h_\ell$.
Then with probability at least $1 - 3\omega_\ell\beta$,
\[
  f(X_{t(\ell)},\, Y_{\mathrm{cur}})
  \;\geq\;
  h_\ell
    - \frac{4C_1\ln(4/(\omega_{t(\ell)}\beta))}{C_\eps\,\omega_\ell\,\eps}\ln\frac{1}{\omega_\ell\beta}
    - \frac{2C_1\ln(4/(\omega_{t(\ell)}\beta))}{C_\eps\,\omega_\ell\,\eps}\ln\frac{4}{\bar{p}}.
\]
\end{lem}

\begin{proof}
Write $t = t(\ell)$ throughout.
We show that three events each hold
with probability at least $1 - \omega_\ell\beta$.

\medskip
\noindent\textbf{Bound on $K$}
Since instance $\ell$ halted, by \Cref{cor:GTM-ExPost}(2b) and $\delta_t/\eps_t \leq \bar{p}/2$, the success probability of each draw satisfies $p_t > \bar{p} - \delta_t/\eps_t \geq \bar{p}/2$.
By \Cref{lem:threshold-select}(3), the number of draws $K$ is stochastically dominated by $\Geo((\bar{p}/2)(1-\xi_{\ell})+\xi_\ell)$. Since $(\bar{p}/2)(1-\xi_{\ell})+\xi_\ell \geq \bar{p}/2$, it follows that it is also stochastically dominated by $\Geo(\bar{p}/2)$.
Setting $k_{\max} = \lceil (2/\bar{p})\ln(1/(\omega_\ell\beta)) \rceil$:
\[
  \Pr[K > k_{\max}]
  \;\leq\; (1 - \bar{p}/2)^{k_{\max}}
  \;\leq\; e^{-(\bar{p}/2)k_{\max}}
  \;\leq\; \omega_\ell\beta.
\]

\medskip
\noindent\textbf{Success of \TSelect.}
By \Cref{lem:threshold-select}(2)
with coin-flip probability $\xi_\ell = \bar{p}\,C_\beta\,\omega_\ell\,\beta$,
using $1/p_t \leq 2/\bar{p}$, $(1-p_t) \leq 1$,
and $(1 + \eps_0/2) \leq 3/2$:
\[
  \Pr[\text{output} = -1]
  \;\leq\; \frac{3\xi_\ell}{\bar{p}}
  \;=\; 3C_\beta\,\omega_\ell\,\beta
  \;\leq\; \omega_\ell\beta,
\]
where the last inequality holds for $C_\beta \leq 1/3$.
When \TSelect succeeds,
it returns $(\hat{Y}, v)$ with $v \geq h_\ell$.

\medskip
\noindent\textbf{Uniform bound on Laplace noise terms.}
By \Cref{lem:mechanism-properties}(3) applied to each of the first
$k_{\max}$ draws with $\beta' = \omega_\ell\beta/k_{\max}$,
a union bound gives:
\[
  \Pr\!\left[\exists\, j \in [k_{\max}] :
    |v_j - f(X_t, \hat{Y}_j)| > \frac{2}{\eps_t}\ln\frac{k_{\max}}{\omega_\ell\beta}
  \right]
  \;\leq\; \omega_\ell\beta.
\]

\medskip
\noindent\textbf{Combining.}
By a union bound, all three events hold simultaneously
with probability at least $1 - 3\omega_\ell\beta$.
On this event:
$K \leq k_{\max}$ (Step~1),
so the returned draw has index $j \leq k_{\max}$;
its score satisfies $v \geq h_\ell$ (Step~2);
and its Laplace noise satisfies
$|v - f(X_t, \hat{Y})| \leq (2/\eps_t)\ln(k_{\max}/(\omega_\ell\beta))$ (Step~3).
Therefore:
\[
  f(X_t, Y_{\mathrm{cur}})
  \;\geq\; v - |v - f(X_t, \hat{Y})|
  \;\geq\; h_\ell - \frac{2}{\eps_t}\ln\frac{k_{\max}}{\omega_\ell\beta}.
\]
We simplify the additive loss.
Since $k_{\max} \leq 1 + (2/\bar{p})\ln(1/(\omega_\ell\beta))
\leq (4/\bar{p})\ln(1/(\omega_\ell\beta))$
for $\omega_\ell\beta \leq 1/e$:
\[
  \ln\frac{k_{\max}}{\omega_\ell\beta}
  \;\leq\; \ln\frac{4}{\bar{p}}
    + \ln\ln\frac{1}{\omega_\ell\beta}
    + \ln\frac{1}{\omega_\ell\beta}
  \;\leq\; 2\ln\frac{1}{\omega_\ell\beta}
    + \ln\frac{4}{\bar{p}}.
\]
Substituting $\eps_t = C_\eps\omega_\ell\eps/(C_1\ln(4/(\omega_t\beta)))$:
\[
  f(X_t, Y_{\mathrm{cur}})
  \;\geq\; h_\ell
    - \frac{4C_1\ln(4/(\omega_t\beta))}{C_\eps\,\omega_\ell\,\eps}\ln\frac{1}{\omega_\ell\beta}
    - \frac{2C_1\ln(4/(\omega_t\beta))}{C_\eps\,\omega_\ell\,\eps}\ln\frac{4}{\bar{p}}.
  \qedhere
\]
\end{proof}

\begin{thm}[Batch-to-Continual-Observation reduction]\label{thm:batch-to-co}
Let $f$ be a bounded-sensitivity, data-monotone maximization problem,
and let $A$ be a batch $(\eps,\delta)$-DP algorithm.
Fix geometric spacing $\kappa > 0$, gap parameter $\gamma \in (1,2]$,
threshold $\pstar \in (0,1/2]$, failure probability $\beta \in (0,1/e)$, and $\eps < 1$.
Run \Cref{alg:batch-to-co}.
For each time $t$, define
\[
  V_A(X_t) \;:=\;
  \sup\!\left\{v : \Pr\!\left[f(X_t,\, A(X_t,\, \eps_t/2,\,
    \delta_t)) \geq v\right] \geq 2\pstar\right\},
\]
where $\eps_t$ and $\delta_t$ are the privacy parameters in effect at time $t$.
Then:
\begin{enumerate}
\item \textbf{Privacy.}
The output sequence $(Y_t)_{t \geq 1}$ is $(\eps, \delta)$-DP.
When $\delta = 0$, the output is $\eps$-DP.

\item \textbf{Accuracy.}
With probability at least $1 - O(\beta)$,
simultaneously for all $t \geq 1$:
\[
  f(X_t, Y_t) \;\geq\;
  \frac{V_A(X_t)}{1+\kappa}
  - O\!\left(
    \frac{\ln(t/\beta)}{\omega_{\ell(t)}\,\eps}
    \cdot \left(\ln\frac{1}{\omega_{\ell(t)}\,\beta}
      + \ln\frac{1}{\bar{p}}\right)
  \right),
\]
where $\bar{p} = \pstar/(\gamma\Lambda)$
and $\ell(t) = O((\ln(t/\eps)
  + \ln\ln(1/\beta)
  + \ln\ln(1/\bar{p}))/\ln(1+\kappa))$.

\item \textbf{Sample complexity.}
At each step $t$, the GTM makes
$O(\ln(t/\beta)/(\pstar(\gamma-1)^2))$ calls to $A$.
At each checkpoint $\ell$, \TSelect makes
at most $O(\ln(1/(\omega_\ell\beta))/\bar{p})$ calls
with probability at least $1 - \omega_\ell\beta$.
\end{enumerate}
\end{thm}

\begin{proof}
Item~1 is \Cref{lem:co-privacy}.

\medskip
\noindent\textbf{Probability accounting (Item~2).}
We condition on the intersection of the following events:
\begin{itemize}
\item The accuracy guarantee of \Cref{cor:GTM-ExPost}
  holds for every instance.
  By a union bound: probability
  $\geq 1 - \sum_i \beta_i \geq 1 - C_\beta\,\beta$.
\item The checkpoint quality guarantee of \Cref{lem:checkpoint-quality}
  holds at every checkpoint.
  By a union bound: probability
  $\geq 1 - 3\sum_i\omega_i\,\beta \geq 1 - 3\beta$.
\end{itemize}
Total failure $\leq (C_\beta + 3)\beta = O(\beta)$.

\medskip
\noindent\textbf{Between-checkpoint argument.}
Fix $t \geq 2$.
Let $\ell(t)$ denote the index of the most recent checkpoint at or before time $t$,
with $\ell(t) = 1$ if no checkpoint has been declared
(corresponding to the initialization at $t=1$).
The currently active GTM instance is $\ell(t)+1$.

\emph{Case $\ell(t) = 1$ (no checkpoint declared).}
Instance $2$ has not halted at step $t$.
By \Cref{lem:translation}~Item~1
with $v = V_A(X_t)$ and threshold $h_2 = (1+\kappa)^2$:
\[
  V_A(X_t) \;<\; (1+\kappa)^2 + \frac{2}{\eps_t}\ln\frac{1}{\bar{p}}.
\]
Since $f(X_t, Y_t) \geq 0$ by non-negativity:
\[
  f(X_t, Y_t) \;\geq\; \frac{V_A(X_t)}{1+\kappa}
  - (1+\kappa) - \frac{2}{(1+\kappa)\eps_t}\ln\frac{1}{\bar{p}}.
\]
Substituting $\eps_t = C_\eps\omega_2\eps/(C_1\ln(4/(\omega_t\beta)))$,
the additive loss is
$O(\ln(t/\beta)\ln(1/\bar{p})/(\omega_2\eps) + 1)$,
which is dominated by the general bound.

\emph{Case $\ell(t) \geq 2$.}

\emph{Step 1: lower bound on $h_{\ell(t)}$.}
Instance $\ell(t)+1$ has not halted at step $t$.
By \Cref{lem:translation}~Item~1
with $v = V_A(X_t)$
and threshold $h_{\ell(t)+1} = (1+\kappa)h_{\ell(t)}$:
\[
  V_A(X_t) \;<\; (1+\kappa)\,h_{\ell(t)} + \frac{2}{\eps_t}\ln\frac{1}{\bar{p}},
\]
so
\begin{equation}\label{eq:h-lower}
  h_{\ell(t)} \;>\;
  \frac{V_A(X_t)}{1+\kappa}
  - \frac{2}{(1+\kappa)\eps_t}\ln\frac{1}{\bar{p}}.
\end{equation}

\emph{Step 2: checkpoint quality.}
By \Cref{lem:checkpoint-quality} at checkpoint $\ell(t)$:
\[
  f(X_{t(\ell(t))}, Y_{\mathrm{cur}})
  \;\geq\;
  h_{\ell(t)}
  - \frac{4C_1\ln(4/(\omega_{t(\ell(t))}\beta))}{C_\eps\,\omega_{\ell(t)}\,\eps}
    \ln\frac{1}{\omega_{\ell(t)}\beta}
  - \frac{2C_1\ln(4/(\omega_{t(\ell(t))}\beta))}{C_\eps\,\omega_{\ell(t)}\,\eps}
    \ln\frac{4}{\bar{p}}.
\]

\emph{Step 3: monotonicity.}
Data monotonicity gives
$f(X_t, Y_t) \geq f(X_{t(\ell(t))}, Y_{\mathrm{cur}})$.
Since $t(\ell(t)) \leq t$,
$\ln(4/(\omega_{t(\ell(t))}\beta)) = O(\ln(t/\beta))$.

\emph{Combining.}
Substituting~\eqref{eq:h-lower} into Step~2
and applying Step~3.
The continuation correction
$(2/((1+\kappa)\eps_t))\ln(1/\bar{p})$
uses $\eps_t = C_\eps\omega_{\ell(t)+1}\eps/(C_1\ln(4/(\omega_t\beta)))$,
giving $O(\ln(t/\beta)\ln(1/\bar{p})/(\omega_{\ell(t)+1}\eps))$,
which is absorbed since $\omega_{\ell(t)+1} = \Theta(\omega_{\ell(t)})$.
The total additive loss is
\[
  O\!\left(
    \frac{\ln(t/\beta)}{\omega_{\ell(t)}\,\eps}
    \cdot \left(\ln\frac{1}{\omega_{\ell(t)}\,\beta}
      + \ln\frac{1}{\bar{p}}\right)
  \right).
\]

\medskip
\noindent\textbf{Checkpoint count.}
By \Cref{lem:checkpoint-count},
$\ell(t) = O((\ln(t/\eps) + \ln\ln(1/\beta)
  + \ln\ln(1/\bar{p}))/\ln(1+\kappa))$.

\medskip
\noindent\textbf{Sample complexity (Item~3).}
The per-step GTM cost is \Cref{cor:GTM-ExPost}(3):
$O(\ln(t/\beta)/(\pstar(\gamma-1)^2))$ calls per step.
At each checkpoint, by \Cref{lem:checkpoint-quality},
\TSelect makes at most
$k_{\max} = O(\ln(1/(\omega_{\ell}\beta))/\bar{p})$ calls
with probability at least $1 - \omega_{\ell}\beta$.
\end{proof}

\begin{cor}[High-probability batch guarantee]\label{cor:batch-to-co-hp}
Under the hypotheses of \Cref{thm:batch-to-co},
suppose additionally that $A$ satisfies
$f(X, A(X)) \geq \alpha \cdot \OPT(X) - E_A$
with probability at least $1 - \beta_A$ for some $\beta_A < 1$.
Set $\pstar = (1-\beta_A)/2$.
Then with probability at least $1 - O(\beta)$,
simultaneously for all $t \geq 1$:
\[
  f(X_t, Y_t) \;\geq\;
  \frac{\alpha \cdot \OPT_t - E_A(t)}{1+\kappa}
  - O\!\left(
    \frac{\ln(t/\beta)}{\omega_{\ell(t)}\,\eps}
    \cdot \left(\ln\frac{1}{\omega_{\ell(t)}\,\beta}
      + \ln\frac{1}{\bar{p}}\right)
  \right),
\]
where $E_A(t)$ is evaluated at
$\eps_{\mathrm{batch}} = \Theta(\omega_{\ell(t)}\eps/\ln(t/\beta))$
and $\delta_{\mathrm{batch}} = \Theta(\omega_{\ell(t)}^2\beta\delta)$, and $\bar{p} = (1-\beta_A)/(2\gamma\Lambda)$.
\end{cor}

\begin{proof}
Since $\Pr[f(X_t, A(X_t)) \geq \alpha\OPT_t - E_A(t)]
\geq 1 - \beta_A = 2\pstar$,
we have $V_A(X_t) \geq \alpha\OPT_t - E_A(t)$.
The result follows from \Cref{thm:batch-to-co}(2).
\end{proof}

\begin{cor}[In-expectation batch guarantee]\label{cor:batch-to-co-exp}
Under the hypotheses of \Cref{thm:batch-to-co},
suppose additionally that
$\E[f(X, A(X))] \geq \alpha \cdot \OPT(X) - E_A$.
Fix a slack $c \in (0,1)$
and set $\pstar = c\alpha/2$.
Then with probability at least $1 - O(\beta)$,
simultaneously for all $t \geq 1$:
\[
  f(X_t, Y_t) \;\geq\;
  \frac{(1-c)\alpha \cdot \OPT_t - E_A(t)}{1+\kappa}
  - O\!\left(
    \frac{\ln(t/\beta)}{\omega_{\ell(t)}\,\eps}
    \cdot \left(\ln\frac{1}{\omega_{\ell(t)}\,\beta}
      + \ln\frac{1}{\bar{p}}\right)
  \right),
\]
where $\bar{p} = c\alpha/(2\gamma\Lambda)$
and $E_A(t)$ is evaluated at
$\eps_{\mathrm{batch}} = \Theta(\omega_{\ell(t)}\eps/\ln(t/\beta))$
and $\delta_{\mathrm{batch}} = \Theta(\omega_{\ell(t)}^2\beta\delta)$.
\end{cor}

\begin{proof}
By \Cref{lem:exp-to-prob} with slack $c$,
$\Pr[f(X_t, A(X_t)) \geq (1-c)\alpha\OPT_t - E_A(t)]
\geq c\alpha = 2\pstar$.
Thus $V_A(X_t) \geq (1-c)\alpha\OPT_t - E_A(t)$.
The result follows from \Cref{thm:batch-to-co}(2).
\end{proof}   
\section{Applications}
\label{sec:co-applications}

Our batch-to-continual-observation reduction (\Cref{thm:batch-to-co}) has many applications. In this section, we cover some of them in detail.

\subsection{Submodular maximization}

\renewcommand{\arraystretch}{1.4}
\begin{table}[t]
\footnotesize
\caption{Application of \Cref{thm:batch-to-co} to batch DP submodular maximization algorithms. $\OPT$: batch optimum; $\OPT_t$: CO optimum at time $t$; $k$: cardinality or matroid rank constraint; $\kappa,\eta > 0$: user-defined parameters; $\beta$: failure probability. Monotone/non-monotone refers to argument-monotonicity; our framework additionally requires data-monotonicity. For in-expectation batch guarantees, \Cref{cor:batch-to-co-exp} is applied with small constant $c > 0$.}
\label{tab:comparison_2}
\centering
    \begin{tabular}{@{}L{3.2cm} L{4.8cm} L{5.2cm}@{}}
    \toprule
    \textbf{Work, Setting} & \textbf{Batch Guarantee} & \textbf{CO Guarantee (Ours)} \\
    \midrule
    \cite{mitrovic2017differentially}\newline
    Monotone,
        $(\eps,\delta)$-DP, $k$-Card. 
        & $(1 {-} \tfrac{1}{e})\OPT - \tilde{O}\bigl(\tfrac{k^{3/2}}{\eps}\ln m \ln \tfrac{1}{\delta} \bigr)$ 
        & $\tfrac{1 - 1/e}{1+\kappa}\,\OPT_t - \tilde{O}\bigl(\tfrac{k^{3/2}}{\kappa \eps} \ln^3 t \ln m \ln \tfrac{1}{\delta} \bigr)$ \\ 
    \addlinespace
    \cite{DBLP:conf/icml/RafieyY20}\newline
    Monotone,
        $(\eps,\delta)$-DP, $k$-Card.
        & $(1 {-} \tfrac{1}{e})\OPT - \tilde{O}\bigl(\tfrac{k^{7}}{\eps^3} m \ln m\bigr)$ 
        & $\tfrac{1 - 1/e}{1+\kappa}\,\OPT_t - \tilde{O}\bigl(\tfrac{k^{7}}{\kappa \eps^3} m \ln m \ln^6 t\bigr)$ \\
    \addlinespace
    \cite{DBLP:conf/aaai/ChaturvediNZ21}\newline
    Decomp.\ non-mon.\ $(\eps,\delta)$-DP, $k$-Card.
        & $(\tfrac{1}{e} {-} \eta) \OPT - \tilde{O}\bigl(\tfrac{k}{\eta\eps}\ln m \ln\tfrac{1}{\delta}\bigr)$
        & $\tfrac{1/e - \eta}{1+\kappa}\,\OPT_t - \tilde{O}\bigl(\tfrac{k}{\eta\kappa \eps}\ln^3 \tfrac{t}{\beta}\ln m \ln\tfrac{1}{\delta}\bigr)$
        \\
    Decomp.\ non-mon.\ $(\eps,\delta)$-DP, $k$-Matroid
        & $(\tfrac{1}{e} {-} \eta) \OPT - \tilde{O}\bigl(\tfrac{k}{\eta\eps}\ln m \ln\tfrac{1}{\delta}\bigr)$
        & $\tfrac{1/e - \eta}{1+\kappa}\,\OPT_t - \tilde{O}\bigl(\tfrac{k}{\eta\kappa \eps}\ln^3 \tfrac{t}{\beta}\ln m \ln\tfrac{1}{\delta}\bigr)$
        \\
    \addlinespace
    \cite{ghazi2024individualized}\newline
    Decomp.\ monotone $\eps$-DP, $k$-Card., $\eta {\in} (0,1)$  
        &  $(1{-}\tfrac{1}{e}{-}\eta)\OPT - O\bigl(\tfrac{k\ln m/\beta}{\eps}\bigr)$ 
        & $\tfrac{1-1/e-\eta}{1+\kappa}\,\OPT_t - \tilde{O}\bigl(\tfrac{k}{\kappa \eps} \ln^3 \tfrac{t}{\beta} \ln \tfrac{m}{\beta}\bigr)$\\
    Decomp.\ monotone $\eps$-DP, $k$-Matroid, $\eta {\in} (0,1)$  
        &  $(1{-}\tfrac{1}{e}{-}\eta)\OPT - O\bigl(\tfrac{k\ln m/\eta \beta}{\eps}\bigr)$ 
        & $\tfrac{1-1/e-\eta}{1+\kappa}\,\OPT_t - \tilde{O}\bigl(\tfrac{k}{\kappa \eps} \ln^3 \tfrac{t}{\beta} \ln \tfrac{m}{\eta \beta}\bigr)$ \\
    \bottomrule
    \end{tabular}
\end{table}

For a natural number $m\in \mathbb N$, we let $\universe=\{1,2,\cdots, m\}$ be the universe. A set function $f: 2^\mathcal U \mapsto \real$ is \textit{submodular} if
\[
\forall S \subseteq T \subseteq \universe \quad \text{and} \quad v \in \universe\backslash T, \quad f(S\cup\{v\})-f(S)\geq f(T \cup \{v\})-f(T).
\]

In the problem of \textit{submodular maximization}, we are given query access to a submodular set function $f: 2^\mathcal U \mapsto \real$. The goal is to find a subset $T\subseteq \universe$ that maximizes $f(T)$ under certain constraints, such as:
\begin{itemize}
    \item Cardinality constraints. The goal is to find $T$ that maximizes $f(T)$ under the constraint that $|T|\leq c$. 
    \item Matroid constraints. Given a rank-$c$ matroid $(\mathcal{U}, \mathcal I)$, we want $T$ to be an independent set of the matroid (i.e., $T \in \mathcal I$)
    \item Knapsack constraints. There is a cost $c_u\in \mathbb{R}_{>0}$ associated with each element $u\in\mathcal{U}$, and one is only allowed to pick sets $S$ such that $\sum_{u\in S} c_u \leq 1$.
    \item $p$-extendible system constraints. If $\calI$ is a collection of subsets of $\calU$, the pair $(\mathcal{U},\mathcal{I})$ is a $p$-extendible system if for all $S\subset T\in\mathcal{I}$, and $u\in\mathcal{U}$ such that $S\cup\{u\} \in\mathcal{I}$, there is some set $V\subset T\backslash S$ of at most $p$ elements such that $(T\backslash V)\cup \{u\}\in\mathcal{I}$.
\end{itemize}

Submodular maximization is also studied under various restrictions on the objective function itself, such as:
\begin{itemize}
    \item Monotone, or argument-monotone. A submodular function $f$ is monotone if whenever $S \subseteq T$, then $f(S) \leq f(T)$. Note that this is distinct from the data monotonicity used in this paper to construct our black-box reduction, which is why we call this condition \emph{argument monotonicity} to avoid any ambiguity. 
    \item Decomposable. Also known as the combinatorial public projects problem, here the function $f$ can be written as 
    \begin{align*}
        f(\cdot) = \sum_{i=1}^m f_i (\cdot)
    \end{align*}
    for some $f_i : \mathcal{U}\to[0,1]$. Note that here since the co-domain is non-negative, decomposability unconditionally implies data-monotonicity.
    \item Bounded curvature. For a monotone submodular function $f$, its total curvature $\kappa_f$ is given by the expression
    \begin{align*}
        \kappa_f &= 1 - \min_{u\in\mathcal{U}} \frac{f(\mathcal{U}) - f(\mathcal{U}\backslash\{u\})}{f(u) - f(\emptyset)}.
    \end{align*}
    For objective functions with known curvature, we can typically give a better approximation guarantee than in the general case.
\end{itemize}

\paragraph{Private submodular maximization in the batch setting.} Motivated by the problem of \emph{private feature selection}, \cite{mitrovic2017differentially} were the first to initiate a systematic study of the optimization of submodular functions under the constraint of differential privacy, introducing privacy-preserving algorithms for both monotone and non-monotone objectives, under cardinality, matroid, and $p$-extendible systems constraints. Previously the work \cite{gupta2010differentially} had considered a special instance of private submodular maximization under $k$-cardinality constraints for \emph{decomposable} objectives, achieving a $(1-1/e)$-approximation and an additive error which is $O((k/\varepsilon)\log m)$ under approximate DP. This approximation factor is known to be optimal in the non-private setting under standard hardness assumptions and \cite{gupta2010differentially} showed that a $\Omega(k/\varepsilon)$ factor is indeed necessary under the constraints of pure-DP. A similar linear dependence of the additive error on $k$ was shown also to hold in the approximate DP setting by \cite{DBLP:conf/icml/ChaturvediNN23}. All of these works require a bounded sensitivity of the objective function in terms of the input data set. Formally, the private submodular maximization problem is stated as follows:

\begin{defn}[Private submodular maximization]
    We are given a correspondence $X\to f_X(\cdot)$ between datasets $X\in\mathbb{N}^{\mathcal{X}}$ and submodular functions $f_X$, with the property of \emph{bounded sensitivity}; i.e. there is some public value $\Delta f$ such that given the set of admissible inputs $\mathcal{C}\subset 2^{\mathcal{Y}}$, it is the case that
    \[ \max_{X\sim X', Y\in \mathcal{C}} |f_X(Y) - f_{X'}(Y)| \leq \Delta f. \]
    The private submodular maximization problem asks for a solution to the submodular maximization problem defined by the objective $f_X$ for a given data set $X$, and the constraint set $\mathcal{C}$. The mapping from the data set $X$ to the output solution $Y^*$ must be differentially private for the given choice of privacy parameters.
\end{defn}

There has since been a long line of work in the DP algorithms literature studying variants of private submodular maximization. We see that since the bounded sensitivity constraint is inherent to this problem, our result \Cref{thm:batch-to-co} is applicable to each work where the data-monotone condition is also fulfilled. In some cases, like non-negative decomposable objectives in both monotone and non-monotone cases, the data-monotone condition is always fulfilled --- note that even if a submodular function $f_X$ is non-monotone, the data monotone condition requires only that for all $X \subset X'$, $f_X \leq f_{X'}$. In the decomposable setting we have $f_X = \sum_{x\in X}f_{\{x\}}$, and as long as $f_{\{x\}}\geq 0$, data-monotonicity is immediate. In other cases where we only have $\Delta f$-sensitivity, data monotonicity is an extrinsic condition not necessarily true of all objectives. 

To the best of our knowledge, there are no prior works for private submodular maximization in the continual observation setting. With our framework, we are able to apply each work in the batch setting to give the first guarantee in the same setting in the continual observation setting, with the additional potentially extrinsic constraint of the objective being data-monotone.

We construct \Cref{tab:comparison_2} by applying \Cref{thm:batch-to-co} to each prior work. The batch approximation factor $\alpha$ is preserved up to the geometric spacing factor $(1+\kappa)$. For algorithms with in-expectation guarantees, \Cref{cor:batch-to-co-exp} converts the in-expectation guarantee to a high-probability one via \Cref{lem:exp-to-prob}, incurring only a $(1-c)$ factor in the multiplicative guarantee for an arbitrarily small user-chosen $c > 0$, at the cost of an $O(\ln(1/c))$ additive term. In the table, we absorb this loss into the existing slack parameter $\eta$ or suppress it as a lower-order term.
 
\subsection{Densest subgraph}
\label{sec:dsg}

Given a  graph $G = (V,E)$ and a node weight function, $c:V \mapsto \real_+$,  a subset $S \subseteq V$, let $$\rho(S) = \frac{|E(S)|}{\sum_{v\in S} c(v)},$$ where $E(S)$ is the set of edges induced by $S \subseteq V$, i.e., with both end points in $S$. The goal of the densest subgraph problem is to find 
    \[
    S^* := \argmax_{S \subseteq V} \rho(S).
    \]
The above definition is the general form of densest subgraph problem (known as {\em weighted densest subgraph}). The unweighted densest subgraph has all node weights $O(1)$, in which case, $\sum_{v\in S}c(v) = O(|S|)$ and the objective becomes
\[
S^* := \argmax_{S \subseteq V} \frac{|E(S)|}{|S|}.
\]

It is easy to see that in both cases, under edge-insertions the objective is monotonically increasing for every candidate solution $S$, and so \Cref{thm:batch-to-co} applies as is. Further, in the weighted setting, the only works considered assume that $c(v)\geq1$, so we also immediately have that the objective of interest is $1$-sensitive.

For graph problems, differential privacy is usually formulated in terms of node-neighboring and edge-neighboring graphs, which is to say that neighboring input graphs $G = (V,E)$ and $G' = (V',E')$ can vary in the presence or absence of at most $1$ node in the node-neighboring setting, and in the presence or absence of $1$ edge in the edge-neighboring setting. The latter is a weaker notion of privacy but is in practice much easier to work with --- we will consider only this notion of privacy. In the CO setting, we consider two streams to consist of edges of the graphs. Two streams are considered to be neighboring if they differ in one time epoch, where an edge is present or absent. 

\paragraph{Prior work in the batch setting.}

The work~\cite{dhulipala2022differential} solves the unweighted densest subgraph problem for any given constant $\eta>0$ with multiplicative approximation $1/(1+\eta)$ and additive error $O(\tfrac{1}{\varepsilon}\log^4 n)$ in the batch setting in an $\varepsilon$-DP manner. This algorithm succeeds with high probability, which they specify to mean with probability $1 - 1/n^c$ for any given constant $c$.

The work \cite{dinitz2025almost} solves the weighted densest subgraph problem in the approximate-DP setting with no multiplicative loss at all ($\alpha = 1$) and additive error $O\paren{\tfrac{1}{\varepsilon}\sqrt{\paren{\log\paren{\tfrac{n}{\delta}}\log\paren{n}}}}$. This bound holds with high probability, which tracking the proof of the relevant result can be set to be $1 - 1/n^c$ for any constant $c>0$. They assume that $c(v) \geq 1$, a standard assumption from the non-privacy literature. 

\paragraph{Implications for CO}
In the unweighted setting, we apply \Cref{thm:batch-to-co} Item~2 to the high-probability guarantee of~\cite{dhulipala2022differential}. Since the batch algorithm achieves a $1/(1+\eta)$-approximation, our reduction gives a $1/((1+\eta)(1+\kappa))$-approximation under CO. For any desired constant multiplicative loss, we can set $\eta$ and $\kappa$ to be small constants. The additive error is
    \[ \tilde{O}\left(\frac{1}{\kappa\varepsilon}\log^2 t \cdot \log^4 t\right) + \tilde{O}\left(\frac{1}{\kappa\varepsilon}\log^{3} t \right)  = \tilde{O} \left(\frac{1}{\kappa\varepsilon} \log^{6+a} t\right). \]
In the above, the $\tilde{O}$ notation hides $\log 1/\varepsilon$ terms and $a$ is an arbitrarily small constant.

In this unweighted setting, the state of the art result is not the one achieved by our framework, but rather by \cite{epasto2024sublinear}, who give an $\varepsilon$-DP algorithm that achieves a $\frac{1}{1+\kappa}$ multiplicative approximation with additive error $O(\tfrac{1}{\kappa^4 \varepsilon} \log^5 n)$ with high probability for all time-steps. The gap between our $\log^{6+a} t$ and their $\log^5 n$ reflects the cost of the black-box approach.

In the weighted setting, to the best of our knowledge, there is no prior work in the CO setting. Here, on applying \Cref{thm:batch-to-co} Item~2 to the batch DP result of \cite{dinitz2025almost} (which has $\alpha = 1$), we get that there is an $(\varepsilon,\delta)$-DP algorithm that with high probability achieves a $\tfrac{1}{1+\kappa}$-approximation and additive error
    \[ \tilde{O}\left(\frac{1}{\varepsilon\kappa} \log^{2} t \sqrt{\log\tfrac{t}{\delta}\log t} \right) + \tilde{O}\left(\frac{1}{\varepsilon\kappa} \log^{3} t \right) = \tilde{O} \left(\frac{1}{\varepsilon\kappa} \log^{3} t \sqrt{\log\frac{1}{\delta}}\right). \]

\subsection{Densest \texorpdfstring{$k$}{k}-subgraph}

We also note the recent work of Khayatian, Vullikanti and Konar~\cite{khayatian2025differentially} who give the first $(\eps,\delta)$-edge DP algorithms for D$k$S.

Given an unweighted undirected graph $G = (V,E)$ with $|V| = n$ and $|E| = m$, and a target size $k \in [n]$, the Densest-$k$-Subgraph (D$k$S) problem seeks a subset $\calS \subseteq V$ of exactly $k$ vertices maximizing the induced edge density $d(\calS) := |\calE_\calS|/\binom{k}{2}$, where $\calE_\calS = \{(u,v) \in E : u,v \in \calS\}$. We write $d_k^* := \max_{|\calS|=k} d(\calS)$ for the optimum. As in the densest subgraph setting, the induced edge count has sensitivity $1$ under edge-DP and satisfies data monotonicity under edge insertions, so \Cref{thm:batch-to-co} applies.

D$k$S is computationally intractable: achieving an $O(n^{1/(\log\log n)^c})$-approximation is hard under standard complexity-theoretic assumptions~\cite{DBLP:conf/stoc/Manurangsi17}. As a result, no worst-case multiplicative approximation ratio $\alpha$ is available, and \Cref{cor:batch-to-co-hp} and \Cref{cor:batch-to-co-exp} do not apply. However, \Cref{thm:batch-to-co} and its data-dependent guarantee via $V_A(X_t)$ imply that the performance of the CO construction based on their algorithm ought to inherit a performance similar to the performance of their algorithm executed at time step $t$ with privacy parameter $\eps/\ln^2 (t/\beta)$, and some additional additive polylogarithmic overhead scaled by $1/\kappa$.

As the performance of their algorithm is measured entirely empirically, there is no corresponding theoretical result that we can state. We conclude this section by observing that our reduction applies in this problem setting, that our analysis indicates non-trivial performance guarantees, but leave it as an open question for future work to test whether this actually pans out in practice.

\subsection{Max-Cut}
Given a  graph $G=(V,E)$ and an edge weight function $w : E \to \real_+$, the MAX-CUT problem asks for a subset $S\subseteq V$ that maximizes the cut value:
    \[ \max_{S \subseteq V} {\sum_{\substack{e=(u,v) \in E \\ u \in S,\, v \in V\backslash S}}} w(e). \]
It is easy to see that for every candidate solution $S$, the objective is monotonically increasing for insertion-only edge streams. The sensitivity of this problem is $\max_{e\in E} w(e)$, so when a bound on this expression is known the objective function can be rescaled and \Cref{thm:batch-to-co} can be applied.

Chandra et al.~\cite{chandra2024differentially} gave an $\varepsilon$-differentially private algorithm $A$ that  $\left(1,\beta, E_A \right)$-solves MAX-CUT, where 
    \[ E_A(X_t,\varepsilon,0,\beta):= O\left(  \frac{1}{\varepsilon} \paren{n \log n + \log\paren{\frac{1}{\beta}}}   \right). \]
Note that the algorithm of Chandra et al.~\cite{chandra2024differentially} is not efficient: the time and space consumption depend on the complexity of finding the exact max-cut on a noisy graph with non-negative edge weights. However, since our framework applies to this problem, our reduction implies that if there exists an efficient algorithm for DP MAX-CUT in the batch setting, then we would immediately have an efficient algorithm in the CO setting as well, with an accuracy guarantee derived from \Cref{thm:batch-to-co}.
\section{Accuracy guarantees in prior work}
\label{sec:PriorAccuracy}

In this section we recall the formal accuracy guarantees in prior work, and compare them at a technical level with the guarantees we derive for the Generalized Thresholding Mechanism.

\subsection{Prior work}

We first give a brief review of the above threshold mechanism in the standard setting, and the techniques of two works most related to our works, i.e., \cite{DBLP:conf/stoc/0001T19} and \cite{cohen2023generalized}.

\medskip
\noindent\textbf{Review of problem description.} Recall that in \textit{standard private testing} there is a sequence of private real numbered values $(v_t)_{t\geq1}$ and a threshold $\tau\in \R_{+}$. A tester processes the stream in an online fashion and outputs $a_t=-1$ until it receives $v_t \geq \tau$; at this point it outputs $a_t=+1$ and halts. However, the output of the tester should be $\varepsilon$-DP where $v$ and $v'$ are considered neighboring streams if $\forall t,\,|v_t - v'_t|\leq 1$. The \emph{Above Threshold} mechanism \cite{dwork2014algorithmic} solves this problem with the \emph{accuracy guarantee} that with probability $1-\beta$ for all $t\geq1$, $a_t = +1 \Rightarrow v_t \geq \tau - O((1/\eps)(\log (t/\beta)))$, and $a_t = -1 \Rightarrow v_t \leq \tau + O((1/\eps)(\log (t/\beta)))$.

For the ease of the readers, we recall the problem of \emph{generalized private testing}. In the generalized private testing, the values $(v_t)_{t\geq 1}$ are replaced with a sequence of mechanisms, datasets and privacy parameters $\calM_t:\calX \to \{+1,-1\}$, $X_t\in\calX$, and $\eps_t\geq 0$, and the threshold $\tau$ is replaced by a target success probability $\pstar \in (0,1)$. For simplicity, we assume here $\eps_t = \eps_1$ for all $t$. We define $p_t = \Pr[\calM_t(X_t) = +1]$ and $p'_t = \Pr[\calM_t(X'_t) = +1]$, where $X'_t$ is a dataset neighboring $X_t$. Two streams of mechanism and dataset pairs $(\calM_t,X_t)$ and $(\calM_t,X'_t)$ are considered neighboring if $p_t \in (e^{-\eps_t}p'_t,e^{\eps_t}p'_t)$ and $(1-p_t)\in (e^{-\eps_t}(1-p'_t),e^{\eps_t}(1-p'_t)$. We call this the \emph{stability guarantee of $p_t$}. The analyst should process the stream and stop at the first  $(\calM_t,X_t)$ for which $p_t \geq \pstar$. Just as in the standard setting, the goal is to suffer a small loss in accuracy of these threshold tests, which is to say that if the test fails, then $p_t$ should not be more than $\pstar$, and if it passes, then $p_t$ should not be much less than $\pstar$. 

\medskip
\noindent\textbf{Technical overview of \cite{DBLP:conf/stoc/0001T19}.} Liu and Talwar approached this problem by initially assuming access to an oracle which gives the analyst the exact value of $p_t$. Then, they observed that the \emph{log-odds ratio} of the input probability, i.e. $\Phi(p_t) = \log (p_t/(1-p_t))$ is $2\eps$-Lipschitz. In this simplified setting, they can reduce the generalized private testing problem to the standard setting, and they show essentially that with this reduction, with probability $1-\beta$ for $t\geq 1$, if $p_t \leq \bar{p}_t$ then $a_t = -1$, where $\bar{p}_t = (\beta^2/t)^{12\eps_1/\eps} \pstar$; for $p_t\in (\bar{p}_t,\pstar)$, there is no guarantee. We see that the logarithmic additive error in accuracy in the standard setting gets amplified to a multiplicative polynomial gap between $\pstar$ and $\bar{p}$ in the generalized setting.

However, in generalized private testing, the analyst does not have access to $p_t$. To circumvent this issue, the generalized Above Threshold mechanism of \cite{DBLP:conf/stoc/0001T19} draws $N_t$ many samples from $\calM_t(X_t)$, and use a coupling argument to show that the estimates $\widehat{p}_t$ of $p_t$ fulfill a weak Lipschitz property. As discussed in the introduction, the downsides of this approach are that (a) we can only guarantee approximate DP outputs while requiring the inputs to be pure DP, (b) the exponent of $\beta^2/T$ is $12$ and (c) the sample complexity has a large $(T/\beta)^{12\eps_1/\eps}$ factor.

\medskip
\noindent\textbf{Technical overview of \cite{cohen2023generalized}}
 Cohen, Lyu, Nelson, Sarl{\'o}s and Stemmer revisited this problem, and gave a unified framework for generalized private selection and generalized private testing. Their Test procedure for generalized private testing draws a value $p^\dagger$, with cumulative density function, $\Pr[p^\dagger \leq x] \leq x^{\theta_C}$ for all $x\in [0,1]$ and some fixed ${\theta_C} \geq 0$ chosen by the analyst. Then, for every mechanism and dataset $\calM_t$ and $X_t$, they draw a value $b_t \sim \Ber(p^\dagger)$. If $b_t = -1$, then they output $-1$ without evaluating $\calM_t(X_t)$. If $b_1 = +1$, then they output whatever they get from the evaluation of $\calM_t(X_t)$. They showed that, if for all $t\geq1$, $\calM_t$ is $(\eps_1,\delta_t)$-DP, then their mechanism achieves $((2 + {\theta_C})\eps_1,\sum_i \delta_i)$-DP. In particular, for a target output probability of $\eps > 2\eps_1$, this implies that picking ${\theta_C} = \eps/\eps_1 - 2$ suffices for $(\eps,\sum_t \delta_t)$-DP.

\cite{cohen2023generalized} introduce two new key technical insights which make this possible: (1) a coupling argument along the lines of the sparse vector technique that only requires an injective mapping, not a bijective one, and (2) the observation that one can lean on the inherent privacy guarantee of each mechanism evaluation to avoid adding more noise for each test. 

The significance of the first insight is that a one-sided perturbation that only systematically reduces the probability of a test passing suffices for proving the uniform privacy loss bound; this means that less noise is needed for achieving privacy, and this improves the accuracy of tests. Along with the second insight, this implies much less additional is being added to achieve a similar privacy loss guarantee, compared to \cite{DBLP:conf/stoc/0001T19}. At a high level, the ${\theta_C} \eps_1$ term in the privacy loss is incurred by the coupling argument, and the $2\eps_1$ term is incurred by the last mechanism evaluation that leads to a halt.

\cite{cohen2023generalized} do not give accuracy guarantee along the line of \cite{DBLP:conf/stoc/0001T19}. In this section, we analyse the accuracy of their mechanism along the lines of \cite{DBLP:conf/stoc/0001T19}.
We also derive a lower bound demonstrating that with this algorithm, one cannot achieve a rejection threshold greater than $\pstar \beta/T$.
The highlight is that the coefficient of $\eps_1/\eps$ is significantly better, but an additional $\beta/T$ factor is incurred in the rejection threshold.

\subsection{The accuracy guarantee of \texorpdfstring{\cite{DBLP:conf/stoc/0001T19}}{Liu--Talwar (2019)}}

We recall the algorithm and formal guarantee for the reader's reference. We start with main steps in the ExtendedAboveThreshold mechanism presented in ~\cite{liu2018private} (the arxiv version of \cite{DBLP:conf/stoc/0001T19}). The mechanism gets target threshold $\pstar \in (0,1)$, a failure probability $\beta \in (0,1)$, a privacy parameter $\delta \in (0,1)$, an auxiliary parameter $\eps_0 \in (0,1)$, a stream length $T > 1$, and input privacy parameter $\eps_1 > 0$. 
Let $\eps > 0$ denote the output privacy parameter ($\eps_3$ in the notation of \cite{DBLP:conf/stoc/0001T19}).  
The ExtendedAboveThreshold algorithm (\cite[Algorithm~4]{DBLP:conf/stoc/0001T19}) proceeds as follows. 
\begin{enumerate}
    \item At initialization, it draws shared noise $\nu \sim \Lap(4(\eps_1+\eps_0)/\eps)$. 
    \item Let $\sigma_{\mathrm{LT}} := 12(\eps_1 + \eps_0)/\eps$ for the exponent that governs the accuracy--privacy tradeoff. Define
    \begin{align}
    \label{eq:parametersforLT19}
    \begin{split}
        C &:= 2(e^{\eps_0+\eps_1} + 1 + e^{\eps_0/2}) < 21, \qquad
        \Delta := \frac{C\ln(8T/\delta)}{\eps_0(e^{\eps_0+\eps_1} - 1)} = O\!\left(\frac{1}{\eps_0^2}\ln\frac{T}{\delta}\right), \\
        N &:= \frac{e^{\eps_0}\,\Delta}{\min\{\pstar, 1-\pstar\}} \cdot \left(\frac{T+1}{\beta}\right)^{\!\sigma_{\mathrm{LT}}}, \qquad
        \Phi_\Delta(x) := \frac{Nx + \Delta}{N(1-x) + \Delta}.
    \end{split}
    \end{align}
    \item At each step $t = 1, \ldots, T$, 
    \begin{enumerate}
        \item draws $N$ i.i.d.\ evaluations of $\calM_t(X_t)$; 
        \item computes the empirical success probability $\widehat{p}_t := \frac{1}{N}\sum_{j=1}^N \mathbf{1}\{y_{t,j} = +1\}$; 
        \item draws per-step noise $\xi_t \sim \Lap(8(\eps_1+\eps_0)/\eps)$
        \item if
        \[
        e^{\xi_t} \cdot \Phi_\Delta(\widehat{p}_t) \;\geq\; e^{\nu} \cdot \Phi_\Delta(\pstar).
        \]
        then output $a_t = +1$ and halt.  
        Otherwise output $a_t = -1$ and continues.
    \end{enumerate}
\end{enumerate}

We now restate their result in our notation:

\begin{thm}[Theorem 4.4 in \cite{DBLP:conf/stoc/0001T19} restated in our notation]\label{thm:LiuTalwar}
Fix a target threshold $\pstar \in (0,1)$, a failure probability $\beta \in (0,1)$, a privacy parameter $\delta \in (0,1)$, an auxiliary parameter $\eps_0 \in (0,1)$, a stream length $T > 1$, and input privacy parameter $\eps_1 > 0$. Let $\eps > 0$ denote the output privacy parameter ($\eps_3$ in the notation of \cite{DBLP:conf/stoc/0001T19}).  Let $\sigma_{LT}, C, \Delta, N$, and $\Phi_\Delta(x)$ be as defined in \cref{eq:parametersforLT19}.

Suppose that each $\calM_t$ is $\eps_1$-DP. Then the following hold:
\begin{enumerate}
\item \textbf{Privacy.}
The mechanism is $(\eps, \delta)$-DP.

\item \textbf{Not stopping too early.} Conditional on the algorithm halting at step $R$, with probability at least $1 - \beta - \delta/2$:
\[
  p_R \;\geq\; \frac{1}{2}\,e^{-\eps_0}\left(\frac{\beta}{R+1}\right)^{\!\sigma_{\mathrm{LT}}}\pstar.
\]

\item \textbf{Not stopping too late.} Conditional on the algorithm halting at step $R$, with probability at least $1 - \beta - \delta/2$: no earlier step $i < R$ satisfies
\[
  \Phi_\Delta(p_i) \;\geq\; e^{\eps_0}\left(\frac{R+1}{\beta}\right)^{\!\sigma_{\mathrm{LT}}} \Phi_\Delta(\pstar).
\]
Further, (\cite[Theorem~4.2(c), ``moreover'']{DBLP:conf/stoc/0001T19}), with probability $1-\beta - \delta/2$, $p_R$ is bounded from below:
\[
  p_R \;\geq\; 1 - \left(1 + \frac{e^{\eps_0}}{2(1-\pstar)}\left(\frac{R+1}{\beta}\right)^{\!\sigma_{\mathrm{LT}}}\right)^{\!-1}.
\]

\item \textbf{Halting.} If at some step $t$ (conditional on not having halted earlier) the success probability satisfies
\[
  \Phi_\Delta(p_t) \;\geq\; e^{\eps_0}\left(\frac{1}{\beta}\right)^{\!\sigma_{\mathrm{LT}}} \Phi_\Delta(\pstar),
\]
then the algorithm halts at step $t$ with probability at least $1 - \beta - \delta/(4T)$. 
\emph{Typographical note:} The statement of \cite[Theorem~4.2(d)]{DBLP:conf/stoc/0001T19} writes $\Pr[a_t = +1 | \ldots] \leq \beta + \delta/(4T)$, but the proof establishes the complementary bound $\Pr[a_t \neq +1 | \ldots] \leq \beta + \delta/(4T)$, as stated above.

\item \textbf{Sample complexity.}
The number of evaluations of $\calM_t$ at each step is
\[
  N \;=\; O\!\left(\frac{1}{\eps_0^2\,\min\{\pstar,1-\pstar\}} \cdot \left(\frac{T}{\beta}\right)^{\!\sigma_{\mathrm{LT}}} \cdot \ln\frac{T}{\delta}\right).
\]
\end{enumerate}
\end{thm}

Note that the halting guarantee is \emph{not} a guarantee that the algorithm halts whenever $p_t \geq \pstar$; the condition requires $p_t$ to be substantially above $\pstar$. For large $N$, $\Phi_\Delta(x) \to x/(1-x)$, so the condition is approximately
\[
  \frac{p_t}{1-p_t} \;\geq\; e^{\eps_0}\left(\frac{1}{\beta}\right)^{\!\sigma_{\mathrm{LT}}} \cdot \frac{\pstar}{1-\pstar}.
\]

\subsection{Accuracy guarantee for the Test procedure of Cohen et al. \texorpdfstring{\cite{cohen2023generalized}}{Cohen et al. (2023)}} \label{sec:cohen2023generalized}

\begin{algorithm}[t]
\caption{Generalized Private Testing via \cite{cohen2023generalized}}
\label{alg:CohenTesting}
\SetAlgoLined
\SetKwProg{Fn}{Function}{:}{}

\KwIn{Stream of $\eps_1$-DP mechanisms $(\calM_t)_{t\geq1}$ with $p_t := \Pr[\calM_t(X_t) = +1]$; target $\pstar \in (0,1)$; failure probability $\beta \in (0,1)$; parameter ${\theta_C} = \eps/\eps_1 - 2 > 0$.}

\BlankLine

Draw $p^\dagger$ from $[0,1]$ with $\Pr[p^\dagger \leq x] = x^{\theta_C}$\;

\BlankLine

\Fn{$\FTest(\calM_t)$}{
  Draw $r \sim \Ber(p^\dagger)$\;
  \lIf{$r = 0$}{\Return $-1$}
  \Return $\calM_t(X_t)$\;
}

\BlankLine

\For{$t \geq 1$}{
  $N_t \gets \lceil \ln(2/\beta) / ((\beta/2)^{1/{\theta_C}}\,\pstar) \rceil$\;
  $\mathit{halted} \gets \mathit{false}$\;
  \For{$j = 1, \ldots, N_t$}{
    \If{$\FTest(\calM_t) = +1$}{
      \textbf{Output} $a_t = +1$\;
      $\mathit{halted} \gets \mathit{true}$\;
      \textbf{break}\;
    }
  }
  \lIf{$\mathit{halted}$}{\textbf{halt}}
  \lElse{\textbf{output} $a_t = -1$ and continue to step $t+1$}
}
\end{algorithm}

\cite{cohen2023generalized} introduce an elegant and lightweight procedure for achieving an SVT-style privacy loss for generalized private testing. They are able to apply this framework to improve the standard SVT, as well as error bounds for the private multiplicative weight update method~\cite{hardt2010multiplicative} and bounds for adaptive data analysis~\cite{dwork2015generalization}. They introduce the Test subroutine in \Cref{alg:CohenTesting}, and show the following result:

\begin{thm}\emph{\cite[Theorem~2]{cohen2023generalized}.}\label{thm:cohen} Let there be a stream of $(\eps_1,\delta_t)$-DP mechanisms $(\calM_t)_{t\geq1}$, and parameter ${\theta_C} = \eps/\eps_1 - 2 > 0$.
The \FTest procedure of \cite[Algorithm~1]{cohen2023generalized} (reproduced here as the \FTest procedure in \Cref{alg:CohenTesting}) takes as input the given stream and halts when it has output $+1$ $c \geq 1$ times . The collection of outputs satisfies $((2c+{\theta_C})\eps_1, \sum_{i=1}^T \delta_t )$-DP for all $\delta\in(0,2^{-\sqrt{c}})$.
\end{thm}

This theorem provides an adaptive privacy guarantee: it applies to any adaptive sequence of calls to $\FTest$, where the choice of mechanism at each call may depend on prior outputs, as long as the sequence terminates after receiving $c$ positive responses. 

In their work, Cohen et al.~\cite{cohen2023generalized} do not give an accuracy guarantee as introduced by \cite{DBLP:conf/stoc/0001T19} for the generalized private testing problem. We derive an accuracy guarantee for their Test procedure along the lines of the notion of accuracy introduced by \cite{DBLP:conf/stoc/0001T19}.

At initialization, the mechanism draws a shared pass probability $p^\dagger$ from $[0,1]$ with $\Pr[p^\dagger \leq x] = x^{\theta_C}$ for all $x \in [0,1]$. At each step $t$, the mechanism makes up to $N_t$ independent calls to $\FTest(\calM_t)$, where each call draws $r \sim \Ber(p^\dagger)$ and returns $\calM_t(X_t)$ if $r = 1$ and $-1$ if $r = 0$. The mechanism halts at step $t$ if any call returns $+1$, and continues to the next step if all $N_t$ calls return $-1$.

We see that with probability $p^\dagger$, the \FTest procedure does not evaluate the mechanism $\calM_t$ at all. To avoid the possibility of a missed halt, we test every mechanism some $N_t$-many times to ensure that if $p_t \geq \pstar$, then we would encounter a $+1$ with probability $1-\beta$. We formalize this guarantee in \Cref{alg:CohenTesting} and the following result.

\begin{prop}[Accuracy of the mechanism of \cite{cohen2023generalized} for generalized private testing]\label{prop:CohenAccuracy}
Fix a target threshold $\pstar \in (0,1)$, a failure probability $\beta \in (0,1)$, an input privacy parameter $\eps_1 > 0$, and a private dataset $X$. Suppose $\eps > 2\eps_1$, and define ${\theta_C} := \eps/\eps_1 - 2 > 0$. When \Cref{alg:CohenTesting} is run on an input sequence of mechanisms $\calM_t(X)$, the following properties hold.
\begin{enumerate}
\item \textbf{Privacy.}
The mechanism is $\eps$-differentially private.

\item \textbf{Accuracy.}
With probability at least $1 - 2\beta$, simultaneously for all $t \geq 1$:
\begin{enumerate}
    \item \emph{Halting:} If $p_t \geq \pstar$, then $a_t = +1$.
    \item \emph{Continuation:} If $p_t \leq \bar{p}_t$, then $a_t = -1$, where
    \[
        \bar{p}_t \;=\;  \frac{\pstar \beta_t}{2 \ln(2/\beta) } \paren{\frac{\beta}{2} - \beta}^{1/\theta_C}.
    \]
\end{enumerate}

\item \textbf{Sample complexity.}
The number of evaluations of $\calM_t$ at each step is at most $N_t = O\!\left(\frac{\ln(2/\beta)}{\pstar(\beta/2)^{1/{\theta_C}}}\right)$.
\end{enumerate}
\end{prop}

\begin{proof}
Each call to $\FTest(\calM_t)$ returns $+1$ with probability $p^\dagger \cdot p_t$, where $p^\dagger \in [0,1]$ with $\Pr[p^\dagger \leq x] = x^{\theta_C}$: the Bernoulli trial passes with probability $p^\dagger$, and conditioned on passing, $\calM_t(X_t)$ returns $+1$ with probability $p_t$. At step $t$, the mechanism makes calls sequentially and halts globally on the first $+1$ response, or outputs $a_t = -1$ after $N_t$ consecutive $-1$ responses. Since the calls are independent, the probability of observing at least one $+1$ among $N_t$ calls is $1 - (1 - p^\dagger\,p_t)^{N_t}$, regardless of whether we stop at the first $+1$ or run all $N_t$ calls.

\medskip
\noindent\textbf{Privacy.}
\Cref{alg:CohenTesting} produces an adaptive sequence of $\FTest$ calls (the choice of $\calM_t$ at each step may depend on prior outputs $a_1, \ldots, a_{t-1}$) that terminates after at most $c = 1$ positive response. Therefore, by \Cref{thm:cohen}, the full transcript is $\eps$-DP for $\eps=(2+{\theta_C})\eps_1$.

\medskip
\noindent\textbf{Noise control.} Define the event $\calG := \{p^\dagger \in ((\beta/2)^{1/{\theta_C}},(1-\beta/2)^{1/{\theta_C}})\}$. Since $\Pr[p^\dagger \leq x] = x^{\theta_C}$, we have $\Pr[\calG^c] = \beta$.

\medskip
\noindent\textbf{Halting.}
Conditioned on $\calG$, suppose $p_t \geq \pstar$ at some step $t$. Each of the $N_t$ independent calls returns $+1$ with probability $p^\dagger \cdot p_t \geq (\beta/2)^{1/{\theta_C}} \pstar$. The probability that all $N_t$ calls return $-1$ is
\[
  (1 - p^\dagger \cdot p_t)^{N_t} \;\leq\; (1 - (\beta/2)^{1/{\theta_C}}\pstar)^{N_t} \;\leq\; \exp(-N_t (\beta/2)^{1/{\theta_C}}\pstar) \;\leq\; \frac{\beta}{2},
\]
where the last inequality uses $N_t \geq \ln(2/\beta)/((\beta/2)^{1/{\theta_C}}\pstar)$. Since we only need to detect the first instance of $p_t \geq\pstar$, the probability of a Type~II error at any step is at most $\beta/2$, conditioned on $\calG$.

\medskip
\noindent\textbf{Continuation.}
Suppose $p_t \leq \bar{p}_t$ at step $t$. Since $x \mapsto 1 - (1-x)^{N_t}$ is non-decreasing on $[0,1]$ and $p^\dagger \leq (1-\beta/2)^{1/{\theta_C}}$, the probability that at least one of the $N_t$ calls returns $+1$ satisfies,
\[
  1 - (1 - p^\dagger \cdot p_t)^{N_t} \;\leq\; p^\dagger\,N_t\,p_t \;\leq\; (1-\beta/2)^{1/{\theta_C}}\,N_t\,\bar{p}_t,
\]
where the first inequality uses $1 - (1-x)^N \leq Nx$ for $x \in [0,1]$. Summing over all $t\geq1$, we require
\[
    \sum_{t\geq 1} \left(1-{\frac{\beta}{ 2}}\right)^{1/{\theta_C}}\,N_t\,\bar{p}_t \;\leq\; {\frac{\beta}{2}}.
\]

Defining $\bar{p}_t = \tfrac{\beta_t}{ 2N_t (1-\beta/2)^{1/{\theta_C}}}$, since $\sum_{t\geq1}\beta_t/2 = \beta/2$, this bound holds. Setting the value of $N_t$, 
\[
  \bar{p}_t \;=\; \frac{\beta_t }{2 \left\lceil \frac{\ln(2/\beta)}{ (\pstar(\beta/2)^{1/{\theta_C}})} \right\rceil \left(1-\frac{\beta}{2} \right)^{1/{\theta_C}}} \;\leq\; \frac{\pstar \beta_t(\beta/2)^{1/{\theta_C}}}{2 \ln(2/\beta) (1-\beta/2)^{1/{\theta_C}}} . 
\]

\medskip
\noindent\textbf{Failure probability.}
The total failure probability is at most
\[
  \underbrace{\Pr[\calG^c]}_{\leq\,\beta} \;+\; \underbrace{\Pr[\text{Type~II error} \mid \calG]}_{\leq\,\beta/2} \;+\; \underbrace{\Pr[\text{Type~I error}\mid \calG]}_{\leq\,\beta/2} \;\leq\; 2\beta.
\]
This completes the proof of \Cref{prop:CohenAccuracy}.
\end{proof}

\begin{remark}[Approximate DP inputs]\label{rem:CohenApproxDP}
\Cref{prop:CohenAccuracy} assumes pure $\eps_1$-DP input mechanisms. When the input mechanisms are $(\eps_1, \delta)$-DP, \cite[Theorem~2]{cohen2023generalized} gives $((2+{\theta_C})\eps_1, \sum_i \delta_i)$-DP, where the sum is over all calls to $\FTest$. Since \Cref{alg:CohenTesting} makes as many as $\sum_{t=1}^T N_t$ calls over $T$ inputs, this yields a total $\delta$-cost of up to $\sum_{t=1}^T N_t \delta_t$, which scales linearly in both the number of trials per step and the stream length. \cite{cohen2023generalized} also provide a R\'enyi DP analysis~\cite[Theorem~2]{cohen2023generalized} to obtain tighter approximate DP guarantees via conversion, at the cost of a more involved privacy accounting.

Both the linear $\delta$-accumulation and the need for R\'enyi DP analysis can be avoided by applying the purification technique of \Cref{lem:ApproxToPure}. Given $(\eps_1, \delta)$-DP mechanisms, each evaluation is passed through a binary symmetric channel with crossover probability $\phi = \delta/(e^{\eps_1} - 1 + 2\delta) \leq \delta/\eps_1$, producing a pure $\eps_1$-DP mechanism. The acceptance and rejection thresholds are then perturbed by at most $\delta/\eps_1$. This approach applies to the mechanism of Cohen et al. \cite{cohen2023generalized} as well as to our \Cref{alg:GTM}, and yields pure DP guarantees for the output in both cases, with no $\delta$-accumulation across calls.
\end{remark}

\begin{lem}[False positive accumulation under shared randomness]\label{lem:shared-noise-limit}
Let $M \geq 1$, $q_0 \in (0,1]$, and $q > 0$. Let $\Omega$ be a probability space and $E \subseteq \Omega$ an event with $\Pr[E] \geq q_0$. Suppose that for every outcome $\omega \in E$, conditioned on $\omega$, at least $M$ conditionally independent trials are observed, each succeeding with conditional probability at least $q$. Then the probability that at least one trial succeeds is at least
\[
    q_0\left(1 - e^{-Mq}\right).
\]
In particular, if this probability is at most $\beta$ for some $\beta < q_0$, then
\[
    q \;\leq\; \frac{1}{M}\ln\frac{1}{1 - \beta/q_0}.
\]
For $\beta \leq q_0/2$, this gives $q \leq 2\beta/(M q_0)$.
\end{lem}

\begin{proof}
For every $\omega \in E$, the probability that all $M$ trials fail is at most $(1-q)^M \leq e^{-Mq}$. Averaging over $E$, $\Pr[\text{all fail} \mid E] \leq e^{-Mq}$, so $\Pr[\text{at least one succeeds} \mid E] \geq 1 - e^{-Mq}$. Marginalizing, the unconditional probability of at least one success is at least $q_0(1 - e^{-Mq})$. Setting this at most $\beta$ and rearranging,
    \[ 1 - e^{-Mq} \leq \frac{\beta}{q_0} \;\Rightarrow\; Mq \leq \log \frac{1}{1-\beta/q_0}.\]
Dividing both sides by M in the latter inequality yields the stated bound. For $\beta \leq q_0/2$, we have $\beta/q_0 \leq 1/2$, and the standard inequality $\ln(1/(1-x)) \leq 2x$ for $x \in [0, 1/2]$ gives $q \leq 2\beta/(Mq_0)$.
\end{proof}

\begin{cor}[Error lower bound of \cref{alg:CohenTesting}]\label{cor:CohenAccuracyLimit}
Fix $\beta \in (0,1/4)$, $T \geq 1$, $\eps_1 > 0$, and $\eps > 2\eps_1$. Write $\theta_C = \eps/\eps_1 - 2 > 0$ for the parameter of the mechanism of \cite{cohen2023generalized}. Consider their testing mechanism applied to $T$ mechanisms, each $\eps_1$-DP with range $\{-1,+1\}$, success probabilities $p_t$, and evaluated $N_t$ times at step $t$.

If there exists $\bar{p} > 0$ such that for every instance with $p_t \leq \bar{p}$ for all $t \in [T]$, the probability that the algorithm halts is at most $\beta$, then
\[
    \bar{p} \;\leq\; \frac{4\beta \cdot 2^{1/\theta_C}}{\sum_{t=1}^T N_t}.
\]
In particular, for any fixed $\theta_C > 0$, $\bar{p} = O\!\left(\beta / \sum_{t=1}^T N_t\right)$.
\end{cor}

\begin{proof}
The mechanism draws $p^\dagger \in [0,1]$ with $\Pr[p^\dagger \leq x] = x^{\theta_C}$. Define $E = \{p^\dagger \geq (1/2)^{1/\theta_C}\}$, which has probability $q_0 = 1 - ((1/2)^{1/\theta_C})^{\theta_C} = 1/2$. At each step $t$, each of the $N_t$ calls to $\FTest$ independently draws $r \sim \Ber(p^\dagger)$ and evaluates $\calM_t (X)$ if $r = 1$. For any fixed $p^\dagger \in E$, the $\FTest$ calls are conditionally independent of each other, each producing $+1$ with conditional probability $p^\dagger \cdot p_t \geq (1/2)^{1/\theta_C} \cdot \bar{p}$. The algorithm halts on any $+1$. Apply \Cref{lem:shared-noise-limit} with $M = \sum_{t=1}^T N_t$, $q = (1/2)^{1/\theta_C}\,\bar{p}$, and $q_0 = 1/2$. Since $\beta < 1/4 = q_0/2$,
\[
    (1/2)^{1/\theta_C}\,\bar{p} \;\leq\; \frac{4\beta}{M},
\]
giving the claimed bound on $\bar p$.
\end{proof}

\subsection{Technical Overview of \texorpdfstring{\cite{ghaziprivate}}{Ghazi et al. (2025)}}
\label{sec:ghaziprivate}

Ghazi, Kamath, Knop, Kumar, Manurangsi, and Zhang~\cite{ghaziprivate} introduce two mechanisms which are relevant to our setting. They work in the generalized private selection setting and consider a stream of mechanisms generating outputs with scores in some arbitrary domain, and their objective is to find an output which achieves a user-defined threshold. 

Their first mechanism is their generalized AboveThreshold mechanism with random dropping \cite[Algorithm 2]{ghaziprivate}. It is similar to the \FTest procedure of \cite{cohen2023generalized}, with a dropping probability $p^\dagger_G = e^{-\eps_i \cdot k}$, where $k\sim \Geo(e^{-\eps'})$. For every mechanism in the input stream they first draw a value $y\sim \Ber(p^\dagger_G)$, and only query the mechanism if this check passes. Then, if the score of the item output by the mechanism beats the given threshold, they output this item, along with the index of the mechanism that generated this output, in the clear. They show that this mechanism fulfills an $(2\eps_i + \eps')$-ex post privacy guarantee.

Their guarantee is more general than that of \cite{cohen2023generalized} by allowing for varying input privacy parameter $\eps_i$, and they also show that one can output the element that achieves the score which passes the threshold test. However, the core idea is similar - each access to $\calM_t$ is mediated by a dropping probability, and in fact the distribution from which this probability is drawn is very similar (we explain this next). Further, the privacy loss guarantee is similar; for $\eps_t = \eps_1$, their guarantee is exactly the $2\eps_i + \eps$ guarantee achieved by \cite{cohen2023generalized}.

To reason about the dropping probability distribution, we see that for $k\sim \Geo(e^{-\eps'})$ and any $j\in\N$, $\Pr[k\geq j] = e^{-\eps' j}$. It follows that the $1-\beta$-probability upper bound on $k$ is $\approx \tfrac{1}{\eps'} \log 1/\beta$. It follows that the $(1-x)$ probability lower bound on the dropping probability $p^\dagger_G$ equals $\exp(-\eps_i \cdot \tfrac{1}{\eps'} \log 1/x) \approx \beta^{\eps_i /\eps'}$. In other words, $\Pr[p^\dagger < x] \approx x^{\eps_i/\eps'}$. This is essentially the same distribution from which the dropping probability $p^\dagger$ of \cite{cohen2023generalized} is drawn, and it follows that similar upper and lower bounds in terms of accuracy for generalized private testing follow. 

\cite{ghaziprivate} introduce another algorithm called the Hyperparameter Tuning Mechanism with Random Dropping (Algorithm 3 in their paper); we recall it here for reference in our notation for ease of comprehension. The algorithm incurs output privacy loss $\eps = 2\eps_i + \eps'$; we write $\theta_G = \eps - 2\eps_1  = \eps'> 0$ to match notation with our discussion.

\begin{algorithm}[t]
\caption{Maximum Selection with Random Dropping~\cite{ghaziprivate}}
\label{alg:GhaziSelection}
\SetAlgoLined
\KwIn{Mechanisms $\calM_t : \calX \to \{-1,+1\}$, each $\eps_1$-DP, for $t \in [T]$; dataset $X$.}
\KwOut{$(o, t) \in \{-1,+1\} \times [T]$ or $\perp$.}
\BlankLine
$S \gets \{\perp\}$\;
Draw $k \sim \mathrm{Geo}(e^{-\theta_G})$\;
\For{$t = 1, \ldots, T$}{
  Draw $y_t \sim \Ber(e^{-\eps_1 \cdot k})$\tcp*{random drop}
  \If{$y_t = 1$}{
    $o \gets \calM_t(X)$\;
    $S \gets S \cup \{(o, t)\}$\;
  }
}
\Return element of $S$ selected by a data-independent rule\;
\end{algorithm}

They prove that their algorithm achieves the following accuracy guarantee:

\begin{thm}[\cite{ghaziprivate}, restated]\label{thm:GhaziAccuracy}
Let $\alpha, \beta \in (0,1)$, $\eps' > 0$, and let $\calM_1, \ldots, \calM_d$ be mechanisms where $\calM_i$ is $\eps_i$-DP with range $\calO$. Define
\[
  T_i \;:=\; \left\lceil \frac{1}{\alpha} \left(\frac{2}{\beta}\right)^{\!\eps_i/\eps'} \ln\!\left(\frac{2}{\beta}\right) \right\rceil.
\]
Run \Cref{alg:GhaziSelection} with $\calE = \Geo(e^{-\eps'})$, repeating each $\calM_i$ for $T_i$ times in the input sequence. Then:
\begin{enumerate}
\item \textbf{Privacy.} The mechanism is ex-post $\tilde{\eps}$-DP with $\tilde{\eps}(o,i) = 2\eps_i + \eps'$ and $\tilde{\eps}(\perp) = 0$.
\item \textbf{Accuracy.} If there exists $o^* \in \calO$ and $i^* \in [d]$ such that $\Pr[\calM_{i^*}(D) \geq o^*] \geq \alpha$, then the algorithm outputs an element at least as large as $(o^*, 0)$ with probability at least $1 - \beta$.
\item \textbf{Sample complexity.} The number of evaluations of $\calM_i$ is at most $T_i = O\!\left(\frac{1}{\alpha}\left(\frac{2}{\beta}\right)^{\!\eps_i/\eps'} \ln\!\left(\frac{2}{\beta}\right)\right)$.
\end{enumerate}
\end{thm}

This guarantee is very compelling, because unlike our upper and lower bounds for generalized private testing, it does not incur poor scaling with the length of the stream. However, we find that when we analyse their mechanism in terms of \emph{testing}, as opposed to for \emph{selection}, it incurs a similar accuracy loss of $\beta/T$. We formalize this via a toy example. Consider a setting where there are some $T$ mechanisms; $G$ of them are \emph{good}, and have a probability $\pstar$ of generating good $+1$ outputs, and the remaining $T-G$ of them are \emph{bad}, and only generate good $+1$ outputs with probability $\bar{p}_t$. If we want a $1-\beta$ probability guarantee that the mechanism whose index is output by the selection mechanism is not bad, then for any fixed choice of $\theta_G$, we will require $\pbar \approx \Theta(\beta/T)$ if $G = o(T)$.

\begin{cor}[Error lower bound for \Cref{alg:GhaziSelection} when used for testing]\label{cor:GhaziAccuracyLimit}
Fix $\beta \in (0,1/2)$, $T \geq 2$, $\pstar \in (0,1)$, $\eps_1 > 0$, and $\eps > 2\eps_1$. Write $\theta_G = \eps - 2\eps_1$ and suppose $\beta \leq (1-e^{-\theta_G})/2$. Consider \Cref{alg:GhaziSelection} applied to $T$ mechanisms with range $\{-1,+1\}$, of which $G \in [T-2]$ have success probability $p_t \geq \pstar$ and the remaining $T-G$ have success probability $p_t \leq \bar{p}$. Suppose the selection rule is uniform random among all $(+1, t)$ in $S$.

If the probability that the algorithm returns $(+1,t)$ with $p_t \leq \bar{p}$ is at most $\beta$, then
\[
    \bar{p} \;\leq\; \frac{2\beta(1 + G\pstar)}{(T-G)(1 - e^{-\theta_G})}.
\]
In particular, for any fixed $\eps > 2\eps_1$, $\bar{p} = O(\beta(1 + G\pstar)/(T-G))$.
\end{cor}

\begin{proof}
The algorithm draws $k \sim \mathrm{Geo}(e^{-\theta_G})$. Conditioned on $k = 0$, which has probability $1 - e^{-\theta_G}$, every Bernoulli gate passes, so all $T$ mechanisms are evaluated independently. Fix a bad configuration $t_0$ with $p_{t_0} \leq \bar{p}$. Under uniform random selection among all $(+1,t)$ in $S$,
\[
    \Pr[t_0 \text{ selected} \mid k = 0] \;=\; \bar{p} \cdot \Ex\!\left[\frac{1}{1 + N_{-t_0}}\right],
\]
where $N_{-t_0}$ is the number of $+1$'s from the remaining $T-1$ configurations, independent of $t_0$'s output. Since $x \mapsto 1/(1+x)$ is convex, Jensen's inequality gives
\[
    \Ex\!\left[\frac{1}{1 + N_{-t_0}}\right] \;\geq\; \frac{1}{1 + \Ex[N_{-t_0}]} \;=\; \frac{1}{1 + (T-G-1)\bar{p} + G\pstar}.
\]
The events $\{t_0 \text{ selected}\}$ are disjoint over distinct bad configurations $t_0$, so summing over all $T - G$ bad configurations,
\[
    \Pr[\text{bad selected} \mid k = 0] \;\geq\; \frac{(T-G)\bar{p}}{1 + (T-G-1)\bar{p} + G\pstar}.
\]
Incorporating $\Pr[k = 0] = 1 - e^{-\theta_G}$ and setting the result at most $\beta$:
\[
    (1 - e^{-\theta_G}) \cdot \frac{(T-G)\bar{p}}{1 + (T-G-1)\bar{p} + G\pstar} \;\leq\; \beta.
\]
Rearranging: $(1-e^{-\theta_G})(T-G)\bar{p} \leq \beta(1 + (T-G-1)\bar{p} + G\pstar) = \beta + \beta(T-G)\bar{p} - \beta\bar{p} + \beta G\pstar$. Moving the $\beta(T-G)\bar{p}$ term to the left:
\[
    (1-e^{-\theta_G}-\beta)(T-G)\bar{p} \;\leq\; \beta(1 + G\pstar - \bar{p}) \;\leq\; \beta(1 + G\pstar).
\]
Since $\beta \leq (1 - e^{-\theta_G})/2$, the factor $1 - e^{-\theta_G} - \beta \geq (1-e^{-\theta_G})/2 > 0$, giving
\[
    \bar{p} \;\leq\; \frac{2\beta(1 + G\pstar)}{(T-G)(1-e^{-\theta_G})}.
\]
Since $1-e^{-\theta_G}$ is a positive constant depending only on the privacy parameters, $\bar{p} = O(\beta(1+G\pstar)/(T-G))$ for any fixed $\eps > 2\eps_1$.
\end{proof}

\section*{Acknowledgements}
    \noindent\begin{minipage}[t]{\dimexpr\linewidth-5.5cm\relax}%
    \vspace{0pt}%
    \indent This research was supported by the European Research Council (ERC) under the European Union's Horizon 2020 research and innovation programme (Grant agreement No.\ 101019564), and the Austrian Science Fund (FWF) under grant DOI 10.55776/Z422.
    \end{minipage}
    \begin{minipage}[t]{5cm}%
    \vspace{0pt}%
    \includegraphics[width=5cm]{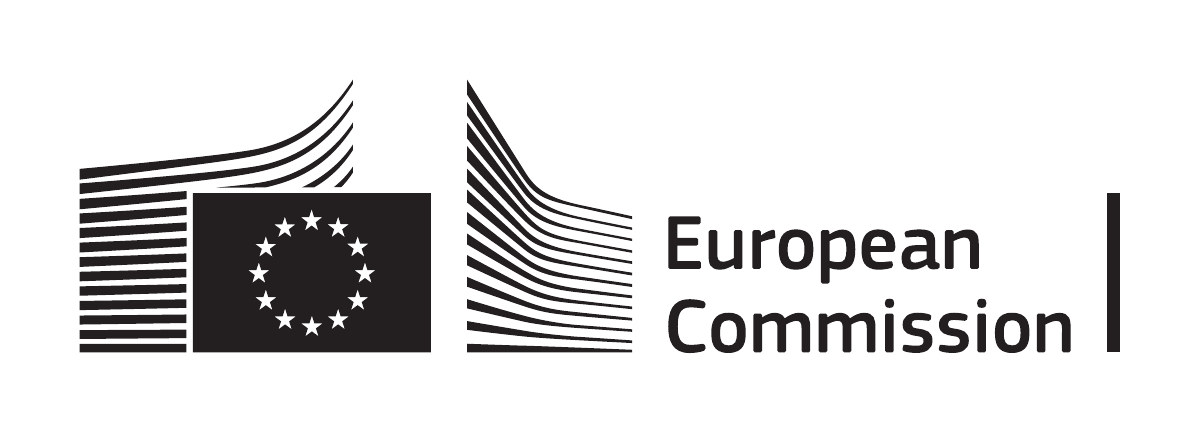}%
    \end{minipage}

  \medskip
  \noindent For open access purposes, the authors have applied a CC BY public copyright license to any author-accepted manuscript version arising from this submission. Views and opinions expressed are however those of the authors only and do not necessarily reflect those of the European Union or the European Research Council Executive Agency. Neither the European Union nor the granting authority can be held responsible for them.

  \medskip
  \noindent Anamay Chaturvedi was supported by an ISTA-Fellowship.

  \medskip
  \noindent Jalaj Upadhyay was supported by NSF CNS 2433628, Google Seed Fund grant, Google Research Scholar Award, Dean Research Seed Fund, and Decanal Research Grant. A part of the work was done while Jalaj was visiting ISTA.
  
\bibliographystyle{alpha}
\bibliography{biblio.bib}
\appendix

\section{Concentration Inequalities}

\begin{lem}[Poisson Chernoff bounds]\label{lem:PoissonChernoffBounds}
Let $M \sim \Po(\lambda)$ with $\lambda > 0$. For any $u \ge 0$, the following tail bounds hold:
\begin{enumerate}
    \item \textbf{Upper tail:} $\Pr[M \ge \lambda + u] \le \exp\!\left(-\frac{u^2}{2(\lambda+u)}\right)$
    \item \textbf{Lower tail:} $\Pr[M \le \lambda - u] \le \exp\!\left(-\frac{u^2}{2\lambda}\right)$
\end{enumerate}
\end{lem}

\begin{proof}
Let $M \sim \Po(\lambda)$. We first derive the moment generating function (MGF). Using the probability mass function $\Pr[M = k] = e^{-\lambda}\frac{\lambda^k}{k!}$, for any $t \in \mathbb{R}$:
\[
\Ex[e^{tM}] = \sum_{k=0}^{\infty} e^{tk} e^{-\lambda} \frac{\lambda^k}{k!} = e^{-\lambda} \sum_{k=0}^{\infty} \frac{(\lambda e^t)^k}{k!} = e^{-\lambda} e^{\lambda e^t} = \exp\!\big(\lambda(e^t - 1)\big).
\]

\paragraph{Upper tail.} 
For any $t > 0$, we apply Markov's inequality to the monotonically increasing function $e^{tM}$:
\[
\Pr[M \ge \lambda + u] = \Pr[e^{tM} \ge e^{t(\lambda + u)}] \le \frac{\Ex[e^{tM}]}{e^{t(\lambda + u)}} = \exp\!\big(\lambda(e^t - 1) - t(\lambda + u)\big).
\]
To minimize the right-hand side, we differentiate the exponent with respect to $t$ and set it to zero, which yields $t = \ln(1 + u/\lambda)$. Note that $t > 0$. Substituting this optimal $t$ back into the exponent gives:
\[
\Pr[M \ge \lambda + u] \le \exp\!\left( \lambda \left( \frac{u}{\lambda} - \left(1 + \frac{u}{\lambda}\right)\ln\!\left(1 + \frac{u}{\lambda}\right) \right) \right).
\]
Let $x = u/\lambda \ge 0$, and define $h(x) = (1+x)\ln(1+x) - x$. We claim that $h(x) \ge \frac{x^2}{2(1+x)}$. 
Let $f(x) = h(x) - \frac{x^2}{2(1+x)}$. We have $f(0) = 0$. The first derivative is $f'(x) = \ln(1+x) - \frac{x^2+2x}{2(1+x)^2}$, with $f'(0) = 0$. The second derivative simplifies to $f''(x) = \frac{1}{1+x} - \frac{1}{(1+x)^3} = \frac{(1+x)^2 - 1}{(1+x)^3} = \frac{x^2+2x}{(1+x)^3}$. 
Since $f''(x) \ge 0$ for all $x \ge 0$, $f'(x)$ is non-decreasing, implying $f'(x) \ge 0$ for all $x \ge 0$, which in turn implies $f(x) \ge 0$ for all $x \ge 0$. 
Thus, the inequality $h(x) \ge \frac{x^2}{2(1+x)}$ holds for all $x \ge 0$. Substituting $x = u/\lambda$ into this bound yields:
\[
\frac{u}{\lambda} - \left(1 + \frac{u}{\lambda}\right)\ln\!\left(1 + \frac{u}{\lambda}\right) \le -\frac{(u/\lambda)^2}{2(1 + u/\lambda)} = -\frac{u^2}{2\lambda(\lambda+u)}.
\]
Multiplying by $\lambda$ and exponentiating gives the desired upper bound:
\[
\Pr[M \ge \lambda + u] \le \exp\!\left(-\frac{u^2}{2(\lambda+u)}\right).
\]

\paragraph{Lower tail.}
For the lower tail, let $0 \le u \le \lambda$. (For $u > \lambda$, the probability is exactly $0$ and the bound trivially holds). For any $t < 0$, we apply Markov's inequality, which yields:
\[
\Pr[M \le \lambda - u] = \Pr[e^{tM} \ge e^{t(\lambda - u)}] \le \frac{\Ex[e^{tM}]}{e^{t(\lambda - u)}} = \exp\!\big(\lambda(e^t - 1) - t(\lambda - u)\big).
\]
Minimizing this over $t < 0$ yields the optimal choice $t = \ln(1 - u/\lambda)$. Substituting this into the exponent gives:
\[
\Pr[M \le \lambda - u] \le \exp\!\left( \lambda \left( -\frac{u}{\lambda} - \left(1 - \frac{u}{\lambda}\right)\ln\!\left(1 - \frac{u}{\lambda}\right) \right) \right).
\]
Let $y = u/\lambda \in [0, 1)$, and define $g(y) = -y - (1-y)\ln(1-y)$. Using the Taylor series expansion for $\ln(1-y) = -\sum_{k=1}^\infty \frac{y^k}{k}$, we have:
\[
g(y) = -y + (1-y)\sum_{k=1}^\infty \frac{y^k}{k} = -y + \sum_{k=1}^\infty \frac{y^k}{k} - \sum_{k=1}^\infty \frac{y^{k+1}}{k} = -y + y + \sum_{k=2}^\infty y^k \left( \frac{1}{k} - \frac{1}{k-1} \right) = \sum_{k=2}^\infty \frac{-y^k}{k(k-1)}.
\]
Since every term in the series is strictly negative, we can bound the sum by taking just the first term ($k=2$), yielding $g(y) \le -\frac{y^2}{2}$. 
Substituting $y = u/\lambda$ gives:
\[
-\frac{u}{\lambda} - \left(1 - \frac{u}{\lambda}\right)\ln\!\left(1 - \frac{u}{\lambda}\right) \le -\frac{u^2}{2\lambda^2}.
\]
Multiplying by $\lambda$ and exponentiating gives the desired lower bound:
\[
\Pr[M \le \lambda - u] \le \exp\!\left(-\frac{u^2}{2\lambda}\right).
\]
This completes the proof.
\end{proof}

\begin{lem}[Poisson Chernoff bounds -- Multiplicative Form]\label{lem:PoissonChernoffMultiplicative}
Let $M \sim \Po(\lambda)$ with $\lambda > 0$. For any $c > 0$, the following tail bounds hold:
\begin{enumerate}
    \item \textbf{Upper tail:} If $c > \lambda$, $\Pr[M \ge c] \le \frac{e^{-\lambda}(e\lambda)^c}{c^c}$.
    \item \textbf{Lower tail:} If $c < \lambda$, $\Pr[M \le c] \le \frac{e^{-\lambda}(e\lambda)^c}{c^c}$.
\end{enumerate}
\end{lem}

\begin{proof}
We use the moment generating function of the Poisson distribution, which is $\E[e^{tM}] = \exp\!\big(\lambda(e^t - 1)\big)$ for any $t \in \mathbb{R}$.

\paragraph{Upper tail.}
Assume $c > \lambda$. By Markov's inequality, for any $t > 0$, we have:
\[
\Pr[M \ge c] = \Pr[e^{tM} \ge e^{tc}] \le \frac{\E[e^{tM}]}{e^{tc}} = \exp\!\big(\lambda(e^t - 1) - tc\big).
\]
We minimize the right-hand side by differentiating the exponent with respect to $t$ and setting it to zero:
\[
\lambda e^t - c = 0 \implies e^t = \frac{c}{\lambda}.
\]
Since we assumed $c > \lambda$, the optimal parameter is $t = \ln(c/\lambda) > 0$, making it valid for the right-tail Markov application. Substituting $e^t = c/\lambda$ and $t = \ln(c/\lambda)$ into the bound yields:
\[
\Pr[M \ge c] \le \exp\!\left( \lambda\left(\frac{c}{\lambda} - 1\right) - c \ln\left(\frac{c}{\lambda}\right) \right) = \exp\!\left( c - \lambda - c \ln\left(\frac{c}{\lambda}\right) \right).
\]
Using the properties of exponents, we can rewrite this exactly as:
\[
\Pr[M \ge c] \le e^{-\lambda} e^c \left(\frac{\lambda}{c}\right)^c = \frac{e^{-\lambda}(e\lambda)^c}{c^c}.
\]

\paragraph{Lower tail.}
Assume $c < \lambda$. For any $t < 0$, Markov's inequality yields:
\[
\Pr[M \le c] = \Pr[e^{tM} \ge e^{tc}] \le \frac{\E[e^{tM}]}{e^{tc}} = \exp\!\big(\lambda(e^t - 1) - tc\big).
\]
Minimizing this over $t < 0$ yields $t = \ln(c/\lambda)$ and optimal value $e^t = c/\lambda$. Because $c < \lambda$, the parameter $t = \ln(c/\lambda) < 0$, which is valid for the left-tail application. Substituting this $t$ into the bound yields the exact same algebraic sequence:
\[
\Pr[M \le c] \le \exp\!\left( c - \lambda - c \ln\left(\frac{c}{\lambda}\right) \right) = \frac{e^{-\lambda}(e\lambda)^c}{c^c}.
\]
This completes the proof.
\end{proof}

\end{document}